\documentclass[oneside,english]{amsart}

\usepackage[hyphens]{url}
\usepackage{hyperref}

\usepackage{cmap}

\usepackage[T1]{fontenc}
\usepackage[utf8]{inputenc}
\usepackage{varioref}
\usepackage{amsthm}
\usepackage{amssymb}

\usepackage[usenames,dvipsnames]{xcolor}

\usepackage{cmap}

\usepackage{mathtools}

\usepackage{enumitem}

\usepackage[backend=bibtex,style=numeric,url=false]{biblatex}
\bibliography{Refs27082016.bib}

\numberwithin{equation}{section} 
\numberwithin{figure}{section} 
\theoremstyle{plain}
\theoremstyle{plain}
\newtheorem{thm}{Theorem}
  \theoremstyle{plain}
  \newtheorem{lem}{Lemma}
  \theoremstyle{plain}
  \newtheorem{prop}{Proposition}
  \theoremstyle{remark}
  \newtheorem*{note*}{Note}
  \theoremstyle{remark}
  \newtheorem*{conclusion*}{Conclusion}
  \theoremstyle{remark}
  \newtheorem{note}{Remark}
 \theoremstyle{definition}
  
  \newtheorem{defc}{Definition}
  \theoremstyle{plain}
  \newtheorem{cor}{Corollary}

\makeatletter

\usepackage{amsthm}
\usepackage{amsfonts}
\usepackage{amscd}
\usepackage[cmtip,arrow]{xy}
\usepackage{pb-diagram,pb-xy}

\usepackage{tikz}
\usetikzlibrary{matrix,arrows}

\newcommand{\g}{\mathfrak{g}}

\newcommand{\kf}{\mathfrak{k}}
\newcommand{\hf}{\mathfrak{h}}

\newcommand{\Ad}{\text{Ad}}

\newcommand{\cI}{{\mathcal I}}

\newcommand{\cL}{{\mathcal L}}

\newcommand{\mR}{\mathbb{R}}

\usepackage[e]{esvect}

\makeatother

\makeatletter
\newcommand*\bigcdot{\mathpalette\bigcdot@{.5}}
\newcommand*\bigcdot@[2]{\mathbin{\vcenter{\hbox{\scalebox{#2}{$\m@th#1\bullet$}}}}}
\makeatother

\usepackage{babel}

\mathtoolsset{showonlyrefs=true}

\begin{document}

\title{Routh reduction and Cartan mechanics}

\author{S. Capriotti}
\address{Departamento de Matemática, UNS and CONICET\\
  Av. Alem 1253 2º piso, 8000 Bahía Blanca\\
  Buenos Aires, Argentina}
\email{santiago.capriotti@uns.edu.ar}

\thanks{This work has been supported by CONICET}

\keywords{Routh reduction, Poincaré-Cartan forms, Lepage-equivalent problems, integrable systems}

\subjclass[2010]{53D20,37J35,37J15,70H33}

\maketitle

\begin{abstract}
In the present work a Cartan mechanics version for Routh reduction is considered, as an intermediate step toward Routh reduction in field theory. Motivation for this generalization comes from an scheme for integrable systems \cite{Feher200258}, used for understanding the occurrence of Toda field theories in so called Hamiltonian reduction of WZNW field theories \cite{Feher:1992yx}. As a way to accomplish with this intermediate aim, this article also contains a formulation of the Lagrangian Adler-Kostant-Symes systems discussed in \cite{Feher200258} in terms of Routh reduction.
\end{abstract}

\tableofcontents

\section{Introduction}

In the present work we will be interested in finding some answers related with the following questions:
\begin{itemize}
\item Generalize Routh reduction, as described in \cite{2010IJGMM..07.1451L}, to the Cartan setting.
\item Find an invariant formulation for the equations of motion associated to Routh reduction of Hamilton-Pontryagin variational principles, complementary to the description for these kind of systems found in \cite{2015arXiv150901946G}.
\item Give a geometrical interpretation of the reduction considered in \cite{Feher200258}, in terms of Routh reduction.
\end{itemize}
At this respect, it can be seen as a continuation of \cite{2010IJGMM..07.1451L,2015arXiv150901946G}, which also deal with Routh reduction of mechanical systems and its equations of motion (see also \cite{zbMATH01639196,CrampinMestdagRouth}).

Nevertheless, the approach taken in this article uses a procedure called \emph{Lepage-equivalent problem}, as a mean to characterise Poincar\'e-Cartan version of Euler-Lagrange equations. A previous work dealing with Routh reduction of Cartan mechanics is \cite{doi:10.1142/S1402925111001180}; an important difference with this reference resides in the fact that we decided not to fix momentum variables in advance, therefore working with a kind of Hamilton-Pontryagin, or unified, variational problem. In this regard, our approach is similar with \cite{2015arXiv150901946G}, as we mentioned before.

In order to describe more precisely the setting underlying this article, let $Q$ be a manifold and let $L\in C^\infty\left(TQ\right)$ a Lagrangian function. Instead of working with the Pontryagin bundle $TQ\oplus T^*Q$, we work in a bundle of $1$-forms $W_L$ on $\mR\times TQ$, locally isomorphic to the Pontryagin bundle. These bundles were called \emph{classical Lepage-equivalent} of the variational problem associated to the data $\left(Q,L\right)$ in the pioneering work of Gotay \cite{GotayCartan}, and allow us to translate equations characterizing extremals of a variational problem, to Cartan-like equations of motion (see Theorem \ref{Thm:CharacterizationSolsLagSystem2} below). The use of these equations with suitable lifts to $W_L$ of vector fields on $\mR\times TQ$, yields to an invariant description of them, just as in \cite{2015arXiv150901946G}. Their basic idea is to take a (perhaps local) basis of vector fields on $Q$, and to lift it to the Pontryagin bundle $TQ\oplus T^*Q$; in particular, this method proves to be very useful when working with equations of motion in presence of symmetry and one is trying to avoid regularity issues. We were able to translate these constructions to our approach: Equations of motion for Cartan-like systems $W_L$ were thus written by means of lifts of vector fields on its base space $\mR\times TQ$.

Now, the setting for Routh reduction used throughout the paper was borrowed from \cite{2010IJGMM..07.1451L}: Given $\left(Q,L,F\right)$ a (general) Lagrangian system and a $G$-action on $Q$ such that $Q\rightarrow Q/G$ is a principal bundle, its solution curves in a momentum map level set are in a one to one correspondence with solution curves of the Lagrangian system $\left(Q,R_\mu,F+G_\mu\right)$ for some function $R_\mu$ (the \emph{Routhian}) and a gyroscopic force term $G_\mu$ determined by a connection in $Q\rightarrow Q/G$. The reduced space of Routh reduction is an intrinsically constrained system
\begin{equation}\label{eq:IntConstSystemReduced}
  \left(T\left(Q/G\right)\times Q/G_\mu\times\widetilde{\g}\rightarrow Q/G_\mu\times\widetilde{\g},\overline{R}_\mu,f+\sigma^\mu\right)
\end{equation}
obtained reducing this last Lagrangian system; thus, given a solution curve for $\left(Q,L,F\right)$, we take the associated solution curve for $\left(Q,R_\mu,F+G_\mu\right)$ and its reduction to system \eqref{eq:IntConstSystemReduced} is the reduction for the original curve.

Our approach to Routh reduction follows a similar path: We provide Cartan-like bundles for $\left(Q,L,F\right)$ and for system \eqref{eq:IntConstSystemReduced}; Corollary \ref{cor:WLmuAndWp1Rmu0Coincide} to Theorem \ref{thm:TheorOnRouthDecomp} links solution curves for $W_L^\mu$ (corresponding to system $\left(Q,L,F\right)$ in the traditional approach, but restricted to a momentum map level set) with solution curves for $W_{p_1^*\overline{R}_\mu}^0$ (corresponding to system \eqref{eq:IntConstSystemReduced} via Proposition \ref{prop:CurveOnWNL}).

Finally, let us briefly describe the structure of the article. Sections \ref{sec:new-setting-routh} and \ref{sec:geometry-lepage} are devoted to introductory matters: In the former, we review basic definitions for Routh reduction as found in the existing literature. The latter provides the reader with notions from Lepage-equivalent theory, used throughout the paper.

Lifting of vector fields to $W_L$, as defined in Section \ref{sec:equat-moti-quasi}, is an original contribution of the present work, and becomes a fundamental tool in writing the equations of motion. The same can be said for the contents of Section \ref{sec:lepage-equiv-system}: Although intrinsically constrained systems are not in the scope of classical Lepage-equivalent problems as it appeared in literature, a proposal for generalization is given in this section, and a theorem relating equations of motions is proved in this context.

Now, when Routh reduction is formulated in the language of intrinsically constrained systems, just reduction of the Lagrangian system defined by Routh Lagrangian is considered; it is then necessary to relate the equations of motion of the Routh Lagrangian system with the equations of motion associated to the original Lagrangian system. This is achieved in Sections \ref{sec:group-action-decomp-WL} and \ref{sec:RouthRedForMechSystem}, using a scheme similar to the one used in \cite{2010IJGMM..07.1451L}: First, a momentum map for classical Lepage-equivalent problems is defined, and then the equivalence between the set of equations is proved. In this last task, a fundamental r\^ole is played by a decomposition of the contact bundle; this decomposition is found to be a consequence of the chosen connection in the principal bundle $Q\rightarrow Q/G$. 

Equations of motion for system $\left(W_{p_1^*\overline{R}_\mu}^0,\lambda_{p_1^*\overline{R}_\mu}^0,\beta^\mu\right)$, where $\beta^\mu$ is the gyroscopic force term induced by the connection $\omega_Q$ and $\mu\in\g^*$ are explicitly constructed in Section \ref{sec:equat-moti-quas}.

An interesting example is discussed in the last section of the present article: A Lagrangian system for a class of integrable systems known as Adler-Kostant-Symes (AKS) systems. It was considered in \cite{Feher200258}, as a mean to understand reduction of WZNW theories \cite{Feher:1992yx} in a more controlled environment. The formulation of this example in terms of Routh reduction turns relevant the search of an equivalent procedure for field theories. On this regard, in this article we will adopt the following unified viewpoint: Every solution for a variational problem either from Mechanics or from field theory, can be regarded as a (perhaps local) section of a bundle $\pi:E\rightarrow M$. For example, every curve $\gamma:I\subset\mR\rightarrow Q$ can be considered as a local section $s:I\subset\mR\rightarrow\mR\times Q:t\mapsto\left(t,\gamma\left(t\right)\right)$ of the trivial bundle
\[
\text{pr}_1:\mR\times Q\rightarrow\mR.
\]
From this perspective, the identification $J^1\text{pr}_1\equiv\mR\times TQ$ given by
\[
j_{\left(t,q\right)}^1s\mapsto\left(t,T_qs\left(\partial/\partial t\right)\right),
\]
allow us to consider the Lagrangian $L$ as a function on $J^1\text{pr}_1$, and the variational problem of Mechanics becomes a field theory variational problem
\[
\delta\int_\mR\left(\text{pr}s\right)^*\left(Ldt\right)=0,
\]
where $\text{pr}s:\mR\rightarrow J^1\text{pr}_1\equiv\mR\times TQ$ is the \emph{prolongation of the section $s:\mR\rightarrow\mR\times Q$}, defined as the unique section of $\left(\text{pr}_1\right)_1:\mR\times TQ\rightarrow\mR$ which is integral for the contact structure. Therefore, formulation of Routh reduction given in the present article is well suited for its generalization to field theory, which will be carried out elsewhere.

\section{Lagrangian systems in Routh reduction}
\label{sec:new-setting-routh}

\subsection{Notation}
\label{sec:notation}

Some conventions regarding notation will be used throughout the article. Given a bundle $f:E\rightarrow M$, the symbol $\mathfrak{X}^{V\left(f\right)}\left(E\right)\subset\mathfrak{X}\left(E\right)$ will represent the set of vector fields on $E$ vertical respect to the map $f$.

Whenever a product manifold $X_1\times X_2$ is considered, the canonical projections onto its factors will be denoted by
\[
\text{pr}_i:X_1\times X_2\rightarrow X_i
\]
for $i=1,2$. For $X$ a manifold, we will indicate by
\[
\tau_X:TX\rightarrow X,\qquad\overline{\tau}_X:T^*X\rightarrow X
\]
the canonical projections of the tangent and cotangent bundles.

If $\left(q^i\right)$ are local coordinates on $X$, the induced coordinates on $TX$ will be generically indicated by $\left(q^i,v^i\right)$.

Moreover, when working with Lie groups $G$ and $G$-spaces $X$ such that $X/G$ is a manifold, we will indicate by $p_G^X:X\rightarrow X/G$ the quotient projection. For every $\xi\in\g$, where $\g$ is the Lie algebra of $G$, $\xi_Q\in\mathfrak{X}\left(X\right)$ will be the infinitesimal generator for the action of $G$ on $X$. On tangent and cotangent spaces of $G$-spaces, we will consider the lifted action. 

Similar conventions will be adopted when working with canonical forms: For every manifold $X$, $\lambda_X\in\Omega^1\left(T^*X\right)$ represents the canonical $1$-form
\[
\left.\lambda_X\right|_{\alpha_q}\left(V_{\alpha_q}\right):=\alpha_q\left(T_{\alpha_q}\overline{\tau_X}\left(V_{\alpha_q}\right)\right)
\]
for every $V_{\alpha_q}\in T_{\alpha_q}\left(T^*X\right)$. Sometimes we will commit an abuse of notation regarding this convention, and we will use this symbol in order to represent pullback of these canonical forms to subbundles of a cotangent bundle.

Given two bundles $q_i:E_i\rightarrow X,i=1,2$ on a manifold $X$, symbol $q_1^*E_2$ will indicate the pullback bundle on $E_1$, defined as
\[
q_1^*E_2:=\left\{\left(e_1,e_2\right)\in E_1\times E_2:q_1\left(e_1\right)=q_2\left(e_2\right)\right\}\subset E_1\times E_2.
\]
Canonical maps $\text{pr}_1:q_1^*E_2\rightarrow E_1$ and $\text{pr}_2:q_1^*E_2\rightarrow E_2$ will be induced by the projections onto the factors of the product bundle. Sometimes a more symmetric symbol $E_1\times_NE_2$ will be used for these spaces, or even $E_1\times E_2$ when no confusion is possible.

Vectors $Z\in T_{\left(e_1,e_2\right)}\left(E_1\times_N E_2\right)$ will be indicated by the symbol $Z=X_1+X_2$, where $X_i\in T_{e_i}E_i,i=1,2$ such that $T_{e_1}q_1\left(X_1\right)=T_{e_2}q_2\left(X_2\right)$; a particular case will be the vertical vectors of the bundle $E_1\times_NE_2\rightarrow X$, for which the symbols $V_1+0,0+V_2$, with $V_i\in Vq_i,i=1,2$ will be used.

\subsection{Lagrangian systems}

This introduction is mainly based in \cite{2010IJGMM..07.1451L}. Our aim is to provide some basic definitions regarding Lagrangian systems and symmetry.

\begin{defc}[Lagrangian systems]
  A \emph{Lagrangian system} is a triple $\left(Q,L,F\right)$ where $Q$ is a manifold, $L:TQ\rightarrow\mR$ is an smooth function and $F:TQ\rightarrow T^*Q$ is a $T^*Q$-valued $1$-form on $Q$. A curve $q:I:=\left[a,b\right]\rightarrow Q$ is \emph{critical} for the Lagrangian system $\left(Q,L,F\right)$ if and only if
  \[
  \delta\int_IL\left(\dot{q}\left(t\right)\right) d t=-\int_I\left<F\left(\dot{q}\left(t\right)\right),\delta q\left(t\right)\right> d t
  \]
  for arbitrary variations $\delta q:I\rightarrow q^*\left(TQ\right)$ with fixed endpoints.
\end{defc}
There exists another kind of Lagrangian-like systems which are important in Routh reduction.
\begin{defc}[Intrinsically constrained Lagrangian system]\label{Def:IntConstLag}
  An \emph{intrinsically constrained Lagrangian system} is a triple $\left(\pi:M\rightarrow N,L,F\right)$, with $L$ a function on $T_MN:=TN\times_NM$ and $F$ a $T^*M$-valued $1$-form on $M$. A curve $\gamma:I\rightarrow M$ is \emph{critical} for the intrinsically constrained system $\left(\pi:M\rightarrow N,L,F\right)$ if and only if it is critical for the Lagrangian system $\left(M,p_1^*L,F\right)$, where $p_1:TM\rightarrow T_MN$ is given by
  \[
  p_1\left(v_m\right):=\left(T_m\pi\left(v\right),m\right).
  \]
\end{defc}
An intrinsically constrained system can be regarded as a Lagrangian system whose Lagrangian function does not depend on the fiber coordinates of the vertical bundle $V\pi$.

\begin{defc}[Invariant Lagrangian system]\label{Def:InvLagSystem}
  Let $G$ be a Lie group acting on $Q$. The Lagrangian system $\left(Q,L,F\right)$ is \emph{$G$-invariant} if and only if $L$ is a $G$-invariant function and $F$ fulfills the following conditions:
  \begin{enumerate}
  \item $F$ is $G$-equivariant, and
  \item $\mathop{\text{Im}}{F}$ is in the annihilator of $\left\{\xi_Q:\xi\in\g\right\}$.
  \end{enumerate}
\end{defc}
As in the Hamiltonian side, there exists a momentum map associated to the $G$-action on $Q$.
\begin{defc}[Momentum map]
  The \emph{momentum map} $J_L:TQ\rightarrow\g^*$ associated to the $G$-action on the Lagrangian system $\left(Q,L,F\right)$ is the map
  \[
  J_L\left(v_q\right)\left(\xi\right):=\left.\frac{\text{d}}{\text{d}t}\right|_{t=0}\left[L\left(v_q+t\xi_Q\left(q\right)\right)\right]
  \]
  for all $\xi\in\g$.
\end{defc}
As usual, it provides us with conserved quantities when working with $G$-invariant Lagrangian systems; nevertheless, a more general situation is possible.
\begin{prop}
  Let $\left(Q,L,F\right)$ be a Lagrangian system such that
  \[
  \left< d L,\xi_{TQ}\right>=-\left<F,\xi_Q\right>
  \]
  for all $\xi\in\g$ on the critical curves. Then $J_L$ is a conserved quantity.
\end{prop}

\section{Geometry of Lepage-equivalent problems}
\label{sec:geometry-lepage}

\subsection{Definitions}
\label{sec:definitions}

The scheme we will develop in the present article requires the notion of \emph{classical Lepage-equivalent variational problems} \cite{GotayCartan,2013arXiv1309.4080C,Krupka1986a,Krupka1986b,KrupkaLagrangeanStructures}, as a setting that, in particular, is suitable for translation into classical field theory \cite{GoldSternberg}. In this realm, we work with sections of the bundle $\text{pr}_1:\mR\times TQ\rightarrow\mR:\left(t,v_q\right)\mapsto t$ instead of working with curves in $TQ$; it is clear that there exists a one to one correspondence between these descriptions, and it is quite straightforward how to change between viewpoints.

Let us consider how a Lagrangian system $\left(Q,L,0\right)$ determines the dynamics in this setting. The main idea is to consider the differential ideal $\cI_{\text{con}}$ in $\Omega^\bullet\left(\mR\times TQ\right)$ generated by the forms $\theta^i:=dq^i-v^idt$; sections $\gamma:I\subset\mR\rightarrow\mR\times TQ$ that correspond to curves in $TQ$ coming from derivatives of curves in $Q$ are represented by \emph{integral sections of $\cI_{\text{con}}$}, namely, such that
\[
\gamma^*\theta^i=0
\]
for all $i$ \cite{book:852048}. A crucial fact about this ideal is that it can be generated by sections of a bundle $I_{\text{con}}\subset\wedge^\bullet\left(\mR\times TQ\right)$; essentially, this bundle is the vector subbundle generated by the set of forms $\left\{\theta^i\right\}$. So instead of working on $TQ$ and perform variations on curves in $TQ$ which come from curves in $Q$, we perform arbitrary variations of curves in a bundle $W_L\rightarrow\mR\times TQ$, which incorporate (via Lagrange multipliers acting on sections of $I_{\text{con}}$) restrictions forcing curves in $TQ$ to be time derivatives of curves in $Q$.

In detail, bundle $I_{\text{con}}\rightarrow\mR\times TQ$ will be called \emph{contact bundle}, and is defined fiberwise as follows.

\begin{defc}\label{Def:AnotherDescriptionIcon}
  The contact subbundle $I_{\text{con}}$ on $\mR\times TQ$ is the subbundle of $T^*\left(\mR\times TQ\right)$ with fiber
  \begin{equation}
    \left.I_{\text{con}}\right|_{\left(t,v_q\right)}:=\left\{\alpha\circ T_{v_q}\tau_Q-\alpha\left(v_q\right) d t:\alpha\in T^*_qQ\right\}\subset T_{\left(t,v_q\right)}^*\left(\mR\times TQ\right).\label{Eq:ContactSubbundleOnTQ}
  \end{equation}
  Forms whose images lie in $I_{\text{con}}$ will be called \emph{contact forms}. 
\end{defc}

The subbundle $W_L\rightarrow\mR\times TQ$ fits in the diagram
\[
\begin{diagram}
  \node{W_L}\arrow{se,b}{\pi_L}\arrow[2]{e,t,J}{i_L}\node[2]{T^*\left(\mR\times TQ\right)}\arrow{sw,b}{\overline{\tau}_{\mR\times TQ}}\\
  \node[2]{\mR\times TQ}
\end{diagram}
\]
and consists essentially of the affine subbundle obtained from $I_{\text{con}}$ by translation along the Lagrangian $1$-form $Ldt$.

The underlying set of this bundle is determined fiberwise by the formula
\begin{equation}
\left.W_L\right|_{\left(t,v\right)}:=\left(L\left(t,v\right)dt+\left.I_{\text{con}}\right|_{\left(t,v\right)}\right)\cap\left(T^*\left(\mR\times TQ\right)\right)^V,\label{eq:WLDefinition}
\end{equation}
where
\[
\left(T^*\left(\mR\times TQ\right)\right)^V:=T^*\left(\mR\times TQ\right)\cap\left(V\left(\text{id}\times\tau_Q\right)\right)^0
\]
is the portion of the cotangent bundle of $\mR\times TQ$ annihilating those vectors which are vertical respect to the projection
\[
\text{id}\times\tau_Q:\mR\times TQ\longrightarrow\mR\times Q.
\]

\begin{note}\label{W_LIdentification}
  In local coordinates $\left(t,q^i,v^i\right)$ this subbundle can be described as
  \[
  \left.W_L\right|_{\left(t,q^i,v^i\right)}=\left\{L\left(t,q^i,v^i\right)dt+p_i\left(dq^i-v^idt\right):p_i\in\mR\right\}.
  \]
  Thus, we have the identification
  \begin{equation}
    W_L\simeq\mR\times\left(TQ\oplus T^*Q\right)\label{eq:IdentLagPont}.
  \end{equation}
  This identification can be seen directly from the local expression for $W_L$, or more intrinsically via Equation~\eqref{Eq:ContactSubbundleOnTQ}, namely, taking into account that $\rho\in\left.W_L\right|_{\left(t,v_q\right)}$ corresponds to $\left(t,w_{q'},\alpha\right)$ if and only if $q=q',w_{q}=v_q$ and
  \[
  \rho=L\left(t,v_q\right)dt+\alpha\circ T_{v_q}\tau_Q-\alpha\left(v_q\right)dt.
  \]
\end{note}

The immersion $W_L\subset T^*\left(\mR\times TQ\right)$ provides it with a canonical $1$-form $\lambda_L$, namely the pullback of the canonical $1$-form $\lambda_{\mR\times TQ}\in\Omega^1\left(T^*\left(\mR\times TQ\right)\right)$ to $W_L$,
\[
\lambda_L:=i_L^*\left(\lambda_{\mR\times TQ}\right)\in\Omega^1\left(W_L\right).
\]
This form will be what we will call \emph{Cartan form} in this context; a reason for this terminology can be found below (Proposition~\ref{prop:CorrespondenceCartan}).

\subsection{Lepage-equivalent problems and Cartan form mechanics}
\label{sec:lepage-equiv-probl-1}

The purpose of the present section is to formulate equations of motion in the realm of Lepage-equivalent problems. In order to proceed, we will provide a definition for solution curves associated to the data $\left(W_L,\lambda_L\right)$, proving that these curves coincide with extremals of Lagrangian system $\left(Q,L,0\right)$; more details on this correspondence can be found in \cite{book:852048,zbMATH01933856,hsu92:_calcul_variat_griff}. Thus, equations of motion in Cartan form mechanics \cite{Krupka20121154} can be recovered from this setting by identifying a subbundle $F_L\subset W_L$ containing every solution curve, which is essentially the graph of Legendre tranformation for $L$; it can be interpreted saying that Lepage-equivalent formalism have Legendre transformation built into it.

\begin{defc}\label{Def:SolutionCurve}
  A curve $\gamma:I\subset\mR\rightarrow Q$ is \emph{a
    solution curve for the data $\left(W_L,\lambda_L\right)$} if and
  only if there exists a curve $\Gamma:I\rightarrow W_L$ such that
  \begin{enumerate}
  \item
    $\displaystyle\tau_Q\circ\text{pr}_2\circ\pi_L\circ\Gamma=\gamma$,
  \item $\displaystyle\text{pr}_1\circ\pi_L\circ\Gamma=\text{id}_\mR$,
    and
  \item $\displaystyle\Gamma^*\left(X\lrcorner d\lambda_L\right)=0$
    for all $X\in\mathfrak{X}\left(W_L\right)$. \label{item:EXtremals3}
  \end{enumerate}
\end{defc}

\begin{note}
  Equation \eqref{item:EXtremals3} tells us that lifted curves $\Gamma:I\rightarrow W_L$ are extremals of the variational problem (under unrestricted variations with fixed ends) associated to the functional
  \[
  \Gamma\mapsto\int_I\Gamma^*\left(\lambda_L\right).
  \]
\end{note}

Maps in Definition \ref{Def:SolutionCurve} are shown in the following diagram.
\begin{center}
  \begin{tikzpicture}
    \matrix (m) [matrix of math nodes, row sep=3em,
    column sep=3em, text height=1.5ex, text depth=0.25ex]
    { W_L & \mR\times TQ & TQ & Q  \\
      & \mR &  & \\ };
    \path[>=latex,->]
    (m-2-2) edge [bend left=15] node[below] {$ \Gamma $} (m-1-1)
    edge [bend right=15] node[above] {$ \dot{\gamma} $} (m-1-3)
    edge [bend right=25] node[below] {$ \gamma $} (m-1-4)
    (m-1-2) edge node[left] {$ \text{pr}_1 $} (m-2-2)
    (m-1-1) edge node[above] {$ \pi_L $} (m-1-2)
    (m-1-2) edge node[above] {$ \text{pr}_2 $} (m-1-3)
    (m-1-3) edge node[above] {$ \tau_Q $} (m-1-4);
  \end{tikzpicture}
\end{center}

Then, as promised, we have the following correspondence with extremal curves for a Lagrangian system.
\begin{thm}
  $\gamma:I\rightarrow Q$ is a solution curve for the data $\left(W_L,\lambda_L\right)$ if and only if it is an extremal for the Lagrangian system $\left(Q,L,0\right)$.
\end{thm}
\begin{proof}
  Let us introduce the local coordinates $\left(t,q^i,v^i,p_i\right)$ on $W_L$ induced by the identification~\eqref{eq:IdentLagPont}. Then
  \[
  \lambda_L=Ldt+p_i\left(dq^i-v^idt\right)
  \]
  and we will have that
  \begin{gather*}
    \Gamma^*\left(\frac{\partial}{\partial q^i}\lrcorner d\lambda_L\right)=\Gamma^*\left(\frac{\partial L}{\partial q^i}dt-dp_i\right)\\
    \Gamma^*\left(\frac{\partial}{\partial v^i}\lrcorner d\lambda_L\right)=\Gamma^*\left(\frac{\partial L}{\partial v^i}dt-p_idt\right)\\
    \Gamma^*\left(\frac{\partial}{\partial p_i}\lrcorner d\lambda_L\right)=\Gamma^*\left(-v^idt+dq^i\right).
    \end{gather*}
    Then if $\gamma\left(t\right)=\left(q^i\left(t\right)\right)$ and $\Gamma\left(t\right)=\left(t,q^i\left(t\right),v^i\left(t\right),p_i\left(t\right)\right)$, the result follows.
\end{proof}

Thus, equations of motion in Cartan form mechanics \cite{Krupka20121154} can be recovered as follows: There exists a subbundle $F_L\subset W_L$ defined through
\[
F_L:=\left\{\alpha\in W_L:\frac{\partial}{\partial t}\lrcorner Z\lrcorner\left.d\lambda_L\right|_\alpha=0\quad\text{for all}Z\in V\left(\text{id}\times\tau_Q\right)\right\}.
\]

It projects onto $\mR\times TQ$ via the restriction $\pi_F:=\left.\pi_L\right|_{F_L}:F_L\rightarrow\mR\times TQ$. This subbundle fits in the following diagram
\[
\begin{diagram}
  \node{F_L}\arrow{se,b}{\pi_F}\arrow{e,t,J}{j_L}\node{W_L}\arrow{s,b}{\pi_L}\arrow{e,t,J}{i_L}\node{T^*\left(\mR\times TQ\right)}\arrow{sw,b}{\overline{\tau}_{\mR\times TQ}}\\
  \node[2]{\mR\times TQ}
\end{diagram}
\]

Locally we have that $\alpha\in F_L$ if and only if
\begin{equation}\label{Eq:LocalFL}
  \alpha=Ldt+\frac{\partial L}{\partial v^i}\left(dq^i-v^idt\right).
\end{equation}

\begin{lem}\label{Lem:CanonicalSection}
  $\pi_F$ is injective. Moreover, there exists a section $s_0:\mR\times TQ\rightarrow W_L$ such that $F_L=\text{Im}\,s_0$.
\end{lem}
\begin{proof}
  Let $E:TQ\rightarrow\mR$ be the energy function associated to $L$
  \cite{A-M} and
  \[
  \theta_L:=\left(\text{pr}_2\circ\mathbb{F}L\right)^*\lambda_Q\in\Omega^1\left(\mR\times
    TQ\right)
  \]
  the pullback of the canonical $1$-form on $T^*Q$ to $\mR\times TQ$. Then
  \begin{equation}
    s_0\left(t,v\right):=-E\left(t,v\right)dt+\left.\theta_L\right|_{\left(t,v\right)}\in W_L;\label{eq:CanonicSectionDef}
  \end{equation}
  by the local expression~\eqref{Eq:LocalFL}, it results that $F_L=\text{Im}\,s_0$.
\end{proof}

Moreover, this submanifold allows us to establish a correspondence between canonical forms defined above and the classical forms.
\begin{prop}\label{prop:CorrespondenceCartan}
  The form $j_L^*\left(\lambda_L\right)$ coincides with the classical Cartan form under identification~\eqref{eq:IdentLagPont}.
\end{prop}
It explains our choice of name for the form $\lambda_L$.

Finally, the section $s_0$ can be used for construct the solutions of $\left(W_L,\lambda_L\right)$ whenever extremals of $\left(Q,L,0\right)$ are known.
\begin{prop}\label{Prop:SolsWLCharacterization}
  $\Gamma$ is a solution for $\left(W_L,\lambda_L\right)$ if and only if
  \[
  \Gamma\left(t\right):=s_0\left(t,\dot{\gamma}\left(t\right)\right),\qquad t\in I\subset\mR
  \]
  for $\gamma:I\rightarrow Q$ an extremal for $\left(Q,L,0\right)$.
\end{prop}

\subsection{General Lagrangian systems}
\label{sec:gener-lagr-syst}

Let us consider ``Cartan-like'' equations of motion for general Lagrangian systems $\left(Q,L,F\right)$, as defined in \cite{eprints21388}. The pair $\left(W_L,\lambda_L\right)$ is determined as before; additionally, we define the $1$-form $\widetilde{F}\in\Omega^1\left(W_L\right)$ such that
\begin{equation}\label{Eq:FFormDefinition}
  \left.\widetilde{F}\right|_\alpha\left(V\right):=\left<F\left(\left(\text{pr}_2\circ\pi_L\right)\left(\alpha\right)\right),T_\alpha\left(\tau_Q\circ\text{pr}_2\circ\pi_L\right)\left(V\right)\right>
\end{equation}
for all $V\in T_\alpha W_L$. In terms of the coordinates $\left(t,q^i,v^i,p_i\right)$ for $W_L$, we have
\[
\begin{diagram}
  \node{W_L}\arrow{e,t}{\pi_L}\node{\mR\times TQ}\arrow{e,t}{\text{pr}_2}\node{TQ}\\
  \node{\left(t,q^i,v^i,p_i\right)}\arrow{e,t,T}{}\node{\left(t,q^i,v^i\right)}\arrow{e,t,T}{}\node{\left(q^i,v^i\right)}
\end{diagram}
\]
and writing
\[
F=\alpha_idq^i
\]
for the force term, with $\alpha_i$ functions locally defined on $TQ$, we will obtain
\[
\widetilde{F}=\alpha_idq^i.
\]

So let us define the notion of solution curve for data $\left(W_L,\lambda_L,F\right)$; as expected, we will see below (Theorem \ref{Thm:CharacterizationSolsLagSystem2}) that these kind of curves produce solutions for the original Lagrangian system $\left(Q,L,F\right)$ and viceversa.

\begin{defc}\label{Def:SolCurveGenData}
  A curve $\gamma:I\subset\mR\rightarrow Q$ is a \emph{solution curve for the data $\left(W_L,\lambda_L,F\right)$} if and only if there exists a curve $\Gamma:I\rightarrow W_L$ such that
\begin{enumerate}
\item\label{Item:UnoGamma} $\displaystyle\tau_Q\circ\text{pr}_2\circ\pi_L\circ\Gamma=\gamma$,
\item\label{Item:DosGamma} $\displaystyle\text{pr}_1\circ\pi_L\circ\Gamma=\text{id}_\mR$, and
\item\label{Item:TresGamma} $\displaystyle\Gamma^*\left(X\lrcorner\left(d\lambda_L+\widetilde{F}\wedge dt\right)\right)=0$ for all $X\in\mathfrak{X}^{V\left(\text{pr}_1\circ\pi_L\right)}\left(W_L\right)$.
\end{enumerate}
\end{defc}

\begin{note}
  A couple of remarks on this definition:
  \begin{itemize}
  \item In local terms, the first two requeriments of the previous
    definition mean that
    $\gamma:t\mapsto\left(q^i\left(t\right)\right)$ and
    $\Gamma:t\mapsto\left(s\left(t\right),\widetilde{q}^i\left(t\right),\widetilde{v}^i\left(t\right),\widetilde{p}_i\left(t\right)\right)$
    are related by the equations
    \[
    q^i\left(t\right)=\widetilde{q}^i\left(t\right),\qquad
    s\left(t\right)=t
    \]
    for all $t$.
  \item It is enough to verify the last item on a set of (perhaps local) generators for $$\mathfrak{X}^{V\left(\text{pr}_1\circ\pi_L\right)}\left(W_L\right).$$ This fact will be exploited more deeply in Section \ref{sec:equat-moti-quasi} below.
  \end{itemize}

\end{note}
The last item can be rewritten as soon as $F$ is a $2$-form on $Q$.
\begin{lem}
  Let $F$ be a $2$-form on $Q$ and $\Gamma:I\rightarrow W_L$ a curve satisfying items \ref{Item:UnoGamma} and \ref{Item:DosGamma} in Definition~\ref{Def:SolCurveGenData}. Then
\[
\Gamma^*\left(X\lrcorner\left(d\lambda_L+\widetilde{F^\flat}\wedge dt\right)\right)=0
\]
for all $X\in\mathfrak{X}^V\left(W_L\right)$, is equivalent to
\[
\Gamma^*\left(X\lrcorner\left(d\lambda_L+\left(\tau_Q\circ\text{pr}_2\circ\pi_L\right)^*{F}\right)\right)=0
\]
for all $X\in\mathfrak{X}^V\left(W_L\right)$.
\end{lem}
\begin{proof}
  For the underlying map $F^\flat:TQ\rightarrow T^*Q$ we construct the $1$-form $\widetilde{F^\flat}\in\Omega^1\left(W_L\right)$. For every $X\in\mathfrak{X}^V\left(W_L\right)$ we have that
  \[
  X\lrcorner\left(\widetilde{F^\flat}\wedge dt\right)=\left(X\lrcorner\widetilde{F^\flat}\right)dt.
  \]
  On the other hand, if $F=f_{ij}dq^i\wedge dq^j$ in local coordinates, we will obtain that
  \[
  \frac{\partial}{\partial v^i}\lrcorner\widetilde{F^\flat}=\frac{\partial}{\partial p_i}\lrcorner\widetilde{F^\flat}=0
  \]
  and
  \[
  \frac{\partial}{\partial q^i}\lrcorner\widetilde{F^\flat}=f_{ji}v^j.
  \]
  Then
  \begin{align*}
    \Gamma^*\left(\frac{\partial}{\partial q^i}\lrcorner\widetilde{F^\flat}\wedge dt\right)&=\Gamma^*\left(f_{ji}v^jdt\right)\\
    &=\Gamma^*\left(f_{ji}dq^j\right)\\
    &=\Gamma^*\left(\frac{\partial}{\partial q^i}\lrcorner\left(\tau_Q\circ\text{pr}_2\circ\pi_L\right)^*F\right)
  \end{align*}
  because in local coordinates, the condition
  \[
  \Gamma^*\left(\frac{\partial}{\partial p_i}\lrcorner\left(d\lambda_L+\left(\tau_Q\circ\text{pr}_2\circ\pi_L\right)^*{F}\right)\right)=0
  \]
  implies $\Gamma^*\left(dq^i-v^idt\right)=0$.
\end{proof}

Then we have the following correspondence between extremals of a general Lagrangian system and solution curves of a triple $\left(W_L,\lambda_L,F\right)$.
\begin{thm}\label{Thm:CharacterizationSolsLagSystem2}
  $\gamma$ is a solution curve for the data $\left(W_L,\lambda_L,F\right)$ if and only if it is an extremal for the Lagrangian system $\left(Q,L,F\right)$.
\end{thm}

\section{Equations of motion in quasi-velocities and quasi-momenta}
\label{sec:equat-moti-quasi}

Let us deduce the implicit equations of motion obtained in \cite{2015arXiv150901946G}, using the formalism developed above. It makes necessary to find a way to lift vector fields on $\mR\times TQ$ to the bundle of forms $W_L$. The first part of this section is devoted to this task.

Later, a characterization for these equations as a set of forms on $W_L$ is found (see Propositions \ref{prop:QuasiELequations} and \ref{prop:QuasiELequationsGeneral} below). Thus a curve is a solution for the Lagrangian system if its tangent vector field belongs to the annihilator of this set of forms. This characterization is useful because of the pullback naturality of forms: When formulated in these terms, equations of motion can be pulled back along maps. A similar viewpoint for working with reduction of differential equations can be found in \cite{1751-8121-45-6-065202}.

\subsection{Infinitesimal symmetries and lifting}
\label{sec:cont-manif-lift}

We want to find a way to lift vector fields from $\mR\times TQ$ to the bundle $W_L$. In the present section we will carry out this task by means of the notion of \emph{infinitesimal symmetry} of the contact structure $\lambda_L$.

\subsubsection{The lift to $W_L$}
\label{sec:LiftToWL}

Let us consider now the lift of vector fields on $\mR\times TQ$ to $W_L$. Recall that associated to the adapted coordinates $\left(t,q^i,v^i\right)$ on $\mR\times TQ$, there exist the coordinates $\left(t,q^i,v^i,p_i\right)$ on $W_L$.
\begin{defc}\label{GeneralLift1}
  A \emph{lift} for a vector field $Z\in\mathfrak{X}\left(\mR\times TQ\right)$ is a vector field $Z^{1_L}\in\mathfrak{X}\left(W_L\right)$ such that
  \begin{itemize}
  \item the map $\pi_L:W_L\rightarrow\mR\times TQ$ projects $Z^{1_L}$ onto $Z$, and
  \item $Z^{1_L}$ is an infinitesimal symmetry for $\lambda_L$, namely
    \[
    \mathcal{L}_{Z^{1_L}}\lambda_L=\mu_Z\lambda_L
    \]
    for some $\mu_Z\in C^\infty\left(W_L\right)$.
  \end{itemize}
\end{defc}

\begin{thm}\label{thm:LiftsExists}
  Let $L\in C^\infty\left(\mR\times TQ\right)$ be a Lagrangian such that $L\left(t,v_q\right)\not=0$ for all $\left(t,v_q\right)\in\mR\times TQ$. Then for every $Z\in\mathfrak{X}\left(\mR\times TQ\right)$ which is projectable along the map $\text{id}\times{\tau_Q}:\mR\times TQ\rightarrow\mR\times Q$, there exists a lift $Z^{1_L}$.
\end{thm}
\begin{proof}
  Let us consider a general vector field
  \[
  Z=U\frac{\partial}{\partial t}+Z^i\frac{\partial}{\partial q^i}+W^i\frac{\partial}{\partial v^i};
  \]
  its lift must read
  \[
  Z^{1_L}=U\frac{\partial}{\partial t}+Z^i\frac{\partial}{\partial q^i}+W^i\frac{\partial}{\partial v^i}+R_i\frac{\partial}{\partial p_i}.
  \]
  The canonical form in these coordinates is
  \[
  \lambda_L=\left(L\left(t,q,v\right)-p_iv^i\right)dt+p_idq^i
  \]
  and so
  \[
  d\lambda_L=\frac{\partial L}{\partial q^i}dq^i\wedge dt+\left(\frac{\partial L}{\partial v^i}-p_i\right)dv^i\wedge dt-v^idp_i\wedge dt+dp_i\wedge dq^i.
  \]
  Let us define $E:=L-p_iv^i$. Second condition in Definition \ref{GeneralLift1} translates into
  \begin{align*}
    \mu_Z E&=Z^i\frac{\partial L}{\partial q^i}+W^i\left(\frac{\partial L}{\partial v^i}-p_i\right)+E\frac{\partial U}{\partial t}+p_k\frac{\partial Z^k}{\partial t}-R_iv^i,\\
    0&=E\frac{\partial U}{\partial p_i}+p_k\frac{\partial Z^k}{\partial p_i},\\
    \mu_Z p_i&=R_i+E\frac{\partial U}{\partial q^i}+p_k\frac{\partial Z^k}{\partial q^i},\\
    0&=E\frac{\partial U}{\partial v^i}+p_k\frac{\partial Z^k}{\partial v^i}.
  \end{align*}
  The second equation is automatically fulfilled, because neither $U$ nor $Z^k$ depend on the fiber coordinates $p_i$. The same happens with the fourth equation, because of the projectability assumption. From the third we have that
  \[
  R_i=\mu_Z p_i-E\frac{\partial U}{\partial q^i}-p_k\frac{\partial Z^k}{\partial q^i},
  \]
  and replacing it in the first equation
  \[
  \mu_Z L=Z^i\frac{\partial L}{\partial q^i}+W^i\left(\frac{\partial L}{\partial v^i}-p_i\right)+ED_tU+p_kD_tZ^k.
  \]
  This equation determines $\mu_Z$ because $L\not=0$.
\end{proof}

\begin{note}
  Condition $L\not=0$ can be overcome by using a new Lagrangian function $L_1:=L+1$. These pair of equivalent Lagrangians $L,L_1$ give us a pair of lifts, defined on a pair of open sets covering $W_L$; as far as equations of motion depend ultimately on derivatives of $L$, any lift yields to the same equations in their common domain, so no ambiguity regarding the equations of motion remains.
\end{note}

\begin{note}\label{rem:DifferentLifts}
  Given a vector field $Z\in\mathfrak{X}\left(M\right)$, we can devise
  another lift to $T^*M$ using any of the following equivalent
  definitions:
  \begin{itemize}
  \item Use $Z$ to define the linear function $\overline{Z}\in
    C^\infty\left(T^*M\right)$; the lift
    $Z^{1*}\in\mathfrak{X}\left(T^*M\right)$ is then the Hamiltonian
    vector field associated to this function.
  \item Take the flow $\Phi^Z_t:M\rightarrow
    M,t\in\left(-\epsilon,\epsilon\right)$ and pull it back to $T^*M$;
    it gives rise to a flow
    \[
    \left(\Phi^Z_t\right)^*:T^*M\rightarrow T^*M
    \]
    and the lift $Z^{1*}$ is the corresponding vector field.
  \end{itemize}

  For a general $Z\in\mathfrak{X}\left(\mR\times TQ\right)$, there is
  no guarantee that these constructions yield to vector fields tangent
  to $W_L$; this is the main reason for the definition of lift adopted
  in the presente work. Nevertheless, when $Z$ comes from an
  infinitesimal symmetry for $L$, these definitions agree, as will be
  shown later (see Proposition \ref{prop:AgreeLiftDefinitions}.)
\end{note}

\subsubsection{A local basis of vector fields on $W_L$}
\label{sec:local-basis-vector}

Given $X\in\mathfrak{X}^{V\left(\text{pr}_1\right)}\left(\mR\times Q\right)$, we can consider the canonical \emph{vertical lift} $X^V\in\mathfrak{X}^{V\left(\text{pr}_1\right)}\left(\mR\times TQ\right)$ and \emph{complete lift} $X^C\in\mathfrak{X}^{V\left(\text{pr}_1\right)}\left(\mR\times TQ\right)$. In local coordinates $\left(t,q^i,v^i\right)$, if
\[
X=X^i\frac{\partial}{\partial q^i}
\]
we have that \cite{CrampinApplicable}
\begin{align*}
  &X^V=X^i\frac{\partial}{\partial v^i}\\
  &X^C=X^i\frac{\partial}{\partial q^i}+v^k\frac{\partial X^j}{\partial q^k}\frac{\partial}{\partial v^j}.
\end{align*}

These vector fields have the following brackets
\begin{equation*}
  \left[X^V,Y^V\right]=0,\qquad\left[X^V,Y^C\right]=\left[X,Y\right]^V,\qquad\left[X^C,Y^C\right]=\left[X,Y\right]^C.
\end{equation*}

Additionally, for every $\sigma\in\Gamma\left(I^1_{\text{con}}\right)$, we can use the affine structure of $W_L$ in order to define another vector field $Z_\sigma\in\mathfrak{X}^{V\left(\pi_L\right)}\left(W_L\right)$ such that
\begin{equation}\label{Eq:ZSigmaDef}
  Z_\sigma\left(\left.\rho\right|_{\left(t,v_q\right)}\right):=\left.\frac{\vec{\text{d}}}{\text{d}s}\right|_{s=0}\left[\left.\rho\right|_{\left(t,v_q\right)}+s\sigma\left(t,v_q\right)\right].
\end{equation}
These vector fields have the following property regarding the canonical form $\lambda_L$.
\begin{prop}\label{prop:LiftSigma}
  Let $\sigma\in\Gamma\left(W_L\right)$ be a section of the affine bundle $\pi_L:W_L\rightarrow\mR\times TQ$. Then $Z_\sigma\lrcorner\lambda_L\equiv0$ and
  \[
  \cL_{Z_\sigma}\lambda_L=-\pi_L^*\sigma.
  \]
\end{prop}
\begin{proof}
  The first property is a consequence of the identity
  \[
  T\pi_L\left(Z_\sigma\right)=0.
  \]
  The flow for $Z_\sigma$ is given by
  \[
  \Phi^\sigma_s:\rho_{\left(t,v_q\right)}\mapsto\rho_{\left(t,v_q\right)}+s\sigma\left(t,v_q\right),
  \]
  for every $s\in\mR$. Then
  \[
  T\Phi^\sigma_s:V_\rho\mapsto V_\rho+s\cdot\left(T\sigma\circ T\overline{\tau}_{\mR\times TQ}\right)\left(V_\rho\right),
  \]
  and so
  \begin{align*}
    \left(\Phi^\sigma_s\right)^*\left(i_L^*\lambda_{\mR\times TQ}\right)&=\left[\text{id}+s\cdot\left(T\sigma\circ T\overline{\tau}_{\mR\times TQ}\right)\right]^*\left(i_L^*\lambda_{\mR\times TQ}\right)\\
    &=i_L^*\lambda_{\mR\times TQ}+s\cdot\left(\overline{\tau}_{\mR\times TQ}\right)^*\left(\sigma^*i_L^*\lambda_{\mR\times TQ}\right)\\
    &=i_L^*\lambda_{\mR\times TQ}+s\cdot\pi_L^*\sigma
  \end{align*}
  because $i_L\circ\sigma=\sigma$ and the property $\sigma^*\lambda_{\mR\times TQ}=\sigma$ of the canonical form.
\end{proof}

For every $\gamma\in\Omega^1\left(Q\right)$, let us indicate by $\overline{\gamma}\in C^\infty\left(TQ\right)$ the linear function
\[
\overline{\gamma}\left(v_q\right):=\gamma_q\left(v_q\right).
\]
Select a local basis $\left\{Z_i\right\}\subset\mathfrak{X}\left(Q\right)$ and let $\left\{\beta^i\right\}\subset\Omega^1\left(Q\right)$ be its dual basis. Thus we can construct the local basis of vertical vector fields
\[
\left\{\left(Z_i^C\right)^{1_L},\left(Z_i^V\right)^{1_L},Z_{\sigma^i}\right\}\subset\mathfrak{X}^V\left(W_L\right),
\]
where
\[
\sigma^i:=\left(\text{pr}_2\circ\tau_Q\right)^*\beta^i-\overline{\beta}^idt
\]
is a basis for $\Gamma\left(I_{\text{con}}\right)$.

As discussed above, the lifting of vector fields to $W_L$ does not coincide with the restriction of more geometrical notions of lifts to the submanifold $W_L\subset T^*\left(\mR\times TQ\right)$. This situation changes whenever $Z\in\mathfrak{X}\left(Q\right)$ is an infinitesimal symmetry for $L$.
\begin{prop}\label{prop:AgreeLiftDefinitions}
  Let $Z\in\mathfrak{X}\left(Q\right)$ be an infinitesimal symmetry for the Lagrangian $L$, i.e.
  \[
  Z^C\cdot L=0.
  \]
  Then $\left(Z^C\right)^{1_L}=\left.\left(Z^C\right)^{1*}\right|_{W_L}$, where $\left(Z^C\right)^{1*}\in\mathfrak{X}\left(T^*\left(\mR\times TQ\right)\right)$ is the lift of $Z^C$ defined in Remark \ref{rem:DifferentLifts}.
\end{prop}
\begin{proof}
  According to the formulas of Theorem \ref{thm:LiftsExists}, for
  \[
  Z=Z^i\frac{\partial}{\partial q^i}
  \]
  we have
  \[
  Z^C=Z^i\frac{\partial}{\partial q^i}+v^k\frac{\partial Z^i}{\partial q^k}\frac{\partial}{\partial v^i}
  \]
  and so
  \[
  \mu_{Z^C}L=0.
  \]
  Then the formula for the lift becomes
  \[
  \left(Z^C\right)^{1_L}=Z^i\frac{\partial}{\partial q^i}+v^k\frac{\partial Z^i}{\partial q^k}\frac{\partial}{\partial v^i}-p_k\frac{\partial Z^k}{\partial q^i}\frac{\partial}{\partial p_i}=\left(Z^C\right)^{1*}
  \]
  as required.
\end{proof}

Under assumption $L\not=0,E\not=0$, functions $\mu$ associated to the elements $\left(Z^C\right)^{1_L}$ and $\left(Z^V\right)^{1_L}$ can be calculated using the formulas given in the proof of Theorem \ref{thm:LiftsExists}: We obtain that
\begin{align*}
  \mu_{Z^C}&=\frac{1}{L}Z^C\cdot L\\
  \mu_{Z^V}&=\frac{1}{E}\left(Z^V\right)^{1_L}\cdot E
\end{align*}
for every $Z\in\mathfrak{X}\left(Q\right)$.

Finally, the contraction of these vector fields with the canonical form $\lambda_L$ has the following properties
\begin{align}
  \left.\lambda_L\right|_{\rho_{\left(t,v_q\right)}}\left(\left(Z_i^V\right)^{1_L}\right)&=0\label{Eq:ContractZLambda.1}\\
  \left.\lambda_L\right|_{\rho_{\left(t,v_q\right)}}\left(\left(Z_i^C\right)^{1_L}\right)&=\rho_{\left(t,v_q\right)}\left(Z_i^C\right)\label{Eq:ContractZLambda.2}\\
  \left.\lambda_L\right|_{\rho_{\left(t,v_q\right)}}\left(Z_{\beta^i}\right)&=0
\end{align}
for every $\rho_{\left(t,v_q\right)}\in\left.W_L\right|_{\left(t,v_q\right)}$. Using \eqref{Eq:ContactSubbundleOnTQ}, we can write
\begin{equation}\label{eq:RhoInTermsAlpha}
  \rho_{\left(t,v_q\right)}=L\left(t,v_q\right)dt+\alpha\circ T_{v_q}\tau_Q-\alpha\left(v_q\right)dt
\end{equation}
for some $\alpha\in T_qQ$. Then
\[
\left.\lambda_L\right|_{\rho_{\left(t,v_q\right)}}\left(\left(Z_i^C\right)^{1_L}\right)=\alpha\left(Z_i\right).
\]
We can write this last equation in an interesting form: Using the map
\begin{equation}\label{eq:DeffOverTau}
  \overline{\tau}:W_L\rightarrow T^*Q:\rho\mapsto\alpha
\end{equation}
if and only if $\rho$ is given by formula \eqref{eq:RhoInTermsAlpha}, we can pull the linear functions
\[
\overline{Z}_i\left(\alpha_q\right):=\alpha\left(\left.Z_i\right|_q\right)
\]
back to $W_L$; then
\[
\left(Z_i^C\right)^{1_L}\lrcorner\lambda_L=\overline{\tau}^*\overline{Z}_i.
\]
From now on, we will drop the map $\overline{\tau}$ in the expression of these functions.

\subsection{Equations of motion for Lagrangian systems without force term}
\label{sec:equat-moti-lagr}

We will find equations of motion for a Lagrangian system without force term. It is interesting to note that equations of similar nature con be found in the literature, see \cite{doi:10.1080/14689360903360888,doi:10.1080/14689360802609344}. 

Now, if $\Gamma:I\subset\mR\rightarrow W_L$ is a solution curve for the data $\left(W_L,\lambda_L\right)$, and $L$ has no zeros, then the conditions found in Section \ref{sec:lepage-equiv-probl-1} can be translated into
\begin{align*}
  \Gamma^*\left(\left(Z^C\right)^{1_L}\lrcorner d\lambda_L\right)&=0\\
  \Gamma^*\left(\left(Z^V\right)^{1_L}\lrcorner d\lambda_L\right)&=0\\
  \Gamma^*\left(Z_\sigma\lrcorner d\lambda_L\right)&=0,
\end{align*}
for $Z\in\mathfrak{X}\left(Q\right)$ and $\sigma\in\Gamma\left(I_{\text{con}}\right)$. Moreover, using Sections \ref{sec:LiftToWL} and \ref{sec:local-basis-vector}, we can describe the equations of motion as follows.

\begin{prop}\label{prop:QuasiELequations}
  If a curve $\Gamma:I\subset\mR\rightarrow W_L$ gives rise to a solution curve for the data $\left(W_L,\lambda_L\right)$ associated to a non zero Lagrangian $L$ with non zero energy, then
\begin{align*}
  \Gamma^*\left(Z^C\cdot Ldt-d\overline{Z}\right)&=0\\
  \Gamma^*\left(Z^V\cdot L-\overline{Z}\right)&=0\\
  \Gamma^*\left(\pi_L^*\sigma\right)&=0,
\end{align*}
for any $Z\in\mathfrak{X}\left(Q\right)$ and $\sigma\in\Gamma\left(I_{\text{con}}\right)$.
\end{prop}
\begin{proof}
  Let $\sigma\in\Gamma\left(I_{\text{con}}\right)$ be a section of the contact bundle; Proposition \ref{prop:LiftSigma} tells us that
  \[
  Z_\sigma\lrcorner d\lambda_L=-\pi_L^*\sigma.
  \]
  Then
  \begin{align}
    0&=\Gamma^*\left(Z_\sigma\lrcorner d\lambda_L\right)\cr
    &=-\Gamma^*\left(\pi_L^*\sigma\right).\label{eq:EqForSigma}
  \end{align}
  Now, defining property of lifts translated into
  \[
  \mu_W\lambda_L=\cL_{W^{1_L}}\lambda_L=W^{1_L}\lrcorner d\lambda_L+d\left(W^{1_L}\lrcorner\lambda_L\right);
  \]
  using Equations \eqref{Eq:ContractZLambda.1} and $W=Z^V$ we see that
  \[
  \frac{1}{L}\left(Z^V\right)^{1_L}\cdot E\lambda_L=\left(Z^V\right)^{1_L}\lrcorner d\lambda_L
  \]
  and so
  \[
  \Gamma^*\left(\frac{1}{L}\left(Z^V\right)^{1_L}\cdot E\lambda_L\right)=0.
  \]
  But we know that
  \[
  \left.\lambda_L\right|_{\rho_{\left(t,v_q\right)}}=Ldt+\alpha\circ T_{v_q}\tau_Q-\alpha\left(v_q\right)dt
  \]
  for $\alpha\in T_q^*Q$, and the term $\alpha\circ T_{v_q}\tau_Q-\alpha\left(v_q\right)dt$ belongs to $\left.I_{\text{con}}\right|_{\left(t,v_q\right)}$, so Equation \eqref{eq:EqForSigma} implies that
  \[
  \Gamma^*\left(\lambda_L\right)=L\circ\Gamma dt.
  \]
  Then
  \[
  \Gamma^*\left(\left(Z^V\right)^{1_L}\cdot E\right)=0;
  \]
  the final form for this equation results from the identity
  \[
  \left(Z^V\right)^{1_L}\cdot E=Z^V\cdot L-\overline{Z}+\mu_{Z^V}\rho
  \]
  taking into account that $\Gamma^*\left(\left(Z^V\right)^{1_L}\cdot E\right)=E\mu_{Z^V}\circ\Gamma=0$.
  
  Finally, recalling that $\left(Z^C\right)^{1_L}\lrcorner\lambda_L=\overline{Z}$,
  \begin{align*}
    0&=\Gamma^*\left(\left(Z^C\right)^{1_L}\lrcorner d\lambda_L\right)\\
    &=\Gamma^*\left(\cL_{\left(Z^C\right)^{1_L}}\lambda_L-d\left(\left(Z^C\right)^{1_L}\lrcorner\lambda_L\right)\right)\\
    &=\Gamma^*\left(\mu_{Z^C}\lambda_L-d\overline{Z}\right)\\
    &=\Gamma^*\left(Z^C\cdot Ldt-d\overline{Z}\right)
  \end{align*}
  as required.
\end{proof}

\subsection{Equations of motion for general Lagrangian systems}
\label{sec:Equations-General-Lagrangian}

It only remains to find the expression of the extremal conditions for general Lagrangian systems, i.e. something similar to Proposition~\ref{prop:QuasiELequations} when a force term is allowed.

\begin{prop}\label{prop:QuasiELequationsGeneral}
  Let $\Gamma:I\subset\mR\rightarrow W_L$ be a curve associated to a solution curve for the general system $\left(W_L,\lambda_L,F\right)$. Let us suppose further that $L\not=0$. Then
  \begin{align*}
    \Gamma^*\left(\left(Z^C\cdot L+\left<F,Z\right>\right)dt-d\overline{Z}\right)&=0\\
    \Gamma^*\left(Z^V\cdot L-\overline{Z}\right)&=0\\
    \Gamma^*\left(\pi_L^*\sigma\right)&=0,
  \end{align*}
  for any $Z\in\mathfrak{X}\left(Q\right)$ and $\sigma\in\Gamma\left(I_{\text{con}}\right)$.
\end{prop}
\begin{proof}
  The proof goes as in Proposition~\ref{prop:QuasiELequations}. The only difference has to do with the terms associated to the force term
  \[
  \widetilde{F}\wedge dt.
  \]
  Because $T\pi_L\circ W^{1_L}=W$ for all $W\in\mathfrak{X}\left(\mR\times TQ\right)$, we obtain
  \[
  \left<\widetilde{F},\left(Z^C\right)^{1_L}\right>=\left<F,Z\right>
  \]
  and $\left<\widetilde{F},\left(Z^V\right)^{1_L}\right>=\left<\widetilde{F},Z_\sigma\right>=0$ for $Z\in\mathfrak{X}\left(Q\right)$ and $\sigma\in\Gamma\left(I_{\text{con}}\right)$; the result follows from these considerations.
\end{proof}
These are the equations of motion in quasi-velocities and quasi-momenta for general Lagrangian systems.

\subsection{On the nature of the equations of motion}
\label{sec:nature-equat-moti}

Propositions \ref{prop:QuasiELequations} and \ref{prop:QuasiELequationsGeneral} tell us that a curve $\Gamma:I\rightarrow W_L$ gives rise to a solution curve for a Lagrangian system if and only if its tangent vector field belongs to the annihilator of the set of forms
\[
\mathcal{B}:=\left\{\left(Z^C\cdot L+\left<F,Z\right>\right)dt-d\overline{Z},Z^V\cdot L-\overline{Z},\pi_L^*\sigma:Z\in\mathfrak{X}\left(Q\right),\sigma\in\Gamma\left(I_{\text{con}}\right)\right\}.
\]
We are assuming that functions are $0$-forms. The next result reduces this set to a more manageable set of forms.
\begin{lem}
  Let $\left\{Z_i\right\}\subset\mathfrak{X}\left(Q\right)$ be a basis of vector fields on $Q$ and $\left\{\sigma_i\right\}\subset\Gamma\left(I_{\text{con}}\right)$ a basis of sections for the bundle $I_{\text{con}}$. Let us define
  \[
  \mathcal{B}':=\left\{\left(Z_i^C\cdot L+\left<F,Z_i\right>\right)dt-d\overline{Z_i},Z_i^V\cdot L-\overline{Z_i},\pi_L^*\sigma_i\right\}.
  \]
  A curve $\Gamma:I\rightarrow W_L$ satisfies the equations
  \[
  \Gamma^*\alpha=0
  \]
  for every $\alpha\in\mathcal{B}$ if and only if
  \[
  \Gamma^*\beta=0
  \]
  for every $\beta\in\mathcal{B}'$.
\end{lem}

In some cases we will have the following situation: We have a set of forms $\mathcal{F}$ as above on a manifold $W$ and a submersion $P:W\rightarrow\widetilde{W}$.
\begin{defc}[Quotient equations]
  The set of forms $\widetilde{\mathcal{F}}$
  on $\widetilde{W}$ defined as follows
  \begin{equation}\label{Eq:QuotientEDS}
    \widetilde{\mathcal{F}}:=\left\{\gamma\in\Omega^\bullet\left(\widetilde{W}\right):P^*\gamma\in\mathcal{F}\right\}.
  \end{equation}
  will be called \emph{set of quotient equations.}
\end{defc}

The following consequence of this definition will be useful later.
\begin{cor}
  If $P^*\beta,\beta\in\Omega^\bullet\left(\widetilde{W}\right)$ belongs to $\mathcal{F}$, then $\beta\in\widetilde{\mathcal{F}}$.
\end{cor}

Then necessary conditions for curves in $\widetilde{W}$ to be projections via $P$ of solution curves for $\mathcal{F}$ in $W$ can be obtained.
\begin{lem}
  Let ${\Gamma}:I\rightarrow{W}$ be a curve in ${W}$ which is a solution curve for $\mathcal{F} $and define
  \[
  \widetilde{\Gamma}:=P\circ\Gamma.
  \]
  Then $\widetilde{\Gamma}^*\beta=0$ for all $\beta\in\widetilde{\mathcal{F}}$.
\end{lem}

When these conditions are also sufficient, it is said that we have solved a \emph{reconstruction problem}. We will not pursue these issues here; readers interested in a formulation of the reconstruction problem from this viewpoint are referred to \cite{MR2151123}.

\subsection{Equations of motion, translations and diffeomorphisms}

On the other hand, it could happen that we have a bundle isomorphism $\Phi:TM\rightarrow TM$ (not necessarily a vector bundle morphism) covering a diffeomorphism $\phi:M\rightarrow M$,
\[
\begin{diagram}
  \node{TM}\arrow{e,t}{\Phi}\arrow{s,l}{\tau_M}\node{TM}\arrow{s,r}{\tau_M}\\
  \node{M}\arrow{e,b}{\phi}\node{M}
\end{diagram}
\]

It is interesting to see under what conditions such bundle morphism relate equations of motion of a Lagrangian system on $TM$ to equations of motion of a Lagrangian system on the same bundle. We will use Cartan-like formulation in order to establish sufficient conditions ensuring equivalence of the set of equations of motion; this problem will become important when discussing Fehér Lagrangian in Section \ref{sec:LagrangianAKSRouth}.

\begin{note}
  Recall that given $X,Y$ manifolds, $W\subset T^*X$ a subbundle and
  $f:Y\rightarrow X$ a surjective submersion, the pullback bundle
  \[
  \begin{diagram}
    \node{f^*W}\arrow{e,t}{\text{pr}_1}\arrow{s,l}{\text{pr}_2}\node{W}\arrow{s,r}{\overline{\tau}_X}\\
    \node{Y}\arrow{e,b}{f}\node{X}
  \end{diagram}
  \]
  can be seen as a subbundle of $T^*Y$ with inclusion given by
  \[
  \left(y,\alpha\right)\mapsto\left(T_yf\right)^*\alpha\in T_y^*Y.
  \]
\end{note}

We expect (see for example \cite{OlvEquiv,Krupka20121154}) that Lagrangians which differ in a total derivative yield to the same equations of motions; the correct way to capture this fact in our setting is to suppose that their Lagrangian forms could differ in a contact form. Additionally \cite{MR1188444}, these equations of motion would remain unchanged if these Lagrangian forms differ in a closed $1$-form, which are associated to surface terms in their corresponding actions.

Given these considerations, it is important to see how translations along a form modify equations of motion of a Cartan-like system. This can be achieved using the following general result.

\begin{prop}\label{Prop:TranslationSpaceForms}
  Let $P$ be a manifold and $\alpha\in\Omega^1\left(P\right)$ a $1$-form on $P$. Set $t_\alpha:T^*P\rightarrow T^*P$ for the translation induced by $\alpha$, i.e
  \[
  t_\alpha\left(\beta\right):=\beta+\left.\alpha\right|_{\overline{\tau}_P\left(\beta\right)}.
  \]
  Let $i:W\hookrightarrow T^*P$ be an affine subbundle and define $W_\alpha:=t_\alpha\left(W\right)$. Then
  \begin{itemize}
  \item $W_\alpha$ is an affine subbundle of $T^*P$.
  \item If $\lambda_{W_\alpha}$ and $\lambda_W$ are the restrictions of the
    canonical $1$-form $\lambda_P$ to $W_\alpha$ and $W$,
    \[
    \lambda_{W_\alpha}=t_{-\alpha}^*\lambda_W+i_\alpha^*\left(\overline{\tau}_P^*\alpha\right)
    \]
    where $i_\alpha:W_\alpha\hookrightarrow T^*P$ is the canonical inclusion.
  \end{itemize}
\end{prop}
\begin{proof}
  First item is consequence of the fact that $t_\alpha:T^*P\rightarrow T^*P$ is a diffeomorphism.
  
  For the second item, let $p:=\overline{\tau}_P\left(\beta\right)$ for $\beta\in W$; then $\left.\alpha\right|_{p}+\beta\in W_\alpha$, and so
  \[
  \left.\lambda_{W_\alpha}\right|_{\left.\alpha\right|_p+\beta}=\left(\left.\alpha\right|_p+\beta\right)\circ T_{\left.\alpha\right|_p+\beta}\overline{\tau}_P=\beta\circ T_{\left.\alpha\right|_p+\beta}\overline{\tau}_P+\left.\alpha\right|_p\circ T_{\left.\alpha\right|_p+\beta}\overline{\tau}_P.
  \]
  On the other side
  \begin{align*}
    t_{-\alpha}^*\left(\left.\lambda_W\right|_\beta\right)&=\beta\circ T_\beta\overline{\tau}_P\circ T_{\beta+\left.\alpha\right|_p}t_{-\alpha}\\
    &=\beta\circ T_{\left.\alpha\right|_p+\beta}\overline{\tau}_P
  \end{align*}
  because of the identity $\overline{\tau}_P\circ t_{-\alpha}=\overline{\tau}_P$; moreover
  \[
  \left.i_\alpha^*\left(\overline{\tau}_P^*\alpha\right)\right|_{\left.\alpha\right|_p+\beta}=\alpha\circ T_{\left.\alpha\right|_p+\beta}\overline{\tau}_P.
  \]
  Comparing with previous equations, the result follows.
\end{proof}

Thus, translations yield to the occurrence of gyroscopic forces in equations of motion, as the following corollary to the previous proposition shows.

\begin{cor}\label{cor:AffineEquations}
  Let $\Gamma_\alpha:M\rightarrow W_\alpha$ be a map such that
  \[
  \Gamma_\alpha^*\left(X\lrcorner d\lambda_{W_\alpha}\right)=0
  \]
  for $X$ a vector field on $W_\alpha$. Then for $\Gamma:=t_{-\alpha}\circ\Gamma_\alpha$ the following identity
  \[
  \Gamma^*\left(\left(Tt_{-\alpha}\circ X\right)\lrcorner\left(d\lambda_{W}+di^*\left(\overline{\tau}_P^*\alpha\right)\right)\right)=0
  \]
  holds.
\end{cor}

We are ready to prove a result concerning equations of motion of Lagrangians related through bundle isomorphisms of tangent bundle; as expected, neither a contact nor a closed form change these equations.

\begin{prop}\label{Prop:DiffeomSolutions}
  Let $\phi:M\rightarrow M,\Phi:TM\rightarrow TM$ be as above, and suppose that for $L,L'\in C^\infty\left(\mR\times TM\right)$ the following identity
  \begin{equation}\label{eq:FirstCondContactProj}
    Ldt+\Theta+\rho=\left(\text{id}\times\Phi\right)^*\left(L'dt\right)
  \end{equation}
  holds, where $\rho,\Theta\in\Omega^1\left(\mR\times TM\right)$, $\rho$ is an arbitrary $1$-form and $\Theta$ is a \emph{contact form} (see Definition \ref{Def:AnotherDescriptionIcon}).
  
  Moreover, suppose further that for $F,F':TM\rightarrow T^*M$ bundle maps on $\text{id}_M$, we have
  \begin{equation}\label{eq:SecondCondContactProj}
    F\left(v_m\right)=\left(T_m\phi\right)^*\left(F'\left(\Phi\left(v_m\right)\right)\right)\in T_m^*M
  \end{equation}
  for every $v_m\in T_mM$, and that $\Phi$ is a \emph{contact map},
  \begin{equation}\label{eq:ThirdCondContactProj}
    \left(\text{id}\times\Phi\right)^*I_{\text{con}}\subset I_{\text{con}}.
  \end{equation}
  
  Then equations of motion of Lagrangian system $\left(M,L,F+d\rho\right)$ and $\left(M,L',F'\right)$ are in one-to-one correspondence via $\Phi$. In particular, equations of motion remains unchanged for closed forms $\rho$.
\end{prop}
\begin{proof}
  In the notation of Proposition \ref{Prop:TranslationSpaceForms}, we have that
  \[
  \left(\text{id}\times\Phi\right)^*\left(W_{L'}\right)=\left(W_L\right)_\rho.
  \]
  Thus, for $\gamma:I\subset\mR\rightarrow\mR\times TM$ a solution curve for $\left(M,L,F\right)$, there exists a curve $\Gamma:I\rightarrow W_L$ such that
  \[
  \Gamma^*\left(X\lrcorner\left(d\lambda_L+\widetilde{F}\wedge dt+d\rho\right)\right)=0
  \]
  for every $X\in\mathfrak{X}^{V\left(\text{pr}_1\circ\pi_L\right)}\left(W_L\right)$.
  
  Then, the curve
  \[
  \gamma_\Phi:I\rightarrow\mR\times TM:t\mapsto\left(\text{id}\times\Phi\right)\left(\gamma\left(t\right)\right)
  \]
  is covered by $\Gamma_\rho:I\rightarrow\left(W_{L'}\right)_\rho$ such that the following diagram commutes
  \begin{center}
    \begin{tikzpicture}
      \matrix (m) [matrix of math nodes, row sep=4em, column sep=4em,
      text height=1.5ex, text depth=0.35ex, inner sep=1.5ex] { & W_L & \left(W_{L'}\right)_\rho \\ I & \mR\times TM & \mR\times TM \\};
      \path[>=latex,->,shorten >=3pt]
      (m-2-1) edge [bend right=30,dashed] node[below] {$ \gamma_\Phi $} (m-2-3) 
      edge [bend left=60,dashed] node[above] {$ \Gamma_\rho $} (m-1-3)
      edge node[above] {$ \Gamma $} (m-1-2)
      edge node[below] {$ \gamma $} (m-2-2)
      (m-1-2) edge node[above] {$ \left(\Phi^{-1}\right)^* $} (m-1-3)
      edge node[right] {$ \pi_L $} (m-2-2)
      (m-1-3) edge node[right] {$ \pi_{\rho} $} (m-2-3)
      (m-2-2) edge node[below] {$ \Phi $} (m-2-3);
    \end{tikzpicture}
  \end{center}
  Then, if $\lambda_\rho$ is the pullback of canonical $1$-form on $T^*\left(\mR\times TM\right)$ to $\left(W_{L'}\right)_\rho$, the curve $\Gamma_\rho=\left(\Phi^{-1}\right)^*\circ\Gamma$ is such that 
  \[
  \Gamma_\rho^*\left(X'\lrcorner\left(d\lambda_\rho+\widetilde{F'}\wedge dt\right)\right)=0
  \]
  for every $X'\in\mathfrak{X}^{V\left(\text{pr}_1\circ\pi_\rho\right)}\left(\left(W_{L'}\right)_\rho\right)$. Now, from Corollary \ref{cor:AffineEquations}, it follows that $\Gamma':=t_{-\rho}\circ\Gamma_\rho:I\rightarrow W_{L'}$ is a curve covering $\gamma_\Phi$ such that
  \[
  \left(\Gamma'\right)^*\left(Z\lrcorner\left(d\lambda_{L'}+\widetilde{F'}\wedge dt\right)\right)=0
  \]
  for every $Z\in\mathfrak{X}^{V\left(\text{pr}_1\circ\pi_{L'}\right)}\left(W_{L'}\right)$. Then $\gamma_\Phi$ is a solution curve for $\left(M,L',F'\right)$.
  
  Finally, using that $\Phi$ and $\phi$ are invertible, every solution curve of $\left(M,L',F'\right)$ gives rise to a solution curve for $\left(M,L,F+d\rho\right)$.
\end{proof}

\section{Cartan-like description for intrinsically constrained systems}
\label{sec:lepage-equiv-system}

Let us see how to reformulate the Cartan-like theory developed for Lagrangian systems in order to include intrinsically constrained systems. From Definition \ref{Def:IntConstLag} we know that an intrinsically constrained system is a triple $\left(\pi:M\rightarrow N,L,F\right)$, and its critical curves are projections of the critical curves of the associated Lagrangian system $\left(M,p_1^*L,F\right)$, where $p_1:TM\rightarrow T_MN:=TN\times_NM$ is determined by the formula
\[
p_1\left(v_m\right):=\left(T_m\pi\left(v_m\right),m\right).
\]
We can form the bundle $W_{p_1^*L}\subset T^*\left(\mR\times TM\right)$ using formula \eqref{eq:WLDefinition}, and equations of motion arise from Proposition \ref{prop:QuasiELequationsGeneral}. In fact, description of Routh reduction in \cite{eprints21388} makes use of these kind of Lagrangian systems. Thus in the present section we will focus on construct Lepage-equivalent problems for them.

\subsection{Lepage-equivalent problem for intrinsically constrained systems}
\label{sec:lepage-equiv-probl-2}

The dynamics of an intrinsically constrained Lagrangian system is tied to the dynamics of the associated Lagrangian system $\left(M,p_1^*L,F\right)$ \cite{eprints21388}. For this system, the subbundle $W_{p_1^*L}\subset T^*\left(\mR\times TM\right)$, defined by formula \eqref{eq:WLDefinition}, allows us to construct its equations of motion; moreover, the nature of the Lagrangian of this kind of system implies that its solutions live in a submanifold. It is proved in the following proposition. So we will be able to concentrate in this subbundle, and prove that dynamics of an intrinsically constrained system is totally determined by this restricted system.

\begin{prop}\label{Prop:Wp1inKerp1}
  Solution curves of $\left(W_{p_1^*L},\lambda_{p_1^*L},F\right)$ (see Definition~\ref{Def:SolCurveGenData}) lie in the subbundle $$\left(\ker{\left(\text{id}_\mR\times Tp_1\right)}\right)^0.$$
\end{prop}
\begin{proof}
  Let us introduce coordinates $\left(q^i\right)$ on $N$ and $\left(q^i,u^A\right)$ on $M$ adapted to the projection $\pi:M\rightarrow N$; let $\left(q^i,v^i;u^A,w^A\right)$ be the induced coordinates on $TM$. 

The map $p_1:TM\rightarrow T_MN$ becomes
\[
p_1\left(q^i,v^i;u^A,w^A\right)=\left(q^i,v^i;u^A\right).
\]

Then on the corresponding coordinate chart in $T^*\left(\mR\times TM\right)$, with coordinates $$\left(t,q^i,v^i,p_i,P_i;s,u^A,w^A,r_A,R_A\right),$$ we will have
  \begin{equation}\label{Eq:Wp1LLocal}
    \left.W_{p_1^*L}\right|_{\left(t,q^i,v^i;u^A,w^A\right)}=\left\{p_1^*Ld t+p_i\left(d q^i-v^id t\right)+r_A\left(d u^A-w^Ad t\right):p_i,r_A\in\mR\right\}.
  \end{equation}

  Then the canonical form reads
  \begin{equation*}
    \lambda_{p_1^*L}=p_1^*Ld t+p_i\left(d q^i-v^id t\right)+r_A\left(d u^A-w^Ad t\right),
  \end{equation*}
and using Definition~\ref{Eq:FFormDefinition} for the force term, the equation
  \[
  \Gamma^*\left(\frac{\partial}{\partial w^A}\lrcorner\left(d\lambda_{p_1^*L}+\widetilde{F}\wedge d t\right)\right)=\Gamma^*\left(\frac{\partial}{\partial w^A}\lrcorner d\lambda_{p_1^*L}\right)=0
  \]
  reduces to $r_A=0$, which is the local expression for the subbundle $\left(\ker{\left(\text{id}_\mR\times Tp_1\right)}\right)^0$.
\end{proof}

The subbundle
\[
W_{p_1^*L}^0:=W_{p_1^*L}\cap\left(\ker{\left(\text{id}_\mR\times Tp_1\right)}\right)^0
\]
will allow us to construct a kind of Lepage-equivalent problem for the intrinsically constrained system on $T_NM$. In order to formulate it, let us define $W_L^N\subset T^*\left(\mR\times T_MN\right)$ playing a similar rôle than $W_{p_1^*L}\rightarrow\mR\times TM$, but changing the base space to $\mR\times T_MN$. So let $p:T_MN\rightarrow TN$ be the canonical projection, and define the subbundle $J_{\text{con}}\subset T^*\left(\mR\times T_MN\right)$ such that
\begin{multline}\label{eq:JContactBundle}
  \left.J_{\text{con}}\right|_{\left(t,v_n,m\right)}:=\\
  =\left\{\beta\in T^*_{\left(t,v_n,m\right)}\left(\mR\times T_MN\right):\beta=\alpha\circ T_{\left(t,v_n,m\right)}\left(\text{id}\times p\right)\text{ for some }\alpha\in\left.I_{\text{con}}\right|_{\left(t,v_n\right)}\right\}.
\end{multline}

In the coordinates introduced in Proposition~\ref{Prop:Wp1inKerp1}, we have
\[
p\left(q^i,v^i;u^A\right)=\left(q^i,v^i\right),
\]
and so
\[
\left.J_{\text{con}}\right|_{\left(t,q^i,v^i;u^A\right)}=\left\{p_i\left(dq^i-v^idt\right):p_i\in\mR\right\}.
\]

Thus $W_L^N\subset T^*\left(\mR\times T_MN\right)$ is given by
\[
\left.W_L^N\right|_{\left(t,v_n,m\right)}:=L\left(t,v_n\right)d t+\left.J_{\text{con}}\right|_{\left(t,v_n,m\right)},
\]
as in Equation~\eqref{eq:WLDefinition}.

\begin{prop}\label{prop:Wp1LsimeqWNL}
  The bundle $W^0_{p_1^*L}$ coincides with the pullback bundle of $W^N_L$ along $\text{id}\times p_1$, namely
  \[
  \begin{diagram}
    \node{W_{p_1^*L}^0}\arrow{e,t,=}{}\arrow{se,b}{\overline{\tau}_{\mR\times TM}}\node{\left(\text{id}\times p_1\right)^*\left(W_{L}^N\right)}\arrow{e,t}{}\arrow{s,l}{}\node{W_L^N}\arrow{s,r}{\overline{\tau}_{\mR\times T_MN}}\\
    \node[2]{\mR\times TM}\arrow{e,b}{\text{id}\times p_1}\node{\mR\times T_MN}
  \end{diagram}
  \]
\end{prop}
\begin{proof}
  The subbundle $W_{p_1^*}^0$ can be described in local terms by using Eq.~\eqref{Eq:Wp1LLocal}; it results
  \[
  \left.W_{p_1^*L}^0\right|_{\left(t,q^i,v^i;u^A,w^A\right)}=\left\{p_1^*Ld t+p_i\left(d q^i-v^id t\right):p_i\in\mR\right\}.
  \]
  
  Now, if a $1$-form $\beta=sd t+r_id q^i+M_id v^i+N_Ad u^A$ belongs to $\left.J_{\text{con}}\right|_{\left(t,q^i,v^i;u^A\right)}$, there exists $\alpha=p_i\left(d q^i-v^id t\right)\in\left.I_{\text{con}}^1\right|_{\left(t,q^i,v^i\right)}$ such that $\beta=\alpha\circ T_{\left(t,q^i,v^i;u^A\right)}\left(\text{id}\times p\right)$; so contraction of both sides of this equation with a generic vector
  \[
  V:=T\frac{\partial}{\partial t}+Q^i\frac{\partial}{\partial q^i}+V^i\frac{\partial}{\partial v^i}+U^A\frac{\partial}{\partial u^A}
  \]
  gives
  \[
  p_iQ^i-\left(p_iv^i\right)T=sT+r_iQ^i+M_iV^i+N_AU^A,
  \]
  namely $\beta=p_i\left(d q^i-v^id t\right)=\alpha$. Therefore
  \[
  \left.W_L^N\right|_{\left(t,q^i,v^i;u^A\right)}=\left\{L\left(t,q^i,v^i;u^A\right)d t+p_i\left(d q^i-v^id t\right):p_i\in\mR\right\};
  \]
  the isomorphism with $W_{p_1^*L}^0$ is given by the map
  \[
  \Psi:\left(\text{id}\times p_1\right)^*W_L^N\rightarrow W_{p_1^*L}^0:\left(t,w_m;\alpha_{\left(t,v_n,m\right)}\right)\xmapsto{\hskip1.3em}\alpha_{\left(t,v_n,m\right)}\circ T_{\left(t,w_m\right)}\left(\text{id}\times p_1\right).\qedhere
  \]
\end{proof}

In an intrinsically constrained Lagrangian system the external force is encoded by a bundle map $F:TM\rightarrow T^*M$ covering the identity in $M$. Definition~\ref{Eq:FFormDefinition} allows us to construct its associated $1$-form on $W^0_{p_1^*L}$, and the isomorphism $\Psi$ gives rise to a force $1$-form on $W_L^N$, which we will represent with the same symbol $\widetilde{F}$.

Thus, translation of Definition~\ref{Def:IntConstLag} for solution curves of an intrinsically constrained Lagrangian system to this new setting, gives rise to the following result.

\begin{prop}\label{prop:CurveOnWNL}
  A curve $m:I\subset\mR\rightarrow M$ is a critical curve for the intrinsically constrained system $\left(\pi:M\rightarrow N,L,F\right)$ if and only if there exists a curve $\Gamma:I\rightarrow W_L^N$ such that
  \begin{enumerate}
  \item $\tau_N\circ\text{pr}_2\circ\left(\text{id}\times p\right)\circ\overline{\tau}_{\mR\times T_MN}\circ\Gamma=m$,
  \item $\text{pr}_1\circ\overline{\tau}_{\mR\times T_MN}\circ\Gamma=\text{id}$, and
  \item $\Gamma^*\left(X\lrcorner\left(d\lambda_L^N+\widetilde{F}\wedge d t\right)\right)=0$ for any $X\in\mathfrak{X}^V\left(W_L^N\right)$, where $\widetilde{F}\in\Omega^1\left(W_L^N\right)$ is the $1$-form determined by $F$ on $W_L^N$.\label{Item:EquationsWLN} If $F$ comes from a $2$-form, this requeriment can be written as
    \[
    \Gamma^*\left(X\lrcorner\left(d\lambda_L^N+{F}\right)\right)=0
    \]
    for any $X\in\mathfrak{X}^V\left(W_L^N\right)$.
  \end{enumerate}
\end{prop}
\begin{proof}
  Relevant bundles and maps involved in this proof are indicated in the following diagram
  
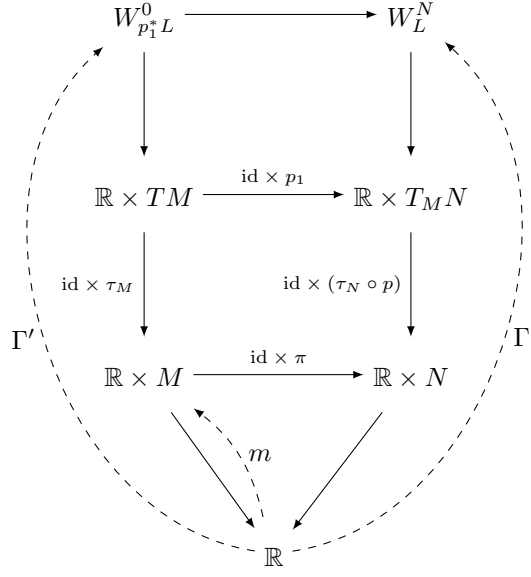
\begin{center}
    \begin{tikzpicture}
      \matrix (m) [matrix of math nodes, row sep=4em, column sep=2em,
      text height=3.5ex, text depth=1.5ex] { \displaystyle W^0_{p_1^*L} &&  \displaystyle W_L^N \\ \displaystyle \mR\times TM & & \displaystyle \mR\times T_MN \\ \displaystyle\mR\times M & & \displaystyle\mR\times N \\ & \displaystyle\mR & \\}; 
      \path[>=latex,->,font=\scriptsize]
      (m-1-1) edge (m-2-1)
      edge (m-1-3)
      (m-1-3) edge (m-2-3)
      (m-2-1) edge node[above] {$\text{id}\times p_1$} (m-2-3)
      edge node[left] {$\text{id}\times\tau_M$} (m-3-1)
      (m-2-3) edge node[left] {$\text{id}\times\left(\tau_N\circ p\right)$} (m-3-3)
      (m-3-1) edge node[above] {$\text{id}\times\pi$} (m-3-3)
      edge (m-4-2)
      (m-3-3) edge (m-4-2)
      ;
      \path[>=latex,->,dashed]
      (m-4-2) edge[bend right=20] node[right] {$m$} (m-3-1)
      (m-4-2) edge[bend left=65] node[left] {${\Gamma}'$} (m-1-1)
      (m-4-2) edge[bend right=65] node[right] {$\Gamma$} (m-1-3);
    \end{tikzpicture}
  \end{center}
  
  Definition~\eqref{Def:IntConstLag} applies to the system on $\left(W_{p_1^*L},\lambda_{p_1^*L},F\right)$, because it describes a general Lagrangian system. Moreover, from condition $\mathop{\text{Im}}{\widetilde\Gamma}\subset W_{p_1^*L}^0$ we have that these equations induce equations for the solution curve $\Gamma':I\rightarrow W_{p_1^*L}^0$; these equations are
  \[
  \Gamma^*\left(X\lrcorner\left( d\lambda_{p_1^*L}^0+\widetilde{F}\wedge d t\right)\right)=0
  \]
  for every $X\in\mathfrak{X}^{V\left(\text{pr}_1\circ\overline{\tau}_{\mR\times TM}\right)}\left(W_{p_1^*L}^0\right)$.

  On the other hand, the canonical forms $\lambda_{p_1^*L}^0$ and $\lambda_L^N$ on $W_{p_1^*L}^0$ and $W_L^N$ respectively have the same form in local coordinates, so we need to see if the extra variables in $W_{p_1^*L}^0$ yield to additional equations. But
  \[
  \frac{\partial}{\partial w^A}\lrcorner d\lambda^0_{p_1^*L}=0
  \]
  identically, so Equations~\eqref{Item:EquationsWLN} coincide the equations characterizing $m$ as a solution curve for the data $\left(W_{p_1^*L},\lambda_{p_1^*L},F\right)$.
\end{proof}

It means that intrinsically constrained systems can be described in terms of a Lepage-equivalent problem.

\subsection{Equations of motion for intrinsically constrained systems}
\label{sec:equat-moti-intr}

Let us make use of the lift to $W_L^N$ in order to find the equations of motion for an intrinsically constrained system. Recall that this bundle is not a classical Lepage-equivalent; therefore, it is not defined on a tangent bundle, and the contact structure used in its construction is borrowed from $TN$ via a pullback. So it is necessary to generalize the lift of vector fields found in Section \ref{sec:cont-manif-lift} to this case. It will be achieved in the present section by using the corresponding lifts on the bundle $W_{p_1^*L}^0$.

So, first let us fix an Ehresmann connection on $\pi:M\rightarrow N$. The lift of vector fields from $N$ to $M$ will be indicated by $X\mapsto X^{H_M}$. For every $Z\in\mathfrak{X}\left(N\right)$, we lift it to $M$ and using complete and vertical lifts, to $\left(Z^{H_M}\right)^C,\left(Z^{H_M}\right)^V\in\mathfrak{X}\left(\mR\times TM\right)$; to any $W\in\Gamma\left(V\pi\right)$, we can assign vector fields $W^C,W^V\in\mathfrak{X}\left(\mR\times TM\right)$. On the other side, we can construct two lifts to $T_MN$, namely
\[
Z^{C_N}+Z^{H_M},Z^{V_N}+0\in\mathfrak{X}\left(T_MN\right).
\]
Here $C_N,V_N$ indicate complete and vertical lift from $N$ to $TN$. Moreover, for $V\in\mathfrak{X}\left(M\right)$ a vertical vector field on $M$, we have the vector field
\[
0+V\in\mathfrak{X}\left(T_MN\right).
\]
Then the following result holds.
\begin{lem}\label{lem:VectorFieldsRelationThroughP1}
  Let $Z\in\mathfrak{X}\left(N\right)$ and $W\in\Gamma\left(V\pi\right)$ be arbitrary vector fields. Then
  \[
  \left(Z^{H_M}\right)^C,\left(Z^{H_M}\right)^V,W^C\in\mathfrak{X}\left(\mR\times TM\right)
  \]
  are $p_1$-related to
  \[
  Z^{C_N}+Z^{H_M},Z^{V_N}+0,0+W
  \]
  respectively; $W^V$ is in $\ker{Tp_1}$.
\end{lem}
\begin{proof}
  Let $\left(q^i,u^A\right)$ be local coordinates on $M$ adapted to $\pi$ and $\left(q^i,v^i,u^A,w^A\right)$ the associated coordinates on $TM$; in terms of these coordinates
  \begin{align*}
    \pi\left(q^i,u^A\right)&=\left(q^i\right)\\
    p_1\left(q^i,v^i,u^A,w^A\right)&=\left(q^i,v^i,u^A\right).
  \end{align*}
  There exist a collection $\left\{\Gamma^A_i\right\}$ of functions on the coordinates domain such that
  \[
  \left(\frac{\partial}{\partial q^i}\right)^{H_M}=\frac{\partial}{\partial q^i}+\Gamma_i^A\frac{\partial}{\partial u^A}.
  \]
  Moreover, there exists local functions $\left\{Z^i\right\}$ on $N$ such that
  \[
  Z=Z^i\frac{\partial}{\partial q^i};
  \]
  then
  \begin{multline*}
    \left(Z^{H_M}\right)^C=Z^i\frac{\partial}{\partial q^i}+Z^i\Gamma_i^A\frac{\partial}{\partial u^A}+v^k\frac{\partial Z^i}{\partial q^k}\frac{\partial}{\partial v^i}+\\
    +\left(v^k\Gamma_i^A\frac{\partial Z^i}{\partial q^k}+v^kZ^i\frac{\partial\Gamma_i^A}{\partial q^k}+w^CZ^i\frac{\partial\Gamma_i^A}{\partial u^C}\right)\frac{\partial}{\partial w^A}
  \end{multline*}
  and
  \[
  \left(Z^{H_M}\right)^V=Z^i\left(\frac{\partial}{\partial v^i}+\Gamma_i^A\frac{\partial}{\partial w^A}\right).
  \]
  So we have that
  \begin{align*}
    Tp_1\circ\left(Z^{H_M}\right)^C&=Z^i\frac{\partial}{\partial q^i}+Z^i\Gamma_i^A\frac{\partial}{\partial u^A}+v^k\frac{\partial Z^i}{\partial q^k}\frac{\partial}{\partial v^i}=Z^{C_N}+Z^{H_M}\\
    Tp_1\circ\left(Z^{H_M}\right)^V&=Z^i\frac{\partial}{\partial v^i}=Z^{V_N}+0.
  \end{align*}
  On the other side, there exist $\left\{W^A\right\}$ functions on $M$ such that
  \[
  W=W^A\frac{\partial}{\partial u^A};
  \]
  therefore
  \begin{align*}
    W^C&=W^A\frac{\partial}{\partial u^A}+\left(v^k\frac{\partial W^A}{\partial q^k}+w^B\frac{\partial W^A}{\partial u^B}\right)\frac{\partial}{\partial w^A}\\
    W^V&=W^A\frac{\partial}{\partial w^A}
  \end{align*}
  and the rest of the lemma follows.
\end{proof}

Let us recall that we have the maps $p:T_MN\rightarrow TN,q:T_MN\rightarrow M$ making the following diagram commutative
\begin{equation}\label{eq:DiagramCommutesPQ}
  \begin{diagram}
    \node{TM}\arrow{e,t}{p_1}\arrow{se,b}{T\pi}\node{T_MN}\arrow{s,l}{q}\arrow{e,t}{p}\node{M}\arrow{s,r}{\pi}\\
    \node[2]{TN}\arrow{e,b}{\tau_N}\node{N}
  \end{diagram}
\end{equation}
Let us indicate with
\[
q_1:W_{p_1^*L}^0\rightarrow W_L^N
\]
the map constructed above (see Proposition \ref{prop:Wp1LsimeqWNL}). The following result allows us to extend the notion of lift to the bundle $W_L^N$.
\begin{prop}
  Let $Y\in\mathfrak{X}\left(\mR\times TM\right)$ be a vector field $\left(p\circ p_1\right)$- and $\left(q\circ p_1\right)$-projectable, and vertical for the projection $\mR\times TM\rightarrow\mR$. Then the lift $Y^{1_{p_1^*L}}\in\mathfrak{X}\left(W_{p_1^*L}\right)$ is tangent to $W_{p_1^*L}^0$, and is $q_1$-related to a vector field
  \[
  Y^{1_N^L}:=Tq_1\circ Y^{1_{p_1^*L}}.
  \]
\end{prop}
\begin{proof}
  The requeriments on $Y$ imply that
  \[
  Y=Y^i\frac{\partial}{\partial q^i}+R^A\frac{\partial}{\partial u^A}+S^i\frac{\partial}{\partial v^i}+T^A\frac{\partial}{\partial w^A},
  \]
  where $Y^i,S^i$ are functions on $N$ and $R^A$ is a function on $M$. Then
  \begin{multline*}
    Y^{1_{p_1^*L}}=Y+\left(\mu_Yp_i-p_k\frac{\partial Y^k}{\partial q^i}-r_A\frac{\partial R^A}{\partial q^i}\right)\frac{\partial}{\partial p_i}+\left(\mu_Yr_B-p_k\frac{\partial Y^k}{\partial u^B}-r_A\frac{\partial R^A}{\partial u^B}\right)\frac{\partial}{\partial r_B}
  \end{multline*}
  where 
  \begin{multline*}
    \mu_Y=\frac{1}{L}\Bigg(Y^i\frac{\partial L}{\partial q^i}+R^A\frac{\partial L}{\partial u^A}+S^i\left(\frac{\partial L}{\partial v^i}-p_i\right)+\\
      +T^A\left(\frac{\partial L}{\partial w^A}-r_A\right)+p_kD_tY^k+r_AD_tR^A\Bigg).
  \end{multline*}
  Now $L$ does not depend on variables $w^A$, and on $W_{p_1^*L}^0$ we have $r_A=0$, so $\mu_Y$ becomes a function on $T_MN$; now
  \[
  \left.Y^{1_{p_1^*L}}\right|_{W_{p_1^*L}^0}=Y+\left(\mu_Yp_i-p_k\frac{\partial Y^k}{\partial q^i}\right)\frac{\partial}{\partial p_i}
  \]
  is tangent to $W_{p_1^*L}^0$, and
  \[
  Tq_1\circ\left.Y^{1_{p_1^*L}}\right|_{W_{p_1^*L}^0}=Y^i\frac{\partial}{\partial q^i}+R^A\frac{\partial}{\partial u^A}+S^i\frac{\partial}{\partial v^i}+\left(\mu_Yp_i-p_k\frac{\partial Y^k}{\partial q^i}\right)\frac{\partial}{\partial p_i}
  \]
  is a vector field on $W^N_L$, as required.
\end{proof}

Then we can relate equations of motion on bundles $W_L^N$ and $W_{p_1^*L}^0$. First, let us recall the following description for these bundles, namely
\begin{align*}
  &\left.W_{p_1^*L}^0\right|_{\left(t,V_m\right)}=\left\{p_1^*Ldt+\widehat{\alpha}\circ T_m\pi\circ T_{V_m}\tau_M-\widehat{\alpha}\left(T_m\pi\left(V_m\right)\right)dt:\widehat{\alpha}\in T_{\pi\left(m\right)}^*N\right\}\\
  &\left.W_{L}^N\right|_{\left(t,v_n,m\right)}=\left\{Ldt+\widehat{\alpha}\circ T_{v_n}\tau_N\circ T_{\left(v_n,m\right)}p-\widehat{\alpha}\left(v_n\right)dt:\widehat{\alpha}\in T_{\pi\left(m\right)}^*N\right\}.
\end{align*}
It means that for every $\sigma\in\Omega^1\left(N\right)$, we can construct, as before, a pair of vector fields $Z_\sigma'\in\mathfrak{X}\left(W_{p_1^*L}^0\right),Z_\sigma''\in\mathfrak{X}\left(W_L^N\right)$ which are $q_1$-related. In fact, for every $\sigma\in\Omega^1\left(N\right)$, we define
\begin{align*}
  \sigma'&:=\left(\pi\circ\tau_M\right)^*\sigma-\overline{\sigma}dt\\
  \sigma''&:=\left(\tau_N\circ p\right)^*\sigma-\overline{\sigma}dt
\end{align*}
where $\overline{\sigma}$ is the linear function induced by $\sigma$ on the corresponding base space; then we use Equation \eqref{Eq:ZSigmaDef} to define the vector fields.

\begin{cor}
  Equations of motion on $W_L^N$ are quotient equations via $q_1$ of equations of motion on $W_{p_1^*L}^0$.
\end{cor}
\begin{proof}
  Using $Y\in\mathfrak{X}\left(\mR\times TM\right)$ vector field $\left(p\circ p_1\right)$- and $\left(q\circ p_1\right)$-projectable, and vertical for the projection $\text{pr}_1:\mR\times TM\rightarrow\mR$, we can found a basis of vertical vector fields on both $W_{p_1^*L}^0$ and $W_L^N$; to them we need to add vector fields of the form $Z_\sigma'\in\mathfrak{X}\left(W_{p_1^*L}^0\right),Z_\sigma\in\mathfrak{X}\left(W_L^N\right)$ for $\sigma\in\Omega^1\left(N\right)$. Now, let $\Gamma':I\rightarrow W_{p_1^*L}^0$ and $\Gamma:=q_1\circ\Gamma'$; then, from
  \[
  \lambda_{p_1^*L}^0=q_1^*\lambda_L^N
  \]
  we have that
  \[
  \left(\Gamma'\right)^*\left(Y^{1_{p_1^*L}}\lrcorner d\lambda^0_{p_1^*L}\right)=\Gamma^*\left(Y^{1^N_L}\lrcorner d\lambda_L^N\right),
  \]
  so $\Gamma'$ will be solution if and only if $\Gamma$ is.
\end{proof}

\section{Lepage-equivalent problems and symmetry}
\label{sec:group-action-decomp-WL}

Now let us concentrate in a Lagrangian system with symmetry. It means that there exists a Lie group $G$ with an action on $Q$ such that its lift to $TQ$ acts by symmetries of the Lagrangian function $L\in C^\infty\left(\mR\times TQ\right)$. As our viewpoint is to represent Lagrangian system $\left(Q,L,F\right)$ with the subbundle $W_L$ and its Cartan form $\lambda_L$, it is necessary to translate symmetry considerations to the new description.

\subsection{Momentum map for Lepage-equivalent problems}
\label{sec:lepage-equiv-probl}

Previously (see Definition~\ref{Def:InvLagSystem}) we defined a Lie group $G$ as being a symmetry group for the Lagrangian system $\left(Q,L,F\right)$ if and only if it acts on $Q$ in such a way that the canonical projection onto its orbit space $p_G^Q:Q\rightarrow Q/G$ defines a principal bundle, and it keeps the Lagrangian $L$ and the map $F$ invariant. Under these hypothesis, we have a natural lifting of this action to $T^*\left(\mR\times TQ\right)$, which preserves the canonical $1$-form $\lambda_{\mR\times TQ}$ and the subbundle $I_{\text{con}}$ (i.e. Proposition \ref{prop:AgreeLiftDefinitions}.)

\begin{lem}
 The subbundle $W_L$ and the canonical form $\lambda_L$ are invariant for the lifted action.
\end{lem}

Thus we can define the \emph{momentum map} $J:W_L\rightarrow\g^*$ via the classical formula
\[
\left<J\left(\alpha\right),\xi\right>:=\left.\lambda_L\right|_{\alpha}\left(\xi_{W_L}\right),\qquad\xi\in\g.
\]

This momentum map coincides with the original.
\begin{prop}
  Let $J_L:TQ\rightarrow\g^*$ be the momentum map for the invariant Lagrangian system $\left(Q,L,F\right)$. Then
  \[
  s_0^*J=J_L
  \]
  where $s_0:\mR\times TQ\rightarrow W_L$ is the section constructed in Lemma \ref{Lem:CanonicalSection}.
\end{prop}
\begin{proof}
  We have that
  \[
  \left.\lambda_L\right|_{\alpha}\left(\xi_{W_L}\right)=\alpha\left(\xi_{\mR\times TQ}\right).
  \]
  By Equation~\eqref{eq:CanonicSectionDef} for the section $s_0$, we have that
  \[
  \left<s_0^*J\left(t,v\right),\xi\right>=s_0\left(t,v\right)\left(\xi_{\mR\times TQ}\right)=\left.\theta_L\right|_{\left(t,v\right)}\left(\xi_{TQ}\right)=J_L\left(v\right),
  \]
  as required.
\end{proof}

This map is a suitable generalization of the momentum map to this setting.
\begin{cor}
  $J$ is conserved on solutions of $\left(W_L,\lambda_L,F\right)$.
\end{cor}
\begin{proof}
  It is just necessary to use the characterizations of the solutions for $\left(W_L,\lambda_L,F\right)$ given in Theorem~\ref{Thm:CharacterizationSolsLagSystem2}, realizing that $\xi_{W_L}\lrcorner\widetilde{F}=\left<F,\xi_Q\right>=0$ by Definition~\ref{Def:InvLagSystem}.
\end{proof}

Let us indicate by $J_0:T^*Q\rightarrow\g^*$ the momentum map associated to the lifted $G$-action on the exact symplectic manifold $T^*Q$.
\begin{prop}\label{prop:RelationMomentums2}
  Momentum maps $J$ and $J_0$ are related through
  \[
  J=\overline{\tau}^*J_0
  \]
  where $\overline{\tau}:W_L\rightarrow T^*Q$ is the map defined in Equation \eqref{eq:DeffOverTau}.
\end{prop}
\begin{proof}
  Let us recall from Remark \ref{W_LIdentification} that $\rho\in\left.W_L\right|_{\left(t,v_q\right)}$ corresponds to $\left(t,v_q,\alpha\right)$ if and only if
  \[
  \rho=L\left(t,v_q\right)dt+\alpha\circ T_{v_q}\tau_Q-\alpha\left(v_q\right)dt.
  \]
  Therefore
  \begin{align*}
    J\left(\rho\right)\left(\xi\right)&=\left.\lambda_L\right|_\rho\left(\xi_{W_L}\right)\\
    &=\rho\left(\xi_{\mR\times TQ}\right)\\
    &=\left(\alpha\circ T_{v_q}\tau_Q\right)\left(\xi_{TQ}\right)\\
    &=\alpha\left(\xi_Q\right)\\
    &=J_0\left(\alpha\right)\left(\xi\right)
  \end{align*}
  for every $\xi\in\g$, as required.
\end{proof}

\subsection{Symmetry and projection of solution curves}
\label{sec:marsd-weinst-like}

We have a $G$-action on $W_L$, a $G$-invariant form on this manifold and a momentum map, so it makes sense to ask about the $G$-invariance of solution curves.

\begin{lem}
  Let $g\in G$ be an element of the symmetry group, and $\gamma:I\rightarrow Q$ a solution for $\left(W_L,\lambda_L,F\right)$. Then $g\cdot\gamma:I\rightarrow Q$ is also a solution for $\left(W_L,\lambda_L,F\right)$. 
\end{lem}
\begin{proof}
  First recall that equivariance of the force term $F$ implies $G$-invariance of the form $\widetilde{F}$. From Definition~\ref{Def:SolCurveGenData}, in order to show $g\cdot\gamma$ is a solution for $\left(W_L,\lambda_L,F\right)$, we need to find a curve $\Gamma^g:I\rightarrow W_L$ such that the requeriments in this definition are met. If $\Gamma:I\rightarrow W_L$ is the corresponding curve for $\gamma$, we can see that $\Gamma^g:=g\cdot\Gamma$ fullfills the first two requeriments.

  For the last item in Definition~\ref{Def:SolCurveGenData}, we select a set of (perhaps local) $G$-invariant generators $\left\{Z\right\}$ for $\mathfrak{X}^V\left(W_L\right)$, and so
  \begin{align*}
    \left(\Gamma^g\right)^*\left(Z\lrcorner\left(d\lambda_L+\widetilde{F}\wedge dt\right)\right)&=\left(g\cdot\Gamma\right)^*\left(Z\lrcorner\left(d\lambda_L+\widetilde{F}\wedge dt\right)\right)\\
    &=\Gamma^*\left(Z\lrcorner\left(d\lambda_L+\widetilde{F}\wedge dt\right)\right)
  \end{align*}
because of the $G$-invariance of $\lambda_L$, $\widetilde{F}$ and $Z$.
\end{proof}

It means in particular that it is possible to project solution curves on quotient spaces by symmetry groups.

\subsection{Routh function and level sets of the momentum mapping}
\label{sec:routh-function-level}

We want to provide a definition for the Routh function associated to the problem $\left(W_L,\lambda_L\right)$. We fix an element $\mu\in\g^*$ which is regular for $J_L$, and define the submanifold
\[
W_L^\mu:=J^{-1}\left(\mu\right).
\]

\begin{lem}\label{lem:WLmuEquivalentMmu}
  Under identification \eqref{eq:IdentLagPont}, we have that
  \[
  J^{-1}\left(\mu\right)\simeq\mR\times\left(TQ\oplus J_0^{-1}\left(\mu\right)\right).
  \]
\end{lem}
\begin{proof}
  It is a consequence of Proposition~\ref{prop:RelationMomentums2}.
\end{proof}

\begin{note}
  Lemma \ref{lem:WLmuEquivalentMmu} tells us that $W_L^\mu$ is equivalent to manifold $M_\mu$ considered in \cite{2015arXiv150901946G}.
\end{note}

\begin{cor}\label{cor:FormWLmu}
  We have that
  \[
  \left.W_L^\mu\right|_{\left(t,v_q\right)}=\left\{L\left(t,v_q\right)dt+\alpha\circ T_{v_q}\tau_Q-\alpha\left(v_q\right)dt:\alpha\in J_0^{-1}\left(\mu\right)\cap T^*_qQ\right\}
  \]
  for every $\left(t,v_q\right)\in\mR\times TQ$.
\end{cor}



\section{Routh reduction for mechanical systems}
\label{sec:RouthRedForMechSystem}

Throughout this section $H:=G_\mu$ indicates the isotropy group for $\mu\in\g^*$ regular value for the momentum map $\mu$; $\hf$ will be the Lie algebra associated to $H$. It is time to relate the dynamics of the unreduced system $\left(Q,L,F\right)$ with the reduced system defined on $T\left(Q/G\right)\times Q/H\times\widetilde{\g}$ with Routhian $\overline{R}_\mu$ and a gyroscopic force coming from reduction of the $2$-form $\left<\mu,d\omega_Q\right>$. We know \cite{eprints21388} that this system can be interpreted as an intrinsically constrained system via a map
\[
p_1:T\left(Q/H\times\widetilde{\g}\right)\rightarrow T\left(Q/G\right)\times Q/H\times\widetilde{\g},
\]
and it was proved in Section \ref{sec:lepage-equiv-system} of the present work that equations of motion on $W^{Q/G}_{\overline{R}_\mu}$ are the projections of equations of motion for $W_{p_1^*\overline{R}_\mu}^0$. Thus Routh reduction in our formulation reduces to relate this latter system with Lagrangian system represented by the bundle $W_L$; the purpose of the present section is to prove this relation.

For $\mu\in\g^*$, we define on $\mR\times TQ$ the Routhian
\[
R_\mu\left(t,v_q\right):=L\left(t,v_q\right)-\left<\mu,\left.\omega_Q\right|_q\left(v_q\right)\right>
\]
where $\omega_Q\in\Omega^1\left(Q,\g\right)$ is a connection $1$-form on the $G$-principal bundle $p_G^Q:Q\rightarrow Q/G$. As $H$ is the isotropy group for $\mu\in\g^*$, $R_\mu$ induces a function $\overline{R}_\mu\in C^\infty\left(\mR\times TQ/H\right)$. We can provide the reader with a quick summary of the steps we will do below:
\begin{itemize}
\item First we will use the connection form $\omega_Q$ in order to
  find a decomposition of the contact bundle.
\item Then we will proceed to
  relate the Cartan bundle $W_L$ associated to the original Lagrangian
  system $\left(Q,L,F\right)$ with the Cartan bundle
  $W^0_{p_1^*\overline{R}_\mu}$ associated to the pullback Routhian
  $p_1^*\overline{R}_\mu$.
\item The force term $p_1^*f$ to be used in this
  system is determined by the reduced force term
  $f:TQ/G\rightarrow T^*\left(Q/G\right)$ induced by the $G$-invariant
  force $F$ \cite{2010IJGMM..07.1451L}.
\end{itemize}
In fact, we have the map
\begin{center}
  \begin{tikzpicture}
    \matrix (m) [matrix of math nodes, row sep=3em, column sep=3em,
    text height=1.5ex, text depth=0.25ex] { T\left(Q/H\times\widetilde{\g}\right) & T\left(Q/G\right)\times Q/H\times\widetilde{\g} & T\left(Q/G\right)\times\widetilde{\g} & TQ/G \\ }; \path[>=latex,->]
    (m-1-1) edge [bend right=15] node[below] {$ r_1 $} (m-1-4) 
            edge node[above] {$ p_1 $} (m-1-2) 
    (m-1-2) edge node[above] {$ \text{pr}_{13} $} (m-1-3)
    (m-1-3) edge node[above] {$ \sim $} (m-1-4);
  \end{tikzpicture}
\end{center}
so force term can be written as
\begin{center}
  \begin{tikzpicture}
    \matrix (m) [matrix of math nodes, row sep=3em, column sep=3em,
    text height=1.5ex, text depth=0.25ex] { T\left(Q/H\times\widetilde{\g}\right) & TQ/G & T^*\left(Q/G\right) &  T^*\left(Q/H\times\widetilde{\g}\right) \\ }; \path[>=latex,->]
    (m-1-1) edge [bend right=15] node[below] {$ p_1^*f $} (m-1-4) 
            edge node[above] {$ r_1 $} (m-1-2) 
    (m-1-2) edge node[above] {$ f $} (m-1-3)
    (m-1-3) edge node[above] {$ \left(\phi^H\right)^* $} (m-1-4);
  \end{tikzpicture}
\end{center}
i.e.
\[
p_1^*f:=\left(f\circ r_1\right)\circ T\phi^H,
\]
where
\[
\phi^H:Q/H\times\widetilde{\g}\rightarrow Q/G:\left(\left[q\right]_H,\left[q,\xi\right]_G\right)\mapsto\left[q\right]_G.
\]

\subsection{Connections and a decomposition of the contact bundle $I_{\text{con}}$}
\label{sec:conn-decomp-tang}

Let us consider connection $\omega_Q\in\Omega^1\left(Q,\g\right)$ previously chosen. It gives rise to a connection $\omega_{TQ}\in\Omega^1\left(TQ,\g\right)$ via pullback along $\tau_Q$
\[
\omega_{TQ}:=\tau_Q^*\omega_Q.
\]

Now, using the description of the contact subbundle provided by Equation~\eqref{Eq:ContactSubbundleOnTQ}, we can find a decomposition of this subbundle induced by a connection on $Q$. In fact, we have the pullback bundle
\[
\begin{diagram}
  \node{\left(p_G^Q\right)^*T^*\left(Q/G\right)}\arrow{e,t}{}\arrow{s,l}{}\node{T^*\left(Q/G\right)}\arrow{s,r}{\overline{\tau}_{Q/G}}\\
  \node{Q}\arrow{e,b}{p_G^Q}\node{Q/G}
\end{diagram}
\]
It gives rise to the decomposition
\[
T^*Q=\left(p_G^Q\right)^*\left(T^*\left(Q/G\right)\right)\times_Q\left(Q\times\g^*\right)
\]
induced by the connection $\omega_Q$, through the correspondence
\[
\left(\widehat{\alpha}_{\left[q\right]},q,\sigma\right)\longmapsto\widehat{\alpha}_{\left[q\right]}\circ T_qp_G^Q+\left<\sigma,\omega_Q\left(\cdot\right)\right>.
\]
It induces a factorization $I_{\text{con}}=\widetilde{I_{\text{con}}}\oplus I_{\g^*}$, where
\begin{align*}
  \left.\widetilde{I_{\text{con}}}\right|_{\left(t,v_q\right)}&:=\left\{\widehat{\alpha}_{\left[q\right]}\circ T_qp_G^Q\circ T_{v_q}\tau_Q-\widehat{\alpha}_{\left[q\right]}\circ T_qp_G^Q\left(v_q\right)dt:\widehat{\alpha}_{\left[q\right]}\in T_{\left[q\right]}^*\left(Q/G\right)\right\}\\
  \left.I_{\g^*}\right|_{\left(t,v_q\right)}&:=\left\{\left<\sigma,\omega_Q\circ T_{v_q}\tau_Q\right>-\left<\sigma,\omega_Q\left(v_q\right)\right>dt:\sigma\in\g^*\right\}\\
  &=\left\{\left<\sigma,\left.\omega_{TQ}\right|_{v_q}-\omega_Q\left(v_q\right)dt\right>:\sigma\in\g^*\right\}.
\end{align*}

\begin{defc}[Routh decomposition]\label{Def:RouthDecomp}
  The decomposition
  \[
  I_{\text{con}}=\widetilde{I_{\text{con}}}\oplus I_{\g^*}
  \]
  for the contact subbundle will be called \emph{Routh decomposition} associated to the connection $\omega_Q$.
\end{defc}

\begin{note}\label{rem:RouthDecompLevelSet}
  Using Routh decomposition, we have that $\rho\in\left.W_L^\mu\right|_{\left(t,v_q\right)}$ if and only if
  \[
  \rho=L\left(t,v_q\right)dt+\widehat{\alpha}_{\left[q\right]}\circ Tp_G^Q\circ T_{v_q}\tau_Q-\widehat{\alpha}_{\left[q\right]}\circ Tp_G^Q\left(v_q\right)dt+\left<\mu,\left.\omega_{TQ}\right|_{v_q}-\omega_Q\left(v_q\right)dt\right>
  \]
  for some $\widehat{\alpha}_{\left[q\right]}\in T_{\left[q\right]}^*\left(Q/G\right)$. This fact will be useful in the proof of Theorem \ref{thm:TheorOnRouthDecomp}; namely, it can be written as
  \begin{align*}
    \rho&=\left[L\left(t,v_q\right)-\left<\mu,\omega_Q\left(v_q\right)\right>\right]dt+\\
    &\hspace*{2cm}+\widehat{\alpha}_{\left[q\right]}\circ Tp_G^Q\circ T_{v_q}\tau_Q-\widehat{\alpha}_{\left[q\right]}\circ Tp_G^Q\left(v_q\right)dt+\left<\mu,\left.\omega_{TQ}\right|_{v_q}\right>\\
    &=R_\mu\left(t,v_q\right)dt+\widehat{\alpha}_{\left[q\right]}\circ Tp_G^Q\circ T_{v_q}\tau_Q-\widehat{\alpha}_{\left[q\right]}\circ Tp_G^Q\left(v_q\right)dt+\left<\mu,\left.\omega_{TQ}\right|_{v_q}\right>
  \end{align*}
  so every element of $W_L^\mu$ is the sum of three terms: A term $R_\mu\left(t,v_q\right)dt$ involving the Routh function, a form in the contact bundle of $\mR\times T\left(Q/G\right)$, and the form $\left<\mu,\left.\omega_{TQ}\right|_{v_q}\right>$, which gives rise to gyroscopic forces. The first two terms can be related to elements in $W_{p_1^*\overline{R}_\mu}^0$; the third element induces a translation in the space of $1$-forms where $W_{p_1^*\overline{R}_\mu}^0$ lives. 
\end{note}

\subsection{Equations on $W_{p_1^*\overline{R}_\mu}^0$}
\label{sec:quat-w_p_1}

It remains to show the equivalence of mechanical systems associated to bundles
\[
W_{p_1^*\overline{R}_\mu}^0\rightarrow\mR\times T\left(Q/H\times\widetilde{\g}\right)\qquad\text{and}\qquad W_L^\mu\rightarrow\mR\times TQ.
\]
It will be done in the present section; the fact that $T\left(Q/H\times\widetilde{\g}\right)$ and $TQ$ are not directly related by a map must be overcome by means of a pullback bundle construction. It is worth to mention that the comparison between these affine bundles of forms requires an additional translation along a form related to the connection form; an interesting r\^ole in the proof is played by Routh decomposition.

\subsubsection{Comparing systems with a pullback bundle construction}
\label{sec:introduction}

We need to compare equations on $W_{p_1^*\overline{R_\mu}}^0$ with the equations of motion on $W_L^\mu$; in order to do it properly, let us define
\begin{equation}
  f_\omega:TQ\rightarrow Q/H\times\widetilde{\g}:v_q\mapsto\left(\left[q\right]_H,\left[q,\left.\omega_Q\right|_q\left(v_q\right)\right]_G\right).\label{eq:FOmegaDefinition}
\end{equation}
This allows us to construct the pullback bundle
\[
\begin{diagram}
  \node{\mR\times f_\omega^*\left(T\left(Q/H\times\widetilde{\g}\right)\right)}\arrow{s,l}{\text{pr}_1^\omega}\arrow{e,t}{\text{pr}_2^\omega}\node{\mR\times T\left(Q/H\times\widetilde{\g}\right)}\arrow{s,r}{\text{id}\times\tau_{Q/H\times\widetilde{\g}}}\\
  \node{\mR\times TQ}\arrow{e,b}{\text{id}\times f_\omega}\node{\mR\times Q/H\times\widetilde{\g}}
\end{diagram}
\]
Let us define
\[
F_\omega^{Q/H}:=f_\omega^*\left(T\left(Q/H\times\widetilde{\g}\right)\right).
\]
We can pullback bundles $W_L^\mu\rightarrow\mR\times TQ$ and $W_{p_1^*\overline{R}_\mu}^0\rightarrow\mR\times T\left(Q/H\times\widetilde{\g}\right)$ along projections $\text{pr}_i^\omega,i=1,2$; for every $\rho:=\left(t,v_q,W_{\left(\left[q\right]_H,\left[q,\xi\right]_G\right)}\right)\in\mR\times F_\omega^{Q/H}$, we have
\[
\left(\rho,\lambda\right)\in\left(\text{pr}_1^\omega\right)^*\left(W_L^\mu\right)\text{ if and only if }\lambda\in\left.W_L^\mu\right|_{\left(t,v_q\right)}
\]
and
\[
\left(\rho,\sigma\right)\in\left(\text{pr}_2^\omega\right)^*\left(W_{p_1^*\overline{R}_\mu}^0\right)\text{ if and only if }\sigma\in\left.W_{p_1^*\overline{R}_\mu}^0\right|_{\left(t,\left[q\right]_H,\left[q,\xi\right]_G\right)}.
\]
It means in particular that $\lambda\circ T_\rho\text{pr}_1^\omega$ and $\sigma\circ T_\rho\text{pr}_2^\omega$ are forms on $\mR\times F_\omega^{Q/H}$; thus we can consider these pullback bundles as subbundles of $T^*\left(\mR\times F_\omega^{Q/H}\right)$. Then we have the diagram
\[
\resizebox{.6\linewidth}{!}{
\divide\dgARROWLENGTH by2
$\begin{diagram}
  \node[2]{T^*\left(\mR\times F_\omega^{Q/H}\right)}\arrow[2]{s,r}{}\\
  \node{\left(\text{pr}_1^\omega\right)^*\left(W_L^\mu\right)}\arrow{se,t}{}\arrow{ne,t,J}{}\arrow{s,l}{}\node[2]{\left(\text{pr}_2^\omega\right)^*\left(W_{p_1^*\overline{R}_\mu}^0\right)}\arrow{sw,t}{}\arrow{nw,t,L}{}\arrow{s,r}{}\\
  \node{W_L^\mu}\arrow{s,l}{}\node{\mR\times F_\omega^{Q/H}}\arrow{sw,b}{\text{pr}_1^\omega}\arrow{se,b}{\text{pr}_2^\omega}\node{W_{p_1^*\overline{R}_\mu}^0}\arrow{s,r}{}\\
  \node{\mR\times TQ}\arrow{se,b}{\text{id}\times f_\omega}\node[2]{\mR\times T\left(Q/H\times\widetilde{\g}\right)}\arrow{sw,b}{\text{id}\times\tau_{Q/H\times\widetilde{\g}}}\\
  \node[2]{\mR\times Q/H\times\widetilde{\g}}
\end{diagram}
$}\]
where, using the identification mentioned before,
\begin{align*}
  \left.\left(\text{pr}_1^\omega\right)^*\left(W_L^\mu\right)\right|_{\rho}&=\left\{\gamma\circ T_{\rho}\text{pr}_1^\omega\in T^*_\rho\left(\mR\times F_\omega^{Q/H}\right):\gamma\in\left.W_L^\mu\right|_{\left(t,v_q\right)}\right\}\\
  \left.\left(\text{pr}_2^\omega\right)^*\left(W_{p_1^*\overline{R}_\mu}^0\right)\right|_{\rho}&=\left\{\sigma\circ T_{\rho}\text{pr}_2^\omega\in T^*_\rho\left(\mR\times F_\omega^{Q/H}\right):\sigma\in\left.W_{p_1^*\overline{R}_\mu}^0\right|_{\left(t,W_{\left(\left[q\right]_H,\left[q,\xi\right]\right)}\right)}\right\}.
\end{align*}
The maps
\begin{align*}
  &\Phi_L:\left(\text{pr}_1^\omega\right)^*\left(W_L^\mu\right)\rightarrow W_L^\mu:\gamma\circ T_\rho\text{pr}_1^\omega\mapsto\gamma,\\
  &\Phi_{p_1^*\overline{R}_\mu}^0:\left(\text{pr}_2^\omega\right)^*\left(W_{p_1^*\overline{R}_\mu}^0\right)\rightarrow W_{p_1^*\overline{R}_\mu}^0:\sigma\circ T_\rho\text{pr}_2^\omega\mapsto\sigma
\end{align*}
are well-defined, because $\text{pr}_i^\omega,i=1,2$ are surjective maps. Moreover, these maps have nice properties regarding the canonical structures on these spaces.
\begin{prop}\label{prop:PremultImmersion}
  Let $\lambda_L',\lambda_{p_1^*\overline{R}_\mu}'$ be the pullback of the canonical $1$-form on $$T^*\left(\mR\times F_\omega^{Q/H}\right)$$ to $\left(\text{pr}_1^\omega\right)^*\left(W_L^\mu\right)$ and $\left(\text{pr}_2^\omega\right)^*\left(W_{p_1^*\overline{R}_\mu}^0\right)$ respectively. Then
  \[
  \left(\Phi_{p_1^*\overline{R}_\mu}^0\right)^*\lambda_{p_1^*\overline{R}_\mu}^0=\lambda_{p_1^*\overline{R}_\mu}',\quad\Phi_L^*\lambda_L=\lambda_L'.
  \]
\end{prop}

\subsubsection{Routh reduction for Cartan-like systems}

We will relate equations in $W_{p_1^*\overline{R}_\mu}^0$ with equations in $W_L^\mu$. As we said above, it is necessary to compare the bundles supporting these equations in $\mR\times F_\omega^{Q/H}$. This is done in two stages:
\begin{itemize}
\item We will prove first that $\left(\text{pr}_1^\omega\right)^*\left(W_L^\mu\right)$ is a subbundle in $T^*\left(\mR\times F_\omega^{Q/H}\right)$ obtained from $\left(\text{pr}_2^\omega\right)^*\left(W_{p_1^*\overline{R}_\mu}^0\right)$ via a translation (in the sense of Proposition \ref{Prop:TranslationSpaceForms} and Corollary \ref{cor:AffineEquations}) along a suitable $1$-form related to connection $\omega_Q$, already chosen in Section \ref{sec:conn-decomp-tang}.
\item After that, the relation between the equations can be set by direct inspection.
\end{itemize}

Now, let us apply these considerations to our problem: We need to compare dynamics associated with bundle $W_{p_1^*\overline{R}_\mu}^0$ with the dynamics of the unreduced system $W_L^\mu$. This will be achieved using translations along a form associated to $\mu\in\g^*$ and the connection form $\omega_Q$ chosen in Section \ref{sec:conn-decomp-tang}; namely, let us define
\begin{equation}\label{eq:sigmamuDefntion}
  \omega_\mu:=\left<\mu,\omega_{Q}\right>\in\Omega^1\left(Q\right).
\end{equation}



Thus, we are ready to establish the main result of this section. From Proposition \ref{prop:Wp1LsimeqWNL} we know that
\[
\left.W_{p_1^*\overline{R}_\mu}^0\right|_{\left(t,W_{\left(\left[q\right]_H,\left[q,\xi\right]_G\right)}\right)}=\left\{\alpha\circ T_{\left(t,W_{\left(\left[q\right]_H,\left[q,\xi\right]_G\right)}\right)}p_1:\alpha\in\left.W_{\overline{R}_\mu}^{Q/G}\right|_{p_1\left(t,W_{\left(\left[q\right]_H,\left[q,\xi\right]_G\right)}\right)}\right\}.
\]
Recall from Equation \eqref{eq:JContactBundle} that
\[
W_{\overline{R}_\mu}^{Q/G}=\overline{R}_\mu dt+J_{\text{con}};
\]
additionally we have the commutative diagram \eqref{eq:DiagramCommutesPQ}, that in this case yields to
\[
\begin{diagram}
  \node{T\left(Q/H\times\overline{\g}\right)}\arrow{e,t}{p_1}\arrow{se,b}{T\phi^H}\node{T_{Q/H\times\widetilde{\g}}\left(Q/G\right)}\arrow{s,r}{p}\\
  \node[2]{T\left(Q/G\right)},
\end{diagram}
\]
so $\alpha'\in\left.W_{p_1^*\overline{R}_\mu}^0\right|_{\left(t,W_{\left(\left[q\right]_H,\left[q,\xi\right]_G\right)}\right)}$ if and only if (for clarity, we drop some indices regarding evaluation for tangent maps involved in the calculation)
\begin{align}\label{eq:AlphaPrimeWp1Rmu}
  \alpha'&=p_1^*\overline{R}_\mu dt+\widehat{\alpha}_{\left[q\right]}\circ T\tau_{Q/G}\circ T\left(p\circ p_1\right)-\widehat{\alpha}_{\left[q\right]}\circ T\left(p\circ p_1\right)\left(W_{\left(\left[q\right]_H,\left[q,\xi\right]_G\right)}\right)dt\cr
  &=p_1^*\overline{R}_\mu dt+\widehat{\alpha}_{\left[q\right]}\circ T\tau_{Q/G}\circ TT\phi^H-\widehat{\alpha}_{\left[q\right]}\circ TT\phi^H\left(W_{\left(\left[q\right]_H,\left[q,\xi\right]_G\right)}\right)dt.
\end{align}

\begin{thm}\label{thm:TheorOnRouthDecomp}
  With the notation introduced above,
  \[
  t_{\omega_\mu}\left(\left(\text{pr}_2^\omega\right)^*\left(W^0_{p_1^*\overline{R}_\mu}\right)\right)=\left(\text{pr}_1^\omega\right)^*\left(W^\mu_L\right).
  \]
\end{thm}
\begin{proof}
  Using Remark \ref{rem:RouthDecompLevelSet}, we see that any element
  \[
  \left(\rho,\alpha\circ T_\rho\text{pr}_1^\omega\right)\in\left(\text{pr}_1^\omega\right)^*\left(W_L^\mu\right)
  \]
  is such that
  \[
  \alpha=L\left(t,v_q\right)dt+\widehat{\alpha}_{\left[q\right]}\circ T_qp_G^Q\circ T_{v_q}\tau_Q-\widehat{\alpha}_{\left[q\right]}\circ T_qp_G^Q\left(v_q\right)dt+\left<\mu,\left.\omega_{TQ}\right|_{v_q}-\omega_Q\left(v_q\right)dt\right>
  \]
  for some $\widehat{\alpha}_{\left[q\right]}\in T_{\left[q\right]}^*\left(Q/G\right)$. This can be rearranged as
  \begin{equation}\label{eq:AlphaOnWLmu}
    \alpha=\left[L\left(t,v_q\right)-\left<\mu,\omega_Q\left(v_q\right)\right>\right]dt+\widehat{\alpha}_{\left[q\right]}\circ T_qp_G^Q\circ T_{v_q}\tau_Q-\widehat{\alpha}_{\left[q\right]}\circ T_qp_G^Q\left(v_q\right)dt+\left<\mu,\left.\omega_{TQ}\right|_{v_q}\right>.
  \end{equation}
  Now, consider the following diagram
  \[
  \begin{diagram}
    \node[2]{F_\omega^{Q/H}}\arrow{sw,t}{\text{pr}_1^\omega}\arrow{se,t}{\text{pr}_2^\omega}\\
    \node{TQ}\arrow{s,l}{\tau_Q}\arrow{se,t}{f_\omega}\node[2]{T\left(Q/H\times\widetilde{\g}\right)}\arrow{s,r}{T\phi^H}\arrow{sw,t}{\tau_{Q/H\times\widetilde{\g}}}\\
    \node{Q}\arrow{se,b}{p_G^Q}\node{Q/H\times\widetilde{\g}}\arrow{s,r}{\phi^H}\node{T\left(Q/G\right)}\arrow{sw,b}{\tau_{Q/G}}\\
    \node[2]{Q/G}
  \end{diagram}
  \]
  From Equation \eqref{eq:AlphaPrimeWp1Rmu}, we have that $\left(\rho,\alpha'\circ T_\rho\text{pr}_2^\omega\right)\in\left(\text{pr}_2^\omega\right)^*\left(W_{p_1^*\overline{R}_\mu}^0\right)$ if and only if
  \begin{align*}
    \alpha'&=p_1^*\overline{R}_\mu dt+\widehat{\alpha}_{\left[q\right]}\circ T\tau_{Q/G}\circ T\phi^H-\widehat{\alpha}_{\left[q\right]}\circ T\phi^H\left(W_{\left(\left[q\right]_H,\left[q,\xi\right]_G\right)}\right)dt,
  \end{align*}
  so
  \begin{equation}\label{AlphaPrimeOnWp1Rmu}
    \alpha'\circ T_\rho\text{pr}_2^\omega=p_1^*\overline{R}_\mu dt+\widehat{\alpha}_{\left[q\right]}\circ Tp_G^Q\circ T\tau_Q\circ T\text{pr}_1^\omega-\widehat{\alpha}_{\left[q\right]}\circ T\phi^H\left(W_{\left(\left[q\right]_H,\left[q,\xi\right]_G\right)}\right)dt.
  \end{equation}
  Finally, using the commutative diagram
  \[
  \begin{diagram}
    \node{F_\omega^{Q/H}}\arrow{e,t}{\text{pr}_2^\omega}\arrow{s,l}{\text{pr}_1^\omega}\node{T\left(Q/H\times\widetilde{\g}\right)}\arrow{s,r}{T\phi^H}\\
    \node{TQ}\arrow{e,b}{Tp_G^Q}\node{T\left(Q/G\right)}
  \end{diagram}
  \]
  it results that
  \[
  \widehat{\alpha}_{\left[q\right]}\circ T\phi^H\left(W_{\left(\left[q\right]_H,\left[q,\xi\right]_G\right)}\right)=\widehat{\alpha}_{\left[q\right]}\circ Tp_G^Q\left(v_q\right),
  \]
  and using it together with Equation \eqref{eq:sigmamuDefntion} in the comparison of Equation \eqref{AlphaPrimeOnWp1Rmu} with Equation \eqref{eq:AlphaOnWLmu}, we obtain the desired result.
\end{proof}

This theorem allows us to prove the following result, relating equations on $W_{p_1^*\overline{R}_\mu}^0$ and $W_L^\mu$.

\begin{cor}\label{cor:WLmuAndWp1Rmu0Coincide}
  Equations of motion on $W_{p_1^*\overline{R}_\mu}^0$ and $W_L^\mu$ coincide.
\end{cor}
\begin{proof}
  In sake of simplicity, we will prove this corollary in absence of forces terms; they can be restored in a straightforward manner. Let us take a curve $\Gamma:I\rightarrow\left(\text{pr}_1^\omega\right)^*\left(W_L^\mu\right)$ and a vector field $Z\in\mathfrak{X}\left(\left(\text{pr}_1^\omega\right)^*\left(W_L^\mu\right)\right)$ such that
  \[
  \Gamma^*\left(Z\lrcorner d\lambda_L'\right)=0.
  \]
  Then, using Proposition \ref{prop:PremultImmersion}, we will have that for $\Gamma_L:=\Phi_L\circ\Gamma$,
  \[
  \Gamma_L^*\left(\left(T\Phi_L\circ Z\right)\lrcorner d\lambda_L\right)=0.
  \]
  Now
  \[
  \left(\text{pr}_1^\omega\right)^*\left(W_L^\mu\right)=\left(\left(\text{pr}_2^\omega\right)^*\left(W^0_{p_1^*\overline{R}_\mu}\right)\right)_{\omega_\mu},
  \]
  so from Corollary \ref{cor:AffineEquations} it results that
  \[
  \Gamma_1:=t_{-\omega_\mu}\circ\Gamma:I\rightarrow\left(\text{pr}_2^\omega\right)^*\left(W^0_{p_1^*\overline{R}_\mu}\right)
  \]
  obeys the equation
  \[
  \Gamma_1^*\left(\left(Tt_{-\omega_\mu}\circ Z\right)\lrcorner\left(d\lambda_{p_1^*\overline{R}_\mu}'+\left(\pi_{p_1^*\overline{R}_\mu}'\right)^*d\omega_\mu\right)\right)=0,
  \]
  where $\pi_{p_1^*\overline{R}_\mu}':\left(\text{pr}_2^\omega\right)^*\left(W^0_{p_1^*\overline{R}_\mu}\right)\rightarrow\mR\times F_\omega^{Q/H}$ is the canonical projection.
  
  Recall now that $d\omega_\mu$ is basic for the projection $p_H^{Q}:Q\rightarrow Q/H$; then there exists $\beta^\mu\in\Omega^2\left(Q/H\right)$ such that
  \[
  \left(p_{H}^Q\right)^*\beta^\mu=d\omega_\mu.
  \]
  
  Therefore using again Proposition \ref{prop:PremultImmersion}, the map
  \[
  \Gamma':=\Phi_{p_1^*\overline{R}_\mu}^0\circ\Gamma:I\rightarrow W^0_{p_1^*\overline{R}_\mu}
  \]
  is a solution of the equation
  \[
  \left(\Gamma'\right)^*\left(\left(T\Phi_{p_1^*\overline{R}_\mu}^0\circ Tt_{-\omega_\mu}\circ Z\right)\lrcorner\left(d\lambda_{p_1^*\overline{R}_\mu}^0+\left(\pi_{p_1^*\overline{R}_\mu}^0\right)^*\beta^\mu\right)\right)=0.\qedhere
  \]
\end{proof}

\section{Reduced implicit Lagrange-Routh equations}
\label{sec:equat-moti-quas}

In the present section we will use the previous considerations in order to write the equations of motion for the system $\left(W_{p_1^*\overline{R}_\mu}^0,\lambda_{p_1^*\overline{R}_\mu}^0\right)$ in terms of quasicoordinates. It will allows us to compare them with the corresponding equations $\left(4.1\right)$ obtained in \cite{2015arXiv150901946G}. As before, throughout this section $H:=G_\mu$ indicates the isotropy group for $\mu\in\g^*$ regular value for the momentum map $\mu$; $\hf$ will be the Lie algebra associated to $H$.

\subsection{Gyroscopic force induced by connection $\omega_Q$}
\label{sec:gyrosc-form-induc}

We will calculate the gyroscopic force term determined on $T\left(Q/H\times\widetilde{\g}\right)$ by the connection defined on $Q\rightarrow Q/G$ by $\omega_Q$.

This connection induces in turn a connection on the bundle $\pi_\mu:Q/H\rightarrow Q/G$, when it is considered as an associated bundle for $Q$ through the bundle isomorphism
\[
Q\times _G\left(G/H\right)\simeq Q/H.
\]

Thus horizontal spaces on $Q/H$ are the projection along $p_H^Q:Q\rightarrow Q/H$ of the horizontal spaces on $Q$ associated to the connection $\omega_Q$. It means in particular that if $Z\in\mathfrak{X}\left(Q/G\right)$ and $Z^{H_Q}\in\mathfrak{X}\left(Q\right),Z^{H_{Q/H}}\in\mathfrak{X}\left(Q/H\right)$ indicate the horizontal lifts for these connections of $Z$ to $Q$ and $Q/H$ respectively, we will have that
\[
Z^{H_{Q/H}}=Tp_H^Q\circ Z^{H_Q}.
\]

Moreover, a similar identity can be set for infinitesimal generators
\[
\xi_{Q/H}=Tp_H^Q\circ\xi_Q,\qquad\xi\in\g
\]
associated to the action of $G$ on $Q$ and $Q/H$; using the fact that $G$ acts transitively on $G/H$, there exists $Z\in\mathfrak{X}\left(Q/G\right),\xi\in\g$ such that
\[
V=Z^{H_{Q/H}}\left(\left[q\right]_H\right)+\xi_{Q/H}\left(\left[q\right]_H\right).
\]
for every $V\in T_{\left[q\right]_H}\left(Q/H\right)$.

Following \cite{2010IJGMM..07.1451L}, let us consider the pullback bundle $\pi_\mu^*\widetilde{\g}=Q/H\times\widetilde{\g}$ and its subbundle $\widetilde{\hf}:=Q\times\hf/H$; then the quotient bundle $\pi_\mu^*\widetilde{\g}/\widetilde{\hf}$ is well-defined.

Now, let $V\in T_{\left[q\right]_H}\left(Q/H\right)$ be any vector on $Q/H$ and $v_q\in T_qQ$ such that
\[
T_qp_H^Q\left(v_q\right)=V;
\]
then
\[
\left(T_qp_H^Q\right)^{-1}\left(V\right)=\left\{v_q+\zeta_Q\left(q\right):\zeta\in\hf\right\},
\]
and so we can define the $\pi_\mu^*\widetilde{\g}/\widetilde{\hf}$-valued $1$-form $\widehat{\omega}$ via
\[
\left.\widehat{\omega}\right|_{\left[q\right]_H}\left(V\right):=\left[q,\left[\left.\omega_Q\right|_q\left(v_q\right)\right]_\hf\right]_G.
\]

It induces a correspondence
\[
T\left(Q/H\right)\simeq\pi_\mu^*T\left(Q/G\right)\times\pi_\mu^*\widetilde{\g}/\widetilde{\hf}
\]
via the map
\[
T_{\left[q\right]_H}\left(Q/H\right)\ni{V}\mapsto T_{\left[q\right]_H}\pi_\mu\left({V}\right)+\left.\widehat{\omega}\right|_{\left[q\right]_H}\left(V\right).
\]
Its inverse is given by
\[
\left(\left[q\right]_H,\widehat{v}_{\left[q\right]_G}\right)+\left(\left[q\right]_H,\left[q,\left[\xi\right]_\hf\right]_G\right)\mapsto\left(\widehat{v}_{\left[q\right]_G}\right)_{\left[q\right]_H}^{H_{Q/H}}+\xi_{Q/H}\left(\left[q\right]_H\right).
\]

Therefore, we are ready to find an expression for the $2$-form $\beta^\mu$, namely, for
\[
V_i=\left(\widehat{v}_{i}\right)_{\left[q\right]_H}^{H_{Q/H}}+\left(\xi_i\right)_{Q/H}\left(\left[q\right]_H\right),
\]
with $\widehat{v}_i\in T_{\left[q\right]_H}\left(Q/H\right),\xi_i\in\g$ and $i=1,2$, we obtain
\begin{align*}
  &\left.\beta^\mu\right|_{\left[q\right]_H}\left(V_1,V_2\right)=\\
  &\quad=\left.\beta^\mu\right|_{\left[q\right]_H}\left(\left(\widehat{v}_{1}\right)_{\left[q\right]_H}^{H_{Q/H}}+\left(\xi_1\right)_{Q/H}\left(\left[q\right]_H\right),\left(\widehat{v}_{2}\right)_{\left[q\right]_H}^{H_{Q/H}}+\left(\xi_2\right)_{Q/H}\left(\left[q\right]_H\right)\right)\\
  &\quad=\left.\beta^\mu\right|_{\left[q\right]_H}\circ T_qp_H^Q\left(\left(\widehat{v}_1\right)_q^{H_{Q}}+\left(\xi_1\right)_{Q}\left(q\right),\left(\widehat{v}_2\right)_q^{H_{Q}}+\left(\xi_2\right)_{Q}\left(q\right)\right)\\
  &\quad=\left.d\omega_\mu\right|_q\left(\left(\widehat{v}_1\right)_q^{H_{Q}}+\left(\xi_1\right)_{Q}\left(q\right),\left(\widehat{v}_2\right)_q^{H_{Q}}+\left(\xi_2\right)_{Q}\left(q\right)\right)\\
  &\quad=\left<\mu,\left.\Omega_Q\right|_q\left(\left(\widehat{v}_1\right)_q^{H_Q},\left(\widehat{v}_2\right)_q^{H_Q}\right)-\left[\xi_1,\xi_2\right]\right>.
\end{align*}
where $\Omega_Q$ is the curvature form for $\omega_Q$ on $Q$.

According to \cite{2010IJGMM..07.1451L}, we can define a map $\widetilde{\mu}:G/H\rightarrow\widetilde{\g}^*$ such that
\[
\left<\widetilde{\mu}\left(\left[q\right]_H\right),\left[q,\xi\right]_G\right>=\left<\mu,\xi\right>;
\]
the bracket on $\widetilde{\g}$ gives rise to a section of the bundle $\wedge^2\pi_\mu^*\widetilde{\g}^*/\widetilde{\hf}\rightarrow G/H$ via
\[
\left<\text{ad}^*\widetilde{\mu}\left(\left[q\right]_H\right),\left(\left[q,\left[\xi_1\right]_\hf\right]_G,\left[q,\left[\xi_2\right]_\hf\right]_G\right)\right>=\left<\mu,\left[\xi_1,\xi_2\right]\right>.
\]
Thus writing
\begin{equation}\label{eq:DecompVectorQoverH}
  V=V^h+V^v\in\pi_\mu^*T\left(Q/G\right)\times\pi_\mu^*\widetilde{\g}/\widetilde{\hf}
\end{equation}
we obtain
\begin{equation}\label{eq:BetaMuContractV}
  \left.V\lrcorner\beta_\mu\right|_{\left[q\right]_H}=\left(V^h\right)^{H_Q}_q\lrcorner\left.\Omega_Q\right|_q-V^v\lrcorner\text{ad}^*\widetilde{\mu}\left(\left[q\right]_H\right).
\end{equation}

\subsection{Considerations on the derivatives of the Routh function}
\label{sec:cons-deriv-routh}

It is our aim here to find the derivatives of the Routh function $p_1^*\overline{R}_\mu$ along vertical directions associated to $Q/H$-variables in $T\left(Q/H\times\widetilde{\g}\right)$.

The $1$-form $\omega_\mu$, defined in Equation \eqref{eq:sigmamuDefntion}, induces a fiberwise linear function $\overline{\omega_\mu}$ on $TQ$ closely related to the Routh function; in fact, we have that
\[
R_\mu-L=\overline{\omega_\mu}.
\]

Using relation
\[
\xi_{TQ}=\left(\xi_Q\right)^C,\qquad\xi\in\g
\]
for the infinitesimal generator of the $G$-action on $TQ$ and the complete lift of the corresponding action on $Q$, we have that
\begin{align}
  \xi_{TQ}\cdot\overline{\omega_\mu}&=\left(\xi_Q\right)^C\cdot\overline{\omega_\mu}\cr
  &=\overline{\left(\cL_{\xi_Q}\omega_\mu\right)}\cr
  &=\overline{\left<\mu,{\xi_Q}\lrcorner d\omega_Q\right>}\cr
  &=\overline{\left<\mu,\left[\xi,\omega_Q\right]\right>}.\label{eq:XiOmegaDerivative}
\end{align}

For $\xi\in\hf$, it means that $\overline{\omega_\mu}$ is a $H$-invariant function on $TQ$, thus the pullback of a function $\widehat{\omega_\mu}\in C^\infty\left(T\left(Q/G\right)\times Q/H\times\widetilde{\g}\right)$.

Moreover, if $v_q\in T_qQ$ is horizontal respect to the connection $\omega_Q$, we have that
\[
\overline{\omega_\mu}\left(v_q\right)=0,
\]
and so there exists $\widehat{\sigma_\mu}\in C^\infty\left(Q/H\times\widetilde{\g}\right)$ such that
\begin{equation}\label{eq:SigmaHatDefinition}
  \widehat{\sigma_\mu}\circ f_\omega=\overline{\omega_\mu},
\end{equation}
for $f_\omega:TQ\rightarrow Q/H\times\widetilde{\g}$ defined in Equation \eqref{eq:FOmegaDefinition}.

Now, the bundle $\pi_\mu:Q/H\times\widetilde{\g}\rightarrow Q/G$ can be endowed with a connection associated to $\omega_Q$, using the fact that $Q/H\times\widetilde{\g}$ is an associated bundle to the principal bundle $p_G^Q:Q\rightarrow Q/G$ and the $G$-space $G/H\times\g$.

\begin{lem}
  For every $Z\in\mathfrak{X}\left(Q/G\right)$, its horizontal lift $Z^{H_{Q/H\times\widetilde{\g}}}$ to $Q/H\times\widetilde{\g}$ is given by
  \[
  Z^{H_{Q/H\times\widetilde{\g}}}=Z^{H_{Q/H}}+Z^{H_{\widetilde{\g}}},
  \]
  where $Z^{H_{Q/H}}\in\mathfrak{X}\left(Q/H\right),Z^{H_{\widetilde{\g}}}\in\mathfrak{X}\left(\widetilde{\g}\right)$ are the horizontal lifts to every factor.
\end{lem}

In the following, $r_\mu:\widetilde{\g}\rightarrow Q/G$ indicates the adjoint bundle.

A basis of (local) vector fields on the bundle $\phi^H:Q/H\times\widetilde{\g}\rightarrow Q/G$ can be constructed using vertical vector fields on every factor and the canonically defined connection associated to $\omega_Q$.

\begin{prop}
  Let $\left\{Z_i\right\}$ be a (local) basis of vector fields on $Q/G$, $\left\{V_I\right\}$ a basis of sections for $V\pi_\mu$ and $\left\{W_a\right\}$ a basis of sections of $Vr_\mu$. Then
  \[
  \left\{Z_i^{H_{Q/G}}+Z_i^{H_{\widetilde{\g}}},V_I+0,0+W_a\right\}
  \]
  is a basis of vector fields on $Q/H\times\widetilde{\g}$.
\end{prop}

Now, let us consider the action on $\widehat{\sigma_\mu}$ of vector fields tangent to the factor $Q/H$ in the product $Q/H\times\widetilde{\g}$. To this end, we need the following result, which relates the projection of vector fields along $f_\omega$ with vector fields on $Q/H\times\widetilde{\g}$. 

In the next proof, for every $G$-space $X$, the map $\Phi_g^X:X\rightarrow X$ will indicate the diffeomorphism associated to the element $g\in G$. Moreover, for every $\zeta\in\g$, the symbol $\zeta_{\widetilde{\g}}\in\mathfrak{X}\left(\widetilde{\g}\right),\zeta\in\g$ will be the vector field
  \[
  \zeta_{\widetilde{\g}}\left(\left[q,\xi\right]_G\right):=\left.\frac{\vec{\text{d}}}{\text{d}t}\right|_{t=0}\left[q,\xi+t\zeta\right]_G
  \]
  associated to the linear structure of the bundle $\widetilde{\g}$.

\begin{lem}
  Let $\xi\in\g$ and $Z\in\mathfrak{X}\left(Q/G\right)$. Then
  \begin{align*}
    &Tf_\omega\circ\xi_{TQ}=\left(\xi_{Q/H}+0\right)\circ f_\omega\\
    &Tf_\omega\circ\left(Z^{H_{Q}}\right)^{C_Q}=\left(Z^{H_{Q/H}}+Z^{H_{\widetilde{\g}}}+\left(\overline{\Omega_Q\left(Z,\cdot\right)}\right)_{\widetilde{\g}}\right)\circ f_\omega,
  \end{align*}
  where $\left(\cdot\right)^{C_Q}$ indicates the complete lift of a vector field from $Q$ to $TQ$.
\end{lem}
\begin{proof}
  We have that
  \begin{align*}
    f_\omega\left(T\Phi_g^Q\left(v_q\right)\right)&=\left(\left[\Phi_g^Q\left(q\right)\right]_H,\left[\Phi_g\left(q\right),\text{Ad}_g\omega_Q\left(v_q\right)\right]_G\right)\\
    &=\left(\left[\Phi_g^Q\left(q\right)\right]_H,\left[q,\omega_Q\left(v_q\right)\right]_G\right),
  \end{align*}
  namely
  \[
  f_\omega\circ\Phi_g^{TQ}=\left(\Phi_g^{Q/H}\times\text{id}\right)\circ f_\omega.
  \]
  The infinitesimal counterpart of this equation becomes
  \[
  Tf_\omega\circ\xi_{TQ}=\left(\xi_{Q/H}+0\right)\circ f_\omega
  \]
  for all $\xi\in\g$.
  
  In order to prove the second identity, let us consider the following commutative diagram
  \[
  \begin{diagram}
    \node{TQ}\arrow{e,t}{\tau_Q}\arrow{s,l}{f_\omega}\node{Q}\arrow{s,r}{p_H^Q}\arrow{se,t}{p_G^Q}\\
    \node{Q/H\times\widetilde{\g}}\arrow{e,b}{\text{pr}_1}\node{Q/H}\arrow{e,b}{\pi_\mu}\node{Q/G}
  \end{diagram}
  \]
  Then for $Z\in\mathfrak{X}\left(Q/G\right)$, we have that
  \begin{align}
    Z^{H_{Q/H}}&=Tp_H^Q\circ Z^{H_Q}\cr
    &=Tp_H^Q\circ T\tau_Q\circ\left(Z^{H_Q}\right)^{C_Q}\cr
    &=T\text{pr}_1\circ Tf_\omega\circ\left(Z^{H_Q}\right)^{C_Q}.\label{eq:ZHCProj1}
  \end{align}
  
  On the other side, we have the identification $Vp_G^Q\simeq Q\times\g$ and $\omega_Q$ induces a vertical projection $\Pi_\omega:TQ\rightarrow Vp_G^Q$; namely, we have that
  \[
  \Pi_\omega\left(v_q\right):=\left(q,\omega_Q\left(v_q\right)\right).
  \]
  
  These maps can be integrated to the following diagram
  \[
  \begin{diagram}
    \node{TQ}\arrow{e,t}{\Pi_\omega}\arrow{s,l}{f_\omega}\node{Q\times\g}\arrow{s,r}{p_G^{Q\times\g}}\\
    \node{Q/H\times\widetilde{\g}}\arrow{e,b}{\text{pr}_2}\node{\widetilde{\g}}
  \end{diagram}
  \]
  
  If $\Phi_t:Q\rightarrow Q$ is the flow of the vector field $Z^{H_Q}$, then $T\Phi_t:TQ\rightarrow TQ$ is the corresponding flow for its complete lift $\left(Z^{H_Q}\right)^{C_Q}$; therefore
  \[
  \Pi_\omega\left(T\Phi_t\left(v_q\right)\right)=\left(\Phi_t\left(q\right),\omega_Q\left(T\Phi_t\left(v_q\right)\right)\right)
  \]
  and so
  \begin{align*}
    T\Pi_\omega\circ\left(Z^{H_Q}\right)^{C_Q}&=\left(Z^{H_Q},\overline{\cL_{Z^{H_Q}}\omega_Q}\right)\\
    &=\left(Z^{H_Q},\overline{{Z^{H_Q}}\lrcorner\Omega_Q}\right)\\
    &=\left(Z^{H_Q},0\right)+\left(0,\overline{{Z^{H_Q}}\lrcorner\Omega_Q}\right).
  \end{align*}

  Moreover, the connection in the associated space $\widetilde{\g}$ is defined by projection of the horizontal spaces of $Q$ along the map $p_G^{Q\times\g}$; therefore
  \begin{align}
    Z^{H_{\widetilde{\g}}}\circ f_\omega&=Tp_G^{Q\times\g}\circ\left(Z^{H_Q},0\right)\cr
    &=Tp_G^{Q\times\g}\circ\left[T\Pi_\omega\circ\left(Z^{H_Q}\right)^{C_Q}-\left(0,\overline{{Z^{H_Q}}\lrcorner\Omega_Q}\right)\right]\cr
    &=T\text{pr}_2\circ Tf_\omega\circ\left(Z^{H_Q}\right)^{C_Q}-\left(\overline{{Z^{H_Q}}\lrcorner\Omega_Q}\right)_{\widetilde{\g}}\circ f_\omega.\label{eq:ZHCProj2}
  \end{align}
  Using Equations \eqref{eq:ZHCProj1} and \eqref{eq:ZHCProj2} the second identity follows.
\end{proof}

Thus, from Equation \eqref{eq:SigmaHatDefinition} and using Equation \eqref{eq:XiOmegaDerivative},
\begin{align*}
  d\widehat{\sigma_\mu}\left(\xi_{Q/H}+0\right)\circ f_\omega&=d\widehat{\sigma_\mu}\left(Tf_\omega\circ\xi_{TQ}\right)\\
  &=d\overline{\omega_\mu}\left(\xi_{TQ}\right)\\
  &=\overline{\left<\mu,\left[\xi,\omega_Q\right]\right>}\\
  &=-\left<\text{ad}_{\overline{\omega_Q}}^*\mu,\xi\right>.
\end{align*}

\subsection{Reduced implicit Lagrange-Routh equations}
\label{sec:equations-motion}

We are ready to use Proposition \ref{prop:QuasiELequationsGeneral} in order to find the equations of motion of $\left(W^0_{p_1^*\overline{R}_\mu},\lambda^0_{p_1^*\overline{R}_\mu},\beta^\mu\right)$.

\begin{thm}
  The equations of motion of the triple $\left(W^0_{p_1^*\overline{R}_\mu},\lambda^0_{p_1^*\overline{R}_\mu},\beta^\mu\right)$ are given by
  \begin{align*}
    &\left(0+\zeta_{\widetilde{\g}}\right)\cdot\overline{R}_\mu=0,\qquad V^v-\left[\overline{\omega_Q}\right]_\hf=0,\qquad\overline{Z}-Z^{V_{Q/G}}\cdot\overline{R}_\mu=0,\\
    &d\overline{Z}-\left(\left(Z^{C_{Q/G}}+Z^{H_{Q/H}}+Z^{H_{\widetilde{\g}}}\right)\cdot\overline{R}_\mu+\left<\mu,\Omega_Q\left(\left(V^h\right)^{H_Q},Z^{H_Q}\right)\right>\right)dt=0,
  \end{align*}
  for $Z\in\mathfrak{X}\left(Q/G\right),\zeta\in\g$.
\end{thm}
\begin{note}
  We can relate this result with the reduced implicit Lagrange-Routh equations $\left(4.1\right)$ from \cite{2015arXiv150901946G}. Equation
  \[
  V^v-\left[\overline{\omega_Q}\right]_\hf=0
  \]
  is a global version of the reduced implicit equation
  \[
  \dot{\theta}^I=\widehat{v}^JL_J^I-\dot{x}^i\Lambda_i^I.
  \]
  
  The equation 
  \begin{equation*}
    \left(0+\zeta_{\widetilde{\g}}\right)\cdot\overline{R}_\mu=0.
  \end{equation*}
  corresponds to
  \[
  \frac{\partial R_\mu}{\partial\widehat{v}^a}=0.
  \]
  
  The remaining equations are
  \begin{align}
    &\overline{Z}-Z^{V_{Q/G}}\cdot\overline{R}_\mu=0\\
    &d\overline{Z}-\left(\left(Z^{C_{Q/G}}+Z^{H_{Q/H}}+Z^{H_{\widetilde{\g}}}\right)\cdot\overline{R}_\mu+\left<\mu,\Omega_Q\left(\left(V^h\right)^{H_Q},Z^{H_Q}\right)\right>\right)dt=0.
  \end{align}
  The first of them is equivalent to
  \[
  p_i=\frac{\partial R^\mu}{\partial v^i}
  \]
  and the last
  \[
  \dot{p}_i=\frac{\partial R^\mu}{\partial v^i}-\Lambda_i^I\frac{\partial R^\mu}{\partial\theta^I}-\mu_aB^a_{ij}\dot{x}^j
  \]
  in the previously cited work.
\end{note}
\begin{proof}
  We have to use Proposition \ref{prop:QuasiELequationsGeneral} with the vector fields $Z^{H_{Q/H}}+Z^{H_{\widetilde{\g}}},Z\in\mathfrak{X}\left(Q/G\right)$ and $\xi_{Q/H}+0,0+\zeta_{\widetilde{\g}}$ for $\xi,\zeta\in\g$. It yields to a variety of liftings
  \[
  \left(Z^{H_{Q/H}}+Z^{H_{\widetilde{\g}}}\right)^C,\left(Z^{H_{Q/H}}+Z^{H_{\widetilde{\g}}}\right)^V,\left(\xi_{Q/H}+0\right)^C,\left(0+\zeta_{\widetilde{\g}}\right)^C
  \]
  where, according to Lemma \ref{lem:VectorFieldsRelationThroughP1},
  \[
  \left(\xi_{Q/H}+0\right)^V,\left(0+\zeta_{\widetilde{\g}}\right)^V
  \]
  are vector fields spanning $\ker{Tp_1}$. Moreover
  \begin{align*}
    &Tp_1\circ\left(Z^{H_{Q/H}}+Z^{H_{\widetilde{\g}}}\right)^C=Z^{C_{Q/G}}+Z^{H_{Q/H}}+Z^{H_{\widetilde{\g}}},\\
    &Tp_1\circ\left(Z^{H_{Q/H}}+Z^{H_{\widetilde{\g}}}\right)^V=Z^{V_{Q/G}}+0+0,\\
    &Tp_1\circ\left(\xi_{Q/H}+0\right)^C=0+\xi_{Q/H}+0,\\
    &Tp_1\circ\left(0+\zeta_{\widetilde{\g}}\right)^C=0+0+\zeta_{\widetilde{\g}}.
  \end{align*}
  
  For vertical vector fields $\xi_{Q/H}+0$ and $0+\zeta_{\widetilde{\g}}$ we have that
  \[
  \left(\xi_{Q/H}+0\right)^V\cdot p_1^*\overline{R}_\mu=\left(0+\zeta_{\widetilde{\g}}\right)^V\cdot p_1^*\overline{R}_\mu=0,
  \]
  meaning that the associated momenta annihilate
  \[
  \overline{\xi_{Q/H}+0}=\overline{0+\zeta_{\widetilde{\g}}}=0.
  \]
  Moreover, gyroscopic force term $\beta^\mu$ is the pullback of a $2$-form on $Q/H$, so
  \[
  \left<0+\zeta_{\widetilde{\g}},V\lrcorner\beta^\mu\right>=0;
  \]
  thus vector field $0+\zeta_{\widetilde{\g}}$ gives rise to equation
  \begin{equation}
    \label{eq:ZetaGEquation}
    \left(0+\zeta_{\widetilde{\g}}\right)\cdot\overline{R}_\mu=0.
  \end{equation}
  
  For vector field $\xi_{Q/H}+0$, Equation \eqref{eq:BetaMuContractV} tells us that
  \[
  \left<\xi_{Q/H}+0,V\lrcorner\beta^\mu\right>=\left<V^v\lrcorner\text{ad}^*\widetilde{\mu},\xi_{Q/H}\right>
  \]
  and therefore
  \[
  \left(\xi_{Q/H}+0\right)\cdot\overline{R}_\mu+\left<V^v\lrcorner\text{ad}^*\widetilde{\mu},\xi_{Q/H}\right>=0.
  \]
  Using that $\overline{R}_\mu\circ p_H^{TQ}=R_\mu$ and
  \[
  Tp_H^{TQ}\circ\xi_{TQ}=\left(\xi_{Q/H}+0\right)\circ p_H^{TQ},
  \]
  we can write
  \begin{align*}
    \left(\xi_{Q/H}+0\right)\cdot\overline{R}_\mu&=\xi_{TQ}\cdot R_{\mu}\\
    &=\xi_{TQ}\cdot\left<\mu,\omega_Q\right>\\
    &=-\left<\text{ad}^*_{\overline{\omega_Q}}\mu,\xi\right>,
  \end{align*}
  taking into account the $G$-invariance of $L$ and Equation \eqref{eq:XiOmegaDerivative}. Then the associated equation results
  \begin{equation}
    \label{eq:XiQOverHEquation}
    V^v-\left[\overline{\omega_Q}\right]_\hf=0.
  \end{equation}
  
  The remaining equations, associated to horizontal lift $Z^{H_{Q/H}}+Z^{H_{\widetilde{\g}}}$ of $Z\in\mathfrak{X}\left(Q/G\right)$, become
  \begin{align}
    &\overline{Z}-Z^{V_{Q/G}}\cdot\overline{R}_\mu=0\\
    &d\overline{Z}-\left(\left(Z^{C_{Q/G}}+Z^{H_{Q/H}}+Z^{H_{\widetilde{\g}}}\right)\cdot\overline{R}_\mu+\left<\mu,\Omega_Q\left(\left(V^h\right)^{H_Q},Z^{H_Q}\right)\right>\right)dt=0.
  \end{align}
  This concludes the proof.
\end{proof}

\section{Lagrangian AKS and Routh reduction}
\label{sec:LagrangianAKSRouth}

Adler-Kostant-Symes (AKS) systems \cite{adler80:_compl_integ_system_euclid_lie,kostant79:_solut_to_gener_toda_lattic,symes80:_system_of_toda_type_inver} can be seen as reduced spaces via Marsden-Weinstein reduction \cite{Reyman:1979ru,Reyman:1981az}. In \cite{Feher200258} an \emph{ad hoc} Lagrangian version for this construction is given, motivated in the work of the same authors \cite{Feher:1992yx} in the context of Hamiltonian reduction in WZNW field theories. Specifically, let $K$ be a Lie group which factorises as $K=K_+K_-$. The authors choose as Lagrangian the function on $TK\times\kf_-\times\kf_+$
\begin{align}
  L_F\left(g,\dot{g},\alpha,\beta\right)&:=\frac{1}{2}\left<\dot{g}g^{-1},\dot{g}g^{-1}\right>+\frac{1}{2}\left<\alpha,\alpha\right>+\frac{1}{2}\left<\beta,\beta\right>+\cr
  &\qquad\qquad+\left<\alpha,\dot{g}g^{-1}-\mu\right>+\left<\beta,g^{-1}\dot{g}-\nu\right>+\left<\alpha,\Ad_{g}\beta\right>\cr
  &=\frac{1}{2}\left<\dot{g}g^{-1}+\alpha+\Ad_g\beta,\dot{g}g^{-1}+\alpha+\Ad_g\beta\right>-\left<\alpha,\mu\right>-\left<\beta,\nu\right>,\label{Eq:AKSFeherLagrangian}
\end{align}
where $\mu\in\kf_-,\nu\in\kf_+$ and $\left<\cdot,\cdot\right>$ is a nondegenerate $K$-invariant bilinear form on $\kf$. 

In the present section we will interpret these constructions by means of intrinsically constrained systems and Routh reduction; it is motivated in part by the fact that Routh reduction can be seen as Marsden-Weinstein reduction in the Lagrangian realm. In particular, this Lagrangian appears to be a Routh function \cite{Marsden00reductiontheory,2010IJGMM..07.1451L} associated to the $K_+\times K_-$-action on $K$, defined by
\[
\left(g_+,g_-\right)\cdot g=g_+gg_-^{-1}.
\]

\subsection{Unreduced system}
\label{sec:unreduced-system}

In fact, let us take $M:=K\times K_+\times K_-,N:=K$; consider $TK=K\times\kf$, $T\left(K\times K\right)=TK\times TK=K\times\kf\times K\times\kf$ and $TK_-=K_-\times\kf_-$ by right trivialization, and $TK_+=K_+\times\kf_+$ via left trivialization. It means that
\[
T_MN=TN\times_NM=K\times\kf\times K_+\times K_-.
\]

The map $\pi:M\rightarrow N$ will be projection onto the first component of the Cartesian product $M=K\times K_+\times K_-$; then
\[
p_1:TM\rightarrow T_MN:\left(g,\zeta,g_+,\alpha,g_-,\beta\right)\mapsto\left(g,\zeta,g_+,g_-\right).
\]

On $T_MN$ we take as Lagrangian the function
\[
L'\left(g,\zeta,g_+,g_-\right):=\frac{1}{2}\left<\zeta,\zeta\right>.
\]

The unreduced Lagrangian system for AKS system will be the intrinsically constrained system $\left(\pi:M\rightarrow N,L',0\right)$. 

\subsection{Equations of motion for unreduced system}
\label{sec:equat-moti-unred}

According to Definition \ref{Def:IntConstLag}, the equations of motion for intrinsically constrained system $\left(\pi:M\rightarrow N,L',0\right)$ are determined by Lagrangian system $\left(TM,p_1^*L',0\right)$. In this section we will use Proposition \ref{prop:QuasiELequations} in order to find them. It requires to construct a basis of vector fields on $M$; this is achieved by using invariant vector fields on the different Lie groups in it.

Let us consider the Lie group $K$, with identification $TK\simeq K\times\kf$ via right trivialization. For every $\xi\in\kf$, we have right invariant vector fields on $K$ given by
\[
X_\xi:g\mapsto\left(g,\xi\right).
\]
The flow for these vector fields are
\[
\Phi_t^\xi:g\mapsto\exp{t\xi}g;
\]
then
\[
T\Phi_t^\xi:\left(g,\zeta\right)\mapsto\left.\frac{\vec{\text{d}}}{\text{d}s}\right|_{s=0}\left[\Phi_t^\xi\left(\exp{s\zeta}g\right)\right]=\left(\exp{t\xi}g,\text{Ad}_{\exp{t\xi}}\zeta\right)
\]
is the flow for the complete lift. The flow for the vertical lift of these vector fields becomes
\[
\Psi_t^\xi:\left(g,\zeta\right)\mapsto\left(g,\zeta+t\xi\right).
\]
Then we have that
\begin{align*}
  X_\xi^V&:\left(g,\zeta\right)\mapsto\left(g,\zeta;0,\xi\right)\\
  X_\xi^C&:\left(g,\zeta\right)\mapsto\left(g,\zeta;\xi,\left[\xi,\zeta\right]\right)
\end{align*}
using again the right trivialization.
Now, we can fix a basis on $\kf$ and express any vector field in this basis; so from identity
\[
\left(f\xi\right)^C=f\xi^C+df\xi^V
\]
for every $f\in C^\infty\left(K\right)$, we obtain that a general vector field $X:g\mapsto\left(g,\xi\left(g\right)\right)$ on $K$ has the complete lift
\[
X^C:\left(g,\zeta\right)\mapsto\left(g,\zeta;\xi,d\xi+\left[\xi,\zeta\right]\right).
\]

These equations are valid for Lie group $K_-$ too. For $K_+$ we need to take into account that $TK_+=K_+\times\kf_+$ via left trivialization, so for left invariant vector fields
\[
Y_\xi:g\mapsto\left(g,\xi\right)
\]
for $\xi\in\kf_+$, we have the lifts
\begin{align*}
  Y_\xi^V&:\left(g,\zeta\right)\mapsto\left(g,\zeta;0,\xi\right)\\
  Y_\xi^C&:\left(g,\zeta\right)\mapsto\left(g,\zeta;\xi,-\left[\xi,\zeta\right]\right)
\end{align*}
and in general, for $Y:g\mapsto\left(g,\xi\left(g\right)\right)$, where $\xi:K_+\rightarrow\kf_+$,
\[
Y^C:\left(g,\zeta\right)\mapsto\left(g,\zeta;\xi,d\xi-\left[\xi,\zeta\right]\right).
\]

Let $w=\left(t,g,\zeta,g_+,\alpha,g_-,\beta\right)\in\mR\times TM$ be an arbitrary element; then $\rho\in\left.W_{p_1^*L'}\right|_w$ if and only if
\begin{multline}\label{Eq:ExpressionOnWLToda}
  \rho=L'\left(g,\zeta,g_+,g_-\right)dt+\left(g,\zeta;\sigma,0\right)-\sigma\left(\zeta\right)dt+\\
  +\left(g_+,\alpha;\rho_+,0\right)-\rho_+\left(\alpha\right)dt+\left(g,\beta;\rho_-,0\right)-\rho_-\left(\beta\right)dt
\end{multline}
for some $\sigma\in\kf,\rho_\pm\in\kf_\pm$. Here we are using the identification
\[
T^*TK\simeq K\times\kf\times\kf^*\times\kf^*,\qquad T^*TK_\pm\simeq K_\pm\times\kf_\pm\times\kf_\pm^*\times\kf_\pm^*
\]
using right trivialization for $K$ and $K_-$, and left trivialization for $K_+$. Then we have the isomorphism
\begin{equation}\label{eq:IsomorphismForWLPrime}
  \begin{diagram}
    \node{W_{p_1^*L'}}\arrow[2]{e,t}{\sim}\node[2]{\mR\times TM\times\kf^*\times\kf_+^*\times\kf_-^*}\\
    \node{\rho}\arrow[2]{e,t,T}{}\node[2]{\left(t,g,\zeta,g_+,\alpha,g_-,\beta,\sigma,\rho_+,\rho_-\right).}
  \end{diagram}
\end{equation}

Let $\zeta_1\in\kf,\alpha_1\in\kf_+,\beta_1\in\kf_-$ be arbitrary elements in these Lie algebras; let us indicate by $Z_{\zeta_1},Z_{\alpha_1},Z_{\beta_1}\in\mathfrak{X}\left(M\right)$ the vector fields
\begin{align*}
  Z_{\zeta_1}:&\left(g,g_+,g_-\right)\mapsto\left(g,\zeta_1,g_+,0,g_-,0\right)\\
  Z_{\alpha_1}:&\left(g,g_+,g_-\right)\mapsto\left(g,0,g_+,\alpha_1,g_-,0\right)\\
  Z_{\beta_1}:&\left(g,g_+,g_-\right)\mapsto\left(g,0,g_+,0,g_-,\beta_1\right).
\end{align*}

\begin{thm}
  Equations of motion for unreduced system $\left(M,p_1^*L',0\right)$ are given by
  \begin{align*}
    \sigma=\left<\zeta,\cdot\right>,&\quad\overline{Z}_{\zeta_1}=\sigma,\quad d\overline{Z}_{\zeta_1}=0, \\
      \dot{g}g^{-1}-\zeta=0,&\quad\dot{g_-}g_-^{-1}-\beta=0,\quad
  g_+^{-1}\dot{g_+}-\alpha=0.
  \end{align*}
  where $\overline{Z}_{\zeta_1}\in C^{\infty}\left(W_{p_1^*L'}\right)$ is the function associated to vector field $\overline{Z}_{\zeta_1}$. 
\end{thm}
\begin{proof}
  We have that
  \begin{align*}
    Z_{\zeta_1}^V\cdot p_1^*L'&=\left<\zeta,\zeta_1\right> \\
    Z_{\zeta_1}^C\cdot p_1^*L'&=\left<\zeta,\left[\zeta_1,\zeta\right]\right>=0
  \end{align*}
  with the remaining vector fields acting trivially on $p_1^*L'$. Then Proposition \ref{prop:QuasiELequations} gives
  \begin{align*}
    \sigma&=\left<\zeta,\cdot\right>\\
    \overline{Z}_{\zeta_1}&=\sigma \\
    d\overline{Z}_{\zeta_1}&=0
  \end{align*}
  together with
  \[
  \dot{g}g^{-1}-\zeta=0,\quad\dot{g_-}g_-^{-1}-\beta=0,\quad
  g_+^{-1}\dot{g_+}-\alpha=0.
  \]
  The theorem follows from here.
\end{proof}

\subsection{Symmetries of $\left(M,p_1^*L',0\right)$}
\label{sec:symmetries-leftm-l_1}

We are ready to discuss the symmetries of the Lagrangian system $\left(M,L_1,0\right)$. It results that $p_1^*L':TM\rightarrow\mR$ is invariant for the lifting of an action of the Cartesian product Lie group $K_+\times K_-$; this invariance is directly related with the $\text{Ad}-$invariance of the bilinear form $\left<\cdot,\cdot\right>$.

The direct product Lie group $G:=K_+\times K_-$ acts on $M$ according to the formula
\[
\Psi_{\left(h_+,h_-\right)}\left(g,g_+,g_-\right):=\left(h_+,h_-\right)\cdot\left(g,g_+,g_-\right)=\left(h_+gh_-^{-1},g_+h_+^{-1},h_-g_-\right).
\]

The lift of this equation to $TM$ reads
\begin{multline}\label{eq:LiftPsiActiontoTM}
  T\Psi_{\left(h_+,h_-\right)}:\left(g,\zeta,g_+,\alpha,g_-,\beta\right)\mapsto\\
  \mapsto\left(h_+gh_-^{-1},\text{Ad}_{h_+}\zeta,g_+h_+^{-1},\text{Ad}_{h_+}\alpha,h_-g_-,\text{Ad}_{h_-^{-1}}\beta\right).
\end{multline}

Additionally, let us recall that in the chosen trivializations
\begin{align*}
  T_{g_+}R_{h_+}\left(g_+,\alpha\right)&=\left(g_+h_+,\text{Ad}_{h_+^{-1}}\alpha\right)\\
  T_{g_-}L_{h_-}\left(g_-,\beta\right)&=\left(h_-g_-,\text{Ad}_{h_-}\beta\right).
\end{align*}

\begin{lem}
  The Lagrangian system $\left(M,p_1^*L',0\right)$ is $K_+\times K_-$-invariant.
\end{lem}
\begin{proof}
  From Equation \eqref{eq:LiftPsiActiontoTM} it results that
  \begin{align*}
    &p_1\circ
    T_{\left(g,g_+,g_-\right)}\Psi_{\left(h_+,h_-\right)}\left(g,\zeta,g_+,\alpha,g_-,\beta\right)=\\
    &=\Bigg(T\pi\left(h_+gh_-^{-1},\text{Ad}_{h_+}\zeta,g_+h_+^{-1},\text{Ad}_{h_+}\alpha,h_-g_-,\text{Ad}_{h_-}\beta\right),g_+h_+^{-1},h_-g_{-}\Bigg)\\
    &=\left(h_+gh_-^{-1},\text{Ad}_{h_+}\zeta,g_+h_+^{-1},h_-g_-\right);
  \end{align*}
  then
  \begin{align*}
    p_1^*L'\left(T_{\left(g,g_+,g_-\right)}\Psi_{\left(h_+,h_-\right)}\left(g,\zeta,g_+,\alpha,g_-,\beta\right)\right)&=L'\left(h_+gh_-^{-1},\text{Ad}_{h_+}\zeta,g_+h_+^{-1},h_-g_-\right) \\
    &=\frac{1}{2}\left<\text{Ad}_{h_+}\zeta,\text{Ad}_{h_+}\zeta\right>=\frac{1}{2}\left<\zeta,\zeta\right>
  \end{align*}
  and the invariance follows.
\end{proof}

\subsection{Routh reduction for $\left(M,p_1^*L',0\right)$}
\label{sec:routh-reduct-leftm-1}

Let us implement Routh reduction on this system. We have a $K_+\times K_-$-invariant Lagrangian system $\left(M,p_1^*L',0\right)$, where $M=K\times K_+\times K_-$; it is symmetric by the lift to $TM$ of the $K_+\times K_-$-action
\[
\left(h_+,h_-\right)\cdot\left(g,g_+,g_-\right)=\left(h_+gh_-^{-1},g_+h_+^{-1},h_-g_-\right).
\]
We can use diffeomorphism
\[
\chi_1:M/K_+\times K_-\rightarrow K:\left[g,g_+,g_-\right]\mapsto g_+gg_-
\]
and consider instead of projection $p_{K_+\times K_-}^M:M\rightarrow M/K_+\times K_-$ the simpler map
\[
p:M\rightarrow K:\left(g,g_+,g_-\right)\mapsto g_+gg_-.
\]
Thus we have the commutative diagram
\begin{equation}\label{eq:DiagramQuotientMK}
  \begin{diagram}
    \node{M}\arrow{s,l}{p_{K_+\times K_-}^M}\arrow{se,t}{p}\\
    \node{M/K_+\times K_-}\arrow{e,b}{\chi_1}\node{K}
  \end{diagram}
\end{equation}

In terms of the trivializations adopted in this example its differential reads
\[
\begin{diagram}
  \node{K\times\kf\times K_+\times\kf_+\times K_-\times\kf_-}\arrow[2]{e,t}{Tp}\node[2]{K\times\kf}\\
  \node{\left(g,\zeta,g_+,\alpha,g_-,\beta\right)}\arrow{e,b,T}{}\node[2]{\left(g_+gg_-,\text{Ad}_{g_+}\left(\zeta+\alpha+\text{Ad}_g\beta\right)\right)}
\end{diagram}
\]

For $\left(\alpha,\beta\right)\in\kf_+\times\kf_-$, we have a vector field $\left(\alpha,\beta\right)_M$ on $M$, namely, the infinitesimal generator for the $K_+\times K_-$-action, given by
\begin{align}
  \left(\alpha,\beta\right)_M\left(g,g_+,g_-\right)&=\left.\frac{\vec{\text{d}}}{\text{d}t}\right|_{t=0}\left(\exp{t\alpha}\cdot g\cdot\exp{-t\beta},g_+\cdot\exp{-t\alpha},\exp{t\beta}\cdot g_-\right)\cr
  &=\left(g,\alpha-\text{Ad}_g\beta,g_+,-\alpha,g_-,\beta\right).\label{eq:InfGenOnProduct}
\end{align}

Map $p:M\rightarrow K$ gives rise to a $K_+\times K_-$-principal bundle structure on $K$; it allows to select a connection on this, which will be useful in performing Routh reduction of $\left(M,p_1^*L',0\right)$.

By means of Diagram \eqref{eq:DiagramQuotientMK} and right trivialization, we have the identification
\[
T\left(\frac{M}{K_+\times K_-}\right)\simeq K\times\kf;
\]
as expected, projection $Tp^M_{K_+\times K_-}$ is thus replaced by $Tp$. 

\begin{lem}\label{lem:ConnAndHorLift}
  The $\kf_+\times\kf_-$-valued $1$-form $\omega$ such that
  \[
  \left.\omega\right|_{\left(g,g_+,g_-\right)}\left(\zeta,\alpha,\beta\right):=\left(-\alpha,\beta\right)
  \]
  is a connection form on principal bundle $p:M\rightarrow K$. Its differential is given by
  \[
  \left.d\omega\right|_{\left(g,g_+,g_-\right)}\left(\zeta_1,\alpha_1,\beta_1;\zeta_2,\alpha_2,\beta_2\right)=\left(\left[\alpha_1,\alpha_2\right],-\left[\beta_1,\beta_2\right]\right).
  \]

  The horizontal lift associated to the connection form $\omega$ is given by the formula
  \[
  \left.\left(g',\zeta'\right)^H\right|_{\left(g^{-1}_+g'g^{-1}_-,g_+,g_-\right)}=\left(g_+^{-1}g'g^{-1}_-,\text{Ad}_{g_+^{-1}}\zeta',g_+,0,g_-,0\right).
  \]
  for every $\left(g',\zeta'\right)\in T_{g'}K$ and $\left(g^{-1}_+g'g^{-1}_-,g_+,g_-\right)\in p^{-1}\left(g'\right)$.
\end{lem}
\begin{proof}
  First, we need to verify that
  \begin{itemize}
  \item As above, $\Psi_{\left(h_+,h_-\right)}$ indicates the diffeomorphism of $M$ associated to the action of element $\left(h_+,h_-\right)\in K_+\times K_-$; then
    \[
    \Psi_{\left(h_+,h_-\right)}^*\left(\left.\omega\right|_{\left(h_+gh_-^{-1},g_+h_+^{-1},h_-g_-\right)}\right)=\text{Ad}_{\left(h_+,h_-\right)}\left.\omega\right|_{\left(g,g_+g,_-\right)}.
    \]
  \item For every $\left(\alpha,\beta\right)\in\kf_+\times\kf_-$, we have
    \[
    \left.\omega\right|_{\left(g,g_+,g_-\right)}\left(\left(\alpha,\beta\right)_M\right)=\left(\alpha,\beta\right).
    \]
  \end{itemize}
  For first item, recall Eq. \eqref{eq:LiftPsiActiontoTM} and that by the product group structure on $K_+\times K_-$,
  \[
  \text{Ad}_{\left(h_+,h_-\right)}\left(\alpha_1,\beta_1\right)=\left(\text{Ad}_{h_+}\alpha_1,\text{Ad}_{h_-}\beta_1\right)
  \]
  for every $\left(\alpha_1,\beta_1\right)\in\kf_+\times\kf_-$.
  
  For second item, just use Equation \eqref{eq:InfGenOnProduct}.
  
  Now we proceed to prove the horizontal lift formula.   Namely, we know that horizontal lift of tangent vectors
  $\left(g',\zeta'\right)\in T_{g'}K$ to
  $$\left(g^{-1}_+g'g^{-1}_-,g_+,g_-\right)\in p^{-1}\left(g'\right)$$ is given by
  \[
  \left.\left(g',\zeta'\right)^H\right|_{\left(g^{-1}_+g'g^{-1}_-,g_+,g_-\right)}=\left(g_+^{-1}g'g_-^{-1},\zeta_1,g_+,\alpha_1,g_-,\beta_1\right)
  \]
  if and only if
  \[
  Tp\left(\left.\left(g',\zeta'\right)^H\right|_{\left(g^{-1}_+g'g^{-1}_-,g_+,g_-\right)}\right)=\left(g',\zeta'\right)
  \]
  and
  \[
  \left.\omega\right|_{\left(g^{-1}_+g'g^{-1}_-,g_+,g_-\right)}\left(\left(g',\zeta'\right)^H\right)=0.
  \]
  It means that
  \begin{align*}
    \left(g',\zeta'\right)&=\left(g',\text{Ad}_{g_+}\left(\zeta_1+\alpha_1+\text{Ad}\beta_1\right)\right)\\
    \left(-\alpha_1,\beta_1\right)&=\left(0,0\right);
  \end{align*}
  therefore
  \[
  \left.\left(g',\zeta'\right)^H\right|_{\left(g^{-1}_+g'g^{-1}_-,g_+,g_-\right)}=\left(g_+^{-1}g'g^{-1}_-,\text{Ad}_{g_+^{-1}}\zeta',g_+,0,g_-,0\right)
  \]
  as required.
\end{proof}

There are two quotient bundles which we need to handle in order to work with the reduced system, namely, the adjoint bundle
\[
\widetilde{\kf_+\times\kf_-}:=\frac{M\times\kf_+\times\kf_-}{K_+\times K_-}
\]
and the quotient
\[
\overline{p}_{\left(\mu,\nu\right)}:\frac{M}{\left(K_+\right)_\mu\times\left(K_-\right)_\nu}\rightarrow\frac{M}{K_+\times K_-}.
\]

Now, every element
\[
\left[g,g_+,g_-,\alpha,\beta\right]_{K_+\times K_-}\in\widetilde{\kf_+\times\kf_-}
\]
is an equivalence class
\begin{multline*}
  \left[g,g_+,g_-,\alpha,\beta\right]_{K_+\times K_-}:=\\
  =\left\{\left(h_+g_+h_-^{-1},g_+h_+^{-1},h_-g_-,\text{Ad}_{h_+}\alpha,\text{Ad}_{h_-^{-1}}\beta\right):h_+\in K_+,h_-\in K_-\right\}.
\end{multline*}

Using the following diagram
\begin{center}
  \begin{tikzpicture}
    \matrix (m) [matrix of math nodes, row sep=3em, column sep=3em,
    text height=3.5ex, text depth=1.25ex] { M\times\kf_+\times\kf_- & M & \\ \widetilde{\kf_+\times\kf_-} & M/K_+\times K_- & K \\}; 
    \path[>=latex,->]
    (m-1-1) edge node[above] {$ \text{pr}_1 $} (m-1-2) 
            edge node[left] {$ p_{K_+\times K_-}^{M\times\kf_+\times\kf_-} $} (m-2-1) 
    (m-1-2) edge node[above] {$ p $} (m-2-3)
            edge node[left] {$ p_{K_+\times K_-}^M $} (m-2-2)
    (m-2-1) edge (m-2-2)
    (m-2-2) edge (m-2-3);
    \path[dashed,->] 
    (m-2-1) edge [bend right=25] node[below] {$ p' $} (m-2-3);
  \end{tikzpicture}
\end{center}
we can consider $\widetilde{\kf_+\times\kf_-}$ as a bundle on $K$, with projection $p':\widetilde{\kf_+\times\kf_-}\rightarrow K$ given by the composition of the lower horizontal arrows, namely
\[
p'\left(\left[g,g_+,g_-,\alpha,\beta\right]_{K_+\times K_-}\right)=g_+gg_-.
\]

There exists another bundle isomorphism $\chi_2:\widetilde{\kf_+\times\kf_-}\rightarrow K\times\kf_+\times\kf_-$ such that
\[
\begin{diagram}
  \node{\widetilde{\kf_+\times\kf_-}}\arrow[2]{e,t}{\chi_2}\arrow{se,b}{p'}\node[2]{K\times\kf_+\times\kf_-}\arrow{sw,b}{\text{pr}_1}\\
  \node[2]{K}
\end{diagram}
\]
It is given by
\[
\chi_2:\left[g,g_+,g_-,\alpha,\beta\right]_{K_+\times K_-}\longmapsto\left(g_+gg_-,\text{Ad}_{g_+}\alpha,\text{Ad}_{g_-^{-1}}\beta\right).
\]

Now let us consider the quotient bundle $M/\left(K_+\right)_\mu\times\left(K_-\right)_\nu$. In order to work with it, fix a pair of elements $\mu\in\kf_+^*,\nu\in\kf_-^*$, and indicate with $\mathcal{O}_\mu^+\subset\kf_+^*,\mathcal{O}_\nu^-\subset\kf_-^*$ the coadjoint orbits through them.

Let us indicate by $\left[g,g_+,g_-\right]_{\left(\mu,\nu\right)}$ an equivalence class in $M/\left(K_+\right)_\mu\times\left(K_-\right)_\nu$; then we have a map
\[
\chi_3:M/\left(K_+\right)_\mu\times\left(K_-\right)_\nu\rightarrow K\times\mathcal{O}_\mu^+\times\mathcal{O}_\nu^-:\left[g,g_+,g_-\right]_{\left(\mu,\nu\right)}\mapsto\left(g_+gg_-,\text{Ad}_{g_+^{-1}}^*\mu,\text{Ad}_{g_-}^*\nu\right)
\]
so that the following diagram commutes
\[
\begin{diagram}
  \node{M/\left(K_+\right)_\mu\times\left(K_-\right)_\nu}\arrow{e,t}{\chi_3}\arrow{s,l}{\overline{p}_{\left(\mu,\nu\right)}}\node{K\times\mathcal{O}_\mu^+\times\mathcal{O}_\nu^-}\arrow{s,r}{\text{pr}_1}\\
  \node{M/K_+\times K_-}\arrow{e,b}{\chi_1}\node{K}
\end{diagram}
\]
It is an isomorphism of bundles on $\chi_1$.

\begin{thm}
  The map
  \begin{multline*}
    T\chi_1\times\chi_3\times\chi_2:\\
    T\left(\frac{M}{K_+\times K_-}\right)\times\frac{M}{\left(K_+\right)_\mu\times\left(K_-\right)_\nu}\times\widetilde{\kf_+\times\kf_-}\rightarrow K\times\kf\times\mathcal{O}_\mu^+\times\mathcal{O}_\nu^-\times{\kf_+\times\kf_-} 
  \end{multline*}
is an isomorphism of bundles on $\chi_1:M/K_+\times K_-\rightarrow K$.
\end{thm}
\begin{proof}
  These maps fit in the following diagram
  \[
  \begin{diagram}
    \node{T\left(\frac{M}{K_+\times K_-}\right)}\arrow{s,l}{\tau_{{M}/{K_+\times K_-}}}\arrow{e,t}{T\chi_1}\node{K\times\kf}\arrow{s,r}{\text{pr}_1}\node{\widetilde{\kf_+\times\kf_-}}\arrow{sw,t}{p'}\arrow{s,r}{\chi_2}\\
    \node{\frac{M}{K_+\times K_-}}\arrow{e,b}{\chi_1}\node{K}\node{K\times\kf_+\times\kf_-}\arrow{w,b}{\text{pr}_1}\\
    \node{\frac{M}{\left(K_+\right)_\mu\times\left(K_-\right)_\nu}}\arrow{n,l}{\overline{p}_{\left(\mu,\nu\right)}}\arrow{e,b}{\chi_3}\node{K\times\mathcal{O}_\mu^+\times\mathcal{O}_\nu^-}\arrow{n,r}{\text{pr}_1}
  \end{diagram}
  \]
  where it was used identification $TK=K\times\kf$ with right trivialization.
\end{proof}

\begin{prop}
  Reduced Lagrangian is given by
  \[
  l'\left(g',\zeta',\text{Ad}_{g_+^{-1}}^*\mu,\text{Ad}_{g_-}^*\nu,\widetilde{\alpha},\widetilde{\beta}\right)=\frac{1}{2}\left<\zeta'+\widetilde{\alpha}-\text{Ad}_{g'}\widetilde{\beta},\zeta'+\widetilde{\alpha}-\text{Ad}_{g'}\widetilde{\beta}\right>
  \]
  for any $\left(g',\zeta',\text{Ad}_{g_+^{-1}}^*\mu,\text{Ad}_{g_-}^*\nu,\widetilde{\alpha},\widetilde{\beta}\right)\in K\times\kf\times\mathcal{O}_\mu^+\times\mathcal{O}_\nu^-\times\kf_+\times\kf_-$.
\end{prop}
\begin{proof}
  We have that
  \[
  \left(g_+^{-1}g'g_-^{-1},g_+,g_-\right)\in\left(p_{\left(K_+\right)_\mu\times\left(K_-\right)_\nu}^M\right)^{-1}\left(g',\text{Ad}_{g_+^{-1}}^*\mu,\text{Ad}_{g_-}^*\nu\right)
  \]
  indicates an arbitrary element in this fiber. Moreover, any element
  of $M\times\kf_+\times\kf_-$ belonging to this fiber and projecting
  onto
  $\left[g',\widetilde{\alpha},\widetilde{\beta}\right]\in\widetilde{\kf_+\times\kf_-}$
  is of the form
  \[
  \left(g_+^{-1}g'g_-^{-1},g_+,g_-,\text{Ad}_{g_+^{-1}}\widetilde{\alpha},\text{Ad}_{g_-}\widetilde{\beta}\right)\in
  M\times\kf_+\times\kf_-.
  \]

  Using horizontal lifting calculated in Lemma
  \ref{lem:ConnAndHorLift} and expression \eqref{eq:InfGenOnProduct}
  for infinitesimal generator for the $K_+\times K_-$-action on $M$,
  we can obtain reduced Lagrangian
  \[
  l'\in
  C^\infty\left(K\times\kf\times\mathcal{O}_\mu^+\times\mathcal{O}_\nu^-\times{\kf_+\times\kf_-}\right)
  \]
  using the following formula
  \begin{multline*}
    l'\left(g',\zeta',\text{Ad}_{g_+^{-1}}^*\mu,\text{Ad}_{g_-}^*\nu,\widetilde{\alpha},\widetilde{\beta}\right)=\\
    =L'\left(\left.\left(g',\zeta'\right)^H\right|_{\left(g_+^{-1}g'g_-^{-1},g_+,g_-\right)}+\left(\text{Ad}_{g_+^{-1}}\widetilde{\alpha},\text{Ad}_{g_-}\widetilde{\beta}\right)_M\left(g_+^{-1}g'g_-^{-1},g_+,g_-\right)\right).
  \end{multline*}
  Now
  \begin{multline}
    \left.\left(g',\zeta'\right)^H\right|_{\left(g_+^{-1}g'g_-^{-1},g_+,g_-\right)}+\left(\text{Ad}_{g_+^{-1}}\widetilde{\alpha},\text{Ad}_{g_-}\widetilde{\beta}\right)_M\left(g_+^{-1}g'g_-^{-1},g_+,g_-\right)=\\
    =\left(g_+^{-1}g'g_-^{-1},\text{Ad}_{g_+^{-1}}\zeta'+\text{Ad}_{g_+^{-1}}\widetilde{\alpha}-\text{Ad}_{g_+^{-1}g'g_-^{-1}}\text{Ad}_{g_-}\widetilde{\beta},g_+,-\text{Ad}_{g_+^{-1}}\widetilde{\alpha},g_-,\text{Ad}_{g_-}\widetilde{\beta}\right),\label{eq:InverseMappingConn}
  \end{multline}
  so the formula follows from here.
\end{proof}

Recall that the Routhian is given by formula
\[
R_{\left(\mu,\nu\right)}:=L-\left<\left(\mu,\nu\right),\omega\right>.
\]
We have calculated the reduced Lagrangian $l'$, so we need to take care only of the term containing $\omega$. Recalling Equation \eqref{eq:InverseMappingConn}, it results that
\begin{multline*}
  \left<\left(\mu,\nu\right),\omega\right>\left(\left.\left(g',\zeta'\right)^H\right|_{\left(g_+^{-1}g'g_-^{-1},g_+,g_-\right)}+\left(\text{Ad}_{g_+^{-1}}\widetilde{\alpha},\text{Ad}_{g_-}\widetilde{\beta}\right)_M\left(g_+^{-1}g'g_-^{-1},g_+,g_-\right)\right)=\\
  =\left<\mu,\text{Ad}_{g_+^{-1}}\widetilde{\alpha}\right>+\left<\nu,\text{Ad}_{g_-}\widetilde{\beta}\right>,
\end{multline*}
so we obtain the formula
\begin{multline}
  \overline{R}_{\left(\mu,\nu\right)}\left(g',\zeta',\text{Ad}_{g_+^{-1}}^*\mu,\text{Ad}_{g_-}^*\nu,\widetilde{\alpha},\widetilde{\beta}\right)=\\
  =\frac{1}{2}\left<\zeta'+\widetilde{\alpha}-\text{Ad}_{g'}\widetilde{\beta},\zeta'+\widetilde{\alpha}-\text{Ad}_{g'}\widetilde{\beta}\right>-\left<\text{Ad}_{g_+^{-1}}^*\mu,\widetilde{\alpha}\right>-\left<\text{Ad}_{g_-}^*\nu,\widetilde{\beta}\right>
\end{multline}
for the reduced version of Routhian function.

This function defines a current Lagrangian function on $T\left(K\times\mathcal{O}_\mu^+\times\mathcal{O}_\nu^-\times\kf_+\times\kf_-\right)$ via pull back along a map
\[
\pi_1:T\left(K\times\mathcal{O}_\mu^+\times\mathcal{O}_\nu^-\times\kf_+\times\kf_-\right)\rightarrow K\times\kf\times\mathcal{O}_\mu^+\times\mathcal{O}_\nu^-\times\kf_+\times\kf_-.
\]

This map is defined as follows: In terms of the original spaces, fix an arbitrary element
\[
w:=\left(\left[g,g_+,g_-\right],\left[g,\alpha,\beta\right]_{K_+\times K_-}\right)\in\frac{M}{\left(K_+\right)_\mu\times\left(K_-\right)_\nu}\times\widetilde{\kf_+\times\kf_-};
\]
then it reads
\[
\pi_1\left(V_{w}\right):=\left(T\overline{p}_{\left(\mu,\nu\right)}\left(V_w\right),w\right)
\]
for every
\[
V_w\in T_w\left(\frac{M}{\left(K_+\right)_\mu\times\left(K_-\right)_\nu}\times\widetilde{\kf_+\times\kf_-}\right).
\]

Using the isomorphisms defined above, it simplifies to
\[
\pi_1\left(g',\zeta',\eta_+,u_{\eta_+},\eta_-,v_{\eta_-},\widetilde{\alpha},\widetilde{\alpha}_1,\widetilde{\beta},\widetilde{\beta}_1\right)=\left(g',\zeta',\eta_+,\eta_-,\widetilde{\alpha},\widetilde{\beta}\right).
\]

Thus we have a singular Lagrangian
\[
L_{\left(\mu,\nu\right)}:=\pi_1^*\overline{R}_{\left(\mu,\nu\right)}.
\]
Additionally, the reduced Lagrangian system requires the force term arising from the differential of connection form $\omega$; this force term $f_{\left(\mu,\nu\right)}$ is a bundle map 
\[
f_{\left(\mu,\nu\right)}:T\left(\frac{M}{\left(K_+\right)_\mu\times\left(K_-\right)_\nu}\right)\rightarrow T^*\left(\frac{M}{\left(K_+\right)_\mu\times\left(K_-\right)_\nu}\right)
\]
associated to the $2$-form on ${M}/{\left(K_+\right)_\mu\times\left(K_-\right)_\nu}$ induced by $\left<\left(\mu,\nu\right),d\omega\right>$. Using Lemma \ref{lem:ConnAndHorLift}, it becomes
\begin{multline}\label{eq:ForceTermConnection}
  \left<f_{\left(\mu,\nu\right)}\left(g',\zeta_1',\eta_+,u_{\eta_+}^1,\eta_-,v_{\eta_-}^1\right),\left(g',\zeta_2',\eta_+,u_{\eta_+}^2,\eta_-,v_{\eta_-}^2\right)\right>=\\
  =\left<\mu,\left[\widehat{u_{\eta_+}^1},\widehat{u_{\eta_+}^2}\right]\right>-\left<\nu,\left[\widehat{v_{\eta_-}^1},\widehat{v_{\eta_-}^2}\right]\right>,
\end{multline}
where $\widehat{u_{\eta_+}^i}\in\kf_+,\widehat{v_{\eta_-}^i}\in\kf_-,i=1,2$ are Lie algebra elements such that
\[
Tp_{\left(K_+\right)_\mu}^{K_+}\left(g_+,\widehat{u_{\eta_+}^i}\right)=\left(\eta_+,u_{\eta_+}^i\right),\qquad Tp_{\left(K_-\right)_\nu}^{K_-}\left(g_-,\widehat{v_{\eta_-}^i}\right)=\left(\eta_-,v_{\eta_-}^i\right).
\]

\subsection{Routh reduction and Fehér Lagrangian}
\label{sec:routh-reduct-feher}

Our aim is to relate Lagrangian $L_{\left(\mu,\nu\right)}$ with Fehér Lagrangian \eqref{Eq:AKSFeherLagrangian}, 
\begin{align*}
  L_F\left(g',\zeta',\eta_+,\eta_-,\alpha,\beta\right)&:=\frac{1}{2}\left<\zeta',\zeta'\right>+\frac{1}{2}\left<\alpha,\alpha\right>+\frac{1}{2}\left<\beta,\beta\right>+\cr
  &\qquad\qquad+\left<\alpha,\zeta'-\mu\right>+\left<\beta,\zeta'-\nu\right>+\left<\alpha,\Ad_{g'}\beta\right>\cr
  &=\frac{1}{2}\left<\zeta'+\alpha+\Ad_{g'}\beta,\zeta'+\alpha+\Ad_{g'}\beta\right>-\left<\mu,\alpha\right>-\left<\nu,\beta\right>.
\end{align*}

The main tool in this task will be Proposition \ref{Prop:DiffeomSolutions}; in order to do it, we will need to define a bundle isomorphism
\[
\begin{diagram}
  \node{TM_1}\arrow{e,t}{\Phi}\arrow{s,l}{\tau_{M_1}}\node{TM_1}\arrow{s,r}{\tau_{M_1}}\\
  \node{M_1}\arrow{e,b}{\phi}\node{M_1}
\end{diagram}
\]
and to prove that together with $L_{\left(\mu,\nu\right)}$ and $L_F$, they meet the conditions of this result.

In order to define these maps, we will fix a pair of (perhaps local) sections
\begin{align*}
  &s_+:\mathcal{O}_\mu^+\rightarrow K_+\\
  &s_-:\mathcal{O}_\nu^-\rightarrow K_-
\end{align*}
such that
\[
\eta_+=\text{Ad}^*_{\left[s_+\left(\eta_+\right)\right]^{-1}}\mu,\qquad\eta_-=\text{Ad}^*_{s_-\left(\eta_-\right)}\nu
\]
for every $\eta_+\in\mathcal{O}_\mu^+,\eta_-\in\mathcal{O}_\nu^-$ in a suitable open set. Let us indicate by $Ts_+:T\mathcal{O}_\mu^+\rightarrow\kf_+,Ts_-:T\mathcal{O}_\nu^-\rightarrow\kf_-$ the trivialized differential maps of these sections, i.e.
\begin{align*}
  Ts_+&:u_{\eta_+}\mapsto TL_{\left[s_+\left(\eta_+\right)\right]^{-1}}T_{\eta_+}s_+\left(u_{\eta_+}\right)\\
  Ts_-&:v_{\eta_-}\mapsto TR_{\left[s_-\left(\eta_-\right)\right]^{-1}}T_{\eta_-}s_-\left(v_{\eta_-}\right)
\end{align*}

Then
\begin{multline*}
  \Phi\left(g',\zeta',\eta_+,u_{\eta_+},\eta_-,v_{\eta_-},\widetilde{\alpha},\widetilde{\alpha}_1,\widetilde{\beta},\widetilde{\beta}_1\right)=\\
  =\Bigg(s_+\left(\eta_+\right)g's_-\left(\eta_-\right),\text{Ad}_{s_+\left(\eta_+\right)}\zeta'+\text{Ad}_{s_+\left(\eta_+\right)}Ts_+\left(u_{\eta_+}\right)-\\
  -\text{Ad}_{s_+\left(\eta_+\right)g'}Ts_-\left(v_{\eta_-}\right),\eta_+,u_{\eta_+},\eta_-,v_{\eta_-},\text{Ad}_{s_+\left(\eta_+\right)}\left(\widetilde{\alpha}-Ts_+\left(u_{\eta_+}\right)\right),\alpha_1',\\
  \text{Ad}_{\left[s_-\left(\eta_-\right)\right]^{-1}}\left(\widetilde{\beta}-Ts_-\left(v_{\eta_-}\right)\right),\beta_1'\Bigg)
\end{multline*}
where $\alpha_1',\beta_1'$ are chosen in order to ensure that $\Phi$ is a contact map. Then we have that
\begin{multline}\label{eq:PullBackFeher}
  \left(\left(\text{id}\times\Phi\right)^*L_{\left(\mu,\nu\right)}\right)\left(g',\zeta',\eta_+,u_{\eta_+},\eta_-,v_{\eta_-},\widetilde{\alpha},\widetilde{\alpha}_1,\widetilde{\beta},\widetilde{\beta}_1\right)=\\
  =L_F\left(g',\zeta',\eta_+,\eta_-,\widetilde{\alpha},\widetilde{\beta}\right)+\left<\mu,Ts_+\left(u_{\eta_+}\right)\right>+\left<\nu,Ts_-\left(v_{\eta_-}\right)\right>.
\end{multline}

Now let $\widetilde{\mu}_L\in\Omega^1\left(K_+\right),\widetilde{\nu}_R\in\Omega^1\left(K_-\right)$ be the left- and right-invariant $1$-forms respectively, such that
\[
\widetilde{\mu}_L\left(e\right)=\mu,\qquad\widetilde{\nu}_R\left(e\right)=\nu;
\]
these forms can be pulled back along sections $s_\pm$, giving us $1$-forms
\[
\mu_+^s:=s_+^*\widetilde{\mu}_L,\qquad\nu_-^s:=s_-^*\widetilde{\nu}_R.
\]

These forms, in turn, induced the contact forms
\[
\Theta_\mu\in\Omega^1\left(\mR\times T\mathcal{O}_\mu^+\right),\Theta_\nu\in\Omega^1\left(\mR\times T\mathcal{O}_\nu^-\right)
\]
such that
\begin{align*}
  \left.\Theta_\mu\right|_{\left(t,u_{\eta_+}\right)}&:=\left(T_{u_{\eta_+}}\tau_{\mathcal{O}_\mu^+}\right)^*{\mu}_+^s-{\mu}_+^s\left(u_{\eta_+}\right)dt\\
  \left.\Theta_\nu\right|_{\left(t,v_{\eta_-}\right)}&:=\left(T_{v_{\eta_-}}\tau_{\mathcal{O}_\nu^-}\right)^*{\nu}_-^s-{\nu}_-^s\left(v_{\eta_-}\right)dt.
\end{align*}

Using these definitions, Equation \eqref{eq:PullBackFeher} and that
\[
\mu_+^s\left(u_{\eta_+}\right)=\left<\mu,Ts_+\left(u_{\eta_+}\right)\right>,\qquad\nu_-^s\left(v_{\eta_-}\right)=\left<\nu,Ts_-\left(v_{\eta_-}\right)\right>.
\]
it results
\[
\left(\text{id}\times\Phi\right)^*\left(L_{\left(\mu,\nu\right)}dt\right)=L_Fdt-\Theta_\mu-\Theta_\nu+\left(T_{u_{\eta_+}}\tau_{\mathcal{O}_\mu^+}\right)^*{\mu}_+^s+\left(T_{v_{\eta_-}}\tau_{\mathcal{O}_\nu^-}\right)^*{\nu}_-^s.
\]

Now, forms ${\mu}_+^s,{\nu}_-^s$ are pullback along $s_\pm$ of the contraction with $\mu\in\kf_+^*,\nu\in\kf_-^*$ of the (left and right respectively) Maurer-Cartan forms, so
\[
d{\mu}_+^s=-\frac{1}{2}\left[{\mu}_+^s\stackrel{\wedge}{,}{\mu}_+^s\right],\qquad d{\nu}_-^s=\frac{1}{2}\left[{\nu}_-^s\stackrel{\wedge}{,}{\nu}_-^s\right].
\]
Thus from Proposition \ref{Prop:DiffeomSolutions} we obtain the relation between Lagrangian system (Fehér system) $\left(N_1,L_F,0\right)$ and Routh reduction
\[
\left(M_1,L_{\left(\mu,\nu\right)},f_{\left(\mu,\nu\right)}\right),
\]
where $f_{\left(\mu,\nu\right)}$ is defined by Equation \eqref{eq:ForceTermConnection}.

\begin{thm}
  Equations for Fehér system $\left(N_1,L_F,0\right)$ and Routh reduction $$\left(M_1,L_{\left(\mu,\nu\right)},f_{\left(\mu,\nu\right)}\right)$$ coincide.
\end{thm}
\begin{proof}
  It is consequence of Proposition \ref{Prop:DiffeomSolutions} and the fact that
  \[
  f_{\left(\mu,\nu\right)}+\left(T_{u_{\eta_+}}\tau_{\mathcal{O}_\mu^+}\right)^*d{\mu}_+^s+\left(T_{v_{\eta_-}}\tau_{\mathcal{O}_\nu^-}\right)^*d{\nu}_-^s=0;
  \]
  this last equation can be proved from Equation \eqref{eq:ForceTermConnection} using the fact that $$Ts_+\left(u_{\eta_+}\right),Ts_-\left(v_{\eta_-}\right)$$ are Lie algebra elements that lift vectors $u_{\eta_+},v_{\eta_-}$.
\end{proof}

\begin{note}
  This seemingly miraculous cancellation of the force term with forms coming from a section of the principal bundle $K_+\times K_-\rightarrow\mathcal{O}_\mu^+\times\mathcal{O}_\nu^-$ is related to the fact that the chosen connection is flat. Existence of Fehér Lagrangian is local, and associated to the flatness of the connection. 
\end{note}

\section{Conclusions and outlook}

In the present article an scheme for implicit Lagrange-Routh equations was constructed using a kind of unified formalism for the unreduced Lagrangian system. This yielded to an unified formalisms for reduced systems, and invariant expressions for the associated equations of motion were obtained. These considerations served as a framework for the interpretation of some Lagrangian toy systems related to reduction of WZNW field theories.


\Urlmuskip=0mu plus 1mu\relax
\def\bibfont{\small}
\printbibliography

\end{document}